\documentclass[pra,aps,floatfix,%
 reprint,%
superscriptaddress,
longbibliography]{revtex4-2}

\usepackage{graphicx}% Include figure files
\usepackage{dcolumn}% Align table columns on decimal point
\usepackage{amsfonts}
\usepackage{amssymb}
\usepackage{amsmath}
\usepackage{amsthm}

\usepackage[colorlinks=true,linkcolor=blue,citecolor=blue,urlcolor=blue,plainpages=false,pdfpagelabels]{hyperref}
\usepackage{color}
\usepackage[dvipsnames]{xcolor}
\usepackage{mathtools}
\usepackage{bbm}

\usepackage{float}
\usepackage{array}
\usepackage{longtable}

\usepackage{verbatim}

\usepackage[utf8]{inputenc}
\usepackage[T1]{fontenc}
\usepackage{etoolbox}

\usepackage{anyfontsize}  % Can help to remove certain font size warnings
\usepackage{microtype} % Helps eliminate protrusion into the margins

\usepackage{newtxtext}
\usepackage{newtxmath}
\usepackage[normalem]{ulem}

\makeatletter
\g@addto@macro\bfseries{\boldmath}
\makeatother

\newcommand{\bra}[1]{\langle #1|}
\newcommand{\ket}[1]{|#1\rangle}

\newcommand{\ketbra}[2]{\ket{#1}\!\bra{#2}}

\newcommand{\e}{\mathrm{e}}
\newcommand{\I}{\mathrm{i}}
\newcommand{\Tr}{\mathrm{Tr}}

\theoremstyle{definition}
\newtheorem{theorem}{Theorem}
\newtheorem{lemma}{Lemma}
\newtheorem{corollary}[lemma]{Corollary}

\renewcommand{\qedsymbol}{$\blacksquare$}

\renewcommand{\qedsymbol}{\unskip\nobreak\quad\qedsymbol}
\renewcommand{\qedsymbol}{$\blacksquare$}

\DeclarePairedDelimiter{\ceil}{\lceil}{\rceil}

\begin{document}

\title{Reducing classical communication costs in multiplexed quantum repeaters\\using hardware-aware quasi-local policies}

\author{Stav Haldar}\email{stavhaldar@gmail.com}
\affiliation{Hearne Institute for Theoretical Physics, Department of Physics and Astronomy, Louisiana State University, Baton Rouge, Louisiana 70803, USA}

\author{Pratik J. Barge}
\affiliation{Hearne Institute for Theoretical Physics, Department of Physics and Astronomy, Louisiana State University, Baton Rouge, Louisiana 70803, USA}

\author{Xiang Cheng}
\affiliation{Fang Lu Mesoscopic Optics and Quantum Electronics Laboratory, Department of Electrical and Computer Engineering, University of California, Los Angeles, California 90095, USA}

\author{Kai-Chi Chang}
\affiliation{Fang Lu Mesoscopic Optics and Quantum Electronics Laboratory, Department of Electrical and Computer Engineering, University of California, Los Angeles, California 90095, USA}

\author{Brian~T.~Kirby}
\affiliation{DEVCOM Army Research Laboratory, Adelphi, MD 20783, USA}
\affiliation{Tulane University, New Orleans, LA 70118, USA}

\author{Sumeet Khatri}
\affiliation{Dahlem Center for Complex Quantum Systems, Freie Universit\"{a}t Berlin, 14195 Berlin, Germany}

\author{Chee Wei Wong}
\affiliation{Fang Lu Mesoscopic Optics and Quantum Electronics Laboratory, Department of Electrical and Computer Engineering, University of California, Los Angeles, California 90095, USA}

\author{Hwang Lee}
\affiliation{Hearne Institute for Theoretical Physics, Department of Physics and Astronomy, Louisiana State University, Baton Rouge, Louisiana 70803, USA}

\date{\today}

\begin{abstract}
    Future quantum networks will have nodes equipped with multiple quantum memories, allowing for multiplexing and entanglement distillation strategies for long-distance entanglement distribution. In this work, we focus on \textit{quasi-local} policies for multiplexed quantum repeater chains. In fully-local policies, nodes use the knowledge of only their own states, whereas more efficient global policies use knowledge of the entire network state. The classical communication costs of using this knowledge have not been explored in existing literature. We show that quasi-local policies not only obtain improved performance over local policies, but also reduce classical communication costs considerably.
    Our policies also outperform the widely studied nested purification and doubling policy in practical parameter regimes. We identify parameter regimes where distillation is useful and address the question:~``Should we distill before swapping, or vice versa?'' Finally, we propose an implementation scheme for a multiplexed repeater chain, experimentally demonstrate the key element, a high-dimensional biphoton frequency comb, and evaluate its anticipated performance using our multiplexing-based policies.
\end{abstract}

\maketitle

\section*{Introduction}
\label{sec:intro}

The quantum internet~\cite{WEH18,Dowling_book2,VanMeter_book} has the potential to revolutionize current computation, communication, and sensing technologies by enabling exchange of quantum data. Numerous significant applications have already been recognized, encompassing areas such as secure quantum communication~\cite{BB84,BBC+93,XXQ+20}, distributed quantum computation~\cite{CEHM99,barz2012demonstration}, and distributed quantum sensing~\cite{ge2018distributed,komar2014quantum, proctor2018multiparameter}, to name a few. However, enormous hardware improvements are needed before a practical quantum internet is realized~\cite{heshami2016quantum,awschalom2021interconnects}. Currently, a significant challenge lies in effectively distributing entanglement amongst network nodes using quantum repeaters, with the goal of attaining high fidelities and low waiting times, particularly over significant distances. Recent, state-of-the-art experiments have been performed for only a handful of nodes over short distances~\cite{PHB+21, hermans2022teleportation}. This challenge stems from the delicate nature of quantum systems, leading to issues like photon losses, imperfect measurements, and quantum memories with short coherence times. Since comprehensive, fault-tolerant quantum error correction has yet to be realized to overcome these obstacles, it becomes important to explore how much we can overcome these limitations through other means. This problem of designing effective repeater protocols for entanglement distribution using noisy, imperfect quantum hardware has been addressed in multiple analytical~\cite{BDC98,DBC99,memory_insensitive,BPv11,KKL15,rozpedek2018parameterregimesrepeater,DKD18,KMSD19,rozpedek2019neartermrepeaterNV,VK19,SSv19,LCE20,VGNT20,Kha21b,coopmans2021thesis,kamin2022exactrateswapasap,khatri2022networkMDP,CBE21,DT21,sadhu2023practical} and numerical~\cite{YLX+19,netsquid,sequence,wallnofer2022ReQuSim,quisp,STCMW20,BAKE20,reiss2022deep,IVSW22,haldar2023fastreliable} studies, among others; see also Refs.~\cite{SSR+11,CCT+20,CCvM20,ABC21,MPD+22,ICM+22,azuma2023repeatersRMP} for reviews.

In order to make the best use of today's noisy, imperfect quantum hardware, protocols/policies for entanglement distribution should be \textit{hardware-aware}. This means tailoring the elementary link generation and entanglement swapping policies to the hardware, instead of employing widely-studied but hardware-agnostic policies such as the doubling~\cite{BDC98,DBC99} and the ``swap-as-soon-as-possible''~\cite{coopmans2021thesis,kamin2022exactrateswapasap} policies~\footnote{Another consideration is tailoring the number and locations of repeater stations and entanglement sources to the available hardware, which are architectural questions that have have been deeply studied; see, e.g., Refs.~\cite{VLMN09,RFT+09,VanMeter_book,MLK+16,JKR+16,RYG+18}.}. Determining optimal hardware-aware policies for elementary link generation and entanglement swapping has received heightened attention only recently~\cite{WMDB19,Kha21b,Kha21,SvL21,reiss2022deep,IVSW22,haldar2023fastreliable,cacciapuoti2023entanglement}.

At the same time, for the first-generation quantum repeaters that we consider here, it is well known that the two-way end-to-end classical communication involved makes them rather slow as the end-to-end distance increases, when compared to second- and third-generation repeaters, which use quantum error-correction~\cite{MATN15,MLK+16}. Nevertheless, first-generation quantum repeaters are arguably the most amenable to implementation in the near future. Therefore, it is of interest to develop policies for elementary link generation and entanglement swapping that are not only hardware-aware, but also minimize the amount of classical communication between nodes in the network.

The amount of classical communication in a network is tied to the amount of knowledge nodes have of each other's states when making their decisions regarding elementary link generation and entanglement swapping. Indeed, the more knowledge nodes need to make use of, the more classical communication is required to attain that knowledge. The decisions that nodes make, and the amount of knowledge needed to make the decisions, is dictated by the policy being used. In particular, if a policy calls for the nodes to make use of full, global knowledge of the network---meaning that every node should know the state of every other node at all times---then classical communication over the entire span of the network is required in every time step of the protocol. On the other extreme, for fully local policies, in which nodes need to have knowledge only of their own states, only one round of end-to-end classical communication is required at some prescribed final time. In the middle lie \textit{quasi-local} policies, which we consider in this work. These policies use to their advantage the realization that nodes only need access to knowledge of the connected portion of the chain they are part of. Disconnected regions of the network can determine their actions independently. Thus, barring only a few time steps when long links are present in the network, spanning lengths on the order of the size of the network, full end-to-end classical communication is superfluous and wasteful. Previous works~\cite{CRDW19,kolar2022adaptive,inesta2023continuous} have considered quasi-local policies in the context of continuous entanglement distribution, while we consider on-demand entanglement distribution. To the best of our knowledge, a thorough analysis of end-to-end waiting times and fidelities, as a function of the amount of knowledge nodes have of each other's states, and taking classical communication costs into account, has not been conducted. Recent works~\cite{IVSW22,haldar2023fastreliable} have highlighted the importance of global knowledge in improving waiting times and fidelities, but without fully accounting for the cost of classical communication. We are therefore motivated to ask our first question:
\begin{enumerate}
    \item[(Q1)] How are waiting times and fidelities affected by the amount of classical communication being performed in a network? Are the benefits of additional network knowledge, beyond that of a fully local policy, negated by the increase in classical communication?
\end{enumerate}
In this work, we consider this question in the context of entanglement distribution in linear networks (also called ``repeater chains'') with multi-memory quantum repeaters and multiplexing~\cite{BDC98,DBC99,ladd2006hybridrepeater,memory_insensitive,RPL09,RFT+09,munro2010multiplexing,sheng2010deterministicpurification,sheng2010deterministicpurification_b,simon2010temporalmultiplexing,bratzik2013repeatersdistill,sheng2013hybridpurification,sinclair2014spectralmultiplexing,zhou2020purificationresidualentanglement,hu2021longdistancepurification,dhara2021subexponential,DT21,dhara2022multiplexed,huang2022deterministicpurification,ecker2022remotepolariziationentanglement,chakraborty2022spectralmultiplexrepeater,milligen2023timemultiplexrouting,zang2023repeaterpurification}; see Fig.~\ref{fig:schematic}. In the multiplexing approach, simultaneous elementary link generation attempts can be made, and when performing entanglement swapping~\cite{BBC+93,ZZH93}, it is possible to go beyond the trivial parallel approach and perform cross-channel entanglement swapping.

Many of the aforementioned works on multi-memory quantum repeaters with multiplexing have made use of the hardware-agnostic ``nested purification and doubling swapping'' protocol devised in the seminal works~\cite{BDC98,DBC99} on quantum repeaters. Keeping in mind our goal of developing hardware-aware policies, we therefore ask:
\begin{enumerate}
    \item[(Q2)] Is the nested purification and doubling swapping policy optimal, particularly in the practically-relevant parameter regimes of low coherence times and low elementary link success probabilities?%, and with classical communication times explicitly taken into account?
\end{enumerate}
In the multiplexing approach, there is also the possibility to perform entanglement purification/distillation~\cite{BBP96,bennett1996concentrating,BDSW96,yan2023advancespurification}. Distillation introduces additional classical communication overhead. We are therefore motivated to ponder whether the potential benefits of distillation on end-to-end fidelities are worth the additional classical communication overhead. Specifically, we ask:
\begin{enumerate}
    \item[(Q3)] In what parameter regimes is it beneficial to perform entanglement distillation? What is the best way to combine entanglement swapping with entanglement distillation---should we distill first before swapping, or vice versa?
\end{enumerate}

\begin{figure*}
    \centering
    \includegraphics[width = 0.75\textwidth]{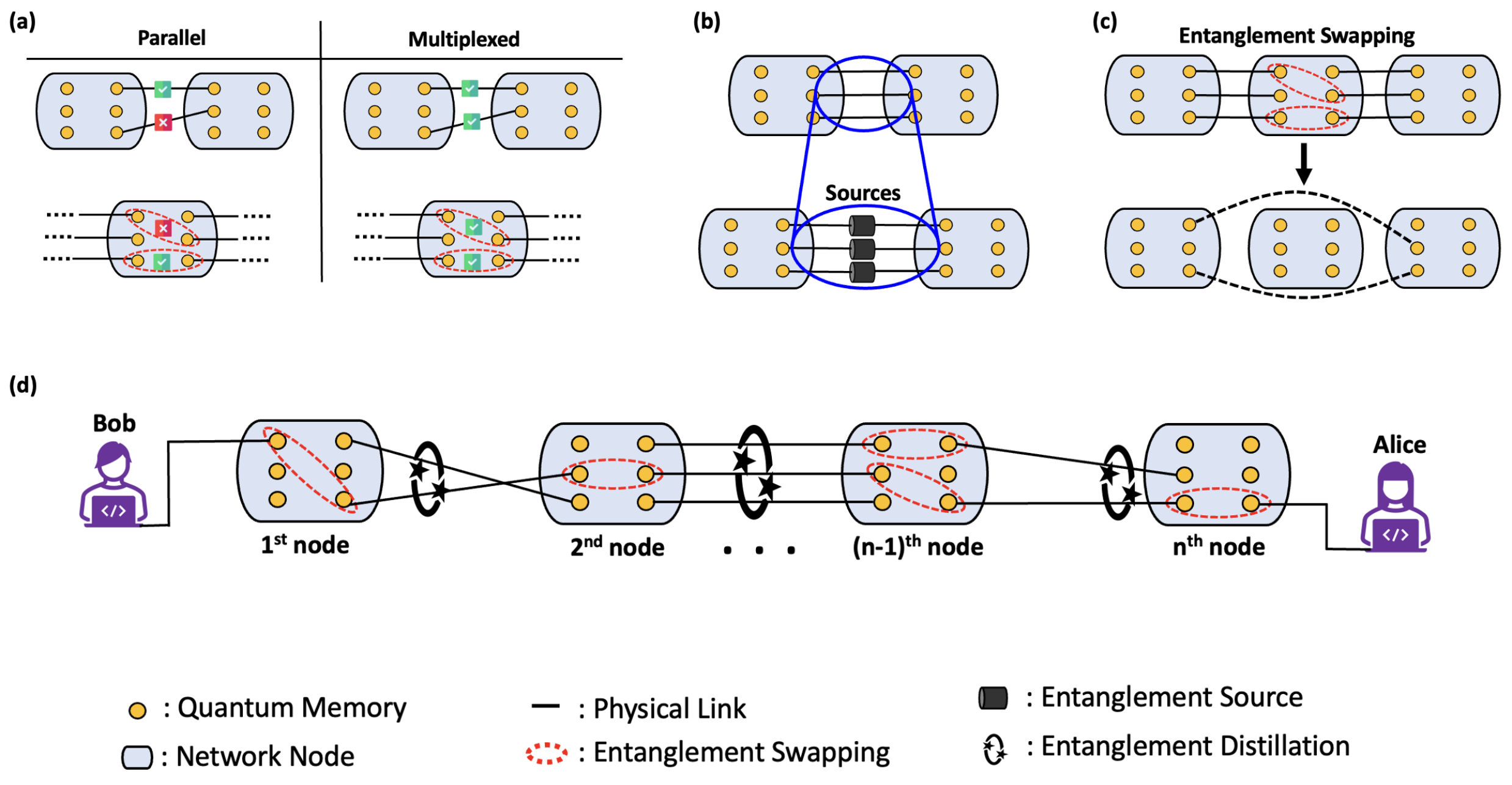}
    \caption{\textbf{Multiplexed entanglement distribution.} (a) In the multiplexing approach to entanglement distribution, it is possible to perform cross-channel elementary link generation and entanglement swapping. (b) Elementary links are created by multiple entanglement generation sources associated with every channel that connects different pairs of quantum memories. (c) Non-neighboring nodes are connected by virtual links generated by performing entanglement swapping operations between neighboring active quantum memories. (d) A linear chain of quantum repeaters with multiple quantum memories capable of establishing multiple physical links, entanglement distillation and entanglement swapping.
    }
    \label{fig:schematic}
\end{figure*}

In this paper, we address the three questions above by extending the framework from Ref.~\cite{haldar2023fastreliable} for modeling quantum repeater chains in terms of Markov decision processes (MDPs) to the case of multi-memory quantum repeater chains. We define quasi-local multiplexing-based policies, which are generalizations of the so-called ``swap-as-soon-as-possible'' (\textsc{swap-asap}) policy~\cite{coopmans2021thesis,kamin2022exactrateswapasap}, to the multiplexing setting. For near-term quantum networks with imperfections such as memories with limited coherence time, probabilistic entanglement swapping, etc., classical communication (CC) costs are considerable, but usually ignored in existing literature. We incorporate the CC costs associated with our policies, and show that these quasi-local policies outperform fully-local policies that do not require any CC. We also show that our policies can outperform the doubling policy in practically relevant parameter regimes. Next, we also add entanglement distillation to our model, and address several policy questions. Finally, we consider a proof-of-principle experimental realization of a multiplexed quantum network equipped using a high-dimensional bi-photon frequency comb (BFC) and assess its performance.

\section*{Results and discussion}
\subsection*{Overview of results}
Our main results and contributions comprise four parts.
\begin{enumerate}
    \item \textit{Quasi-local policies for multiplexed quantum repeaters.}
        
        We define our quasi-local multiplexing-based policies. These policies are generalizations of the so-called ``swap-as-soon-as-possible'' (\textsc{swap-asap}) policy~\cite{coopmans2021thesis,kamin2022exactrateswapasap} to the multiplexing setting. We refer to them as \emph{strongest neighbor} (SN) \textsc{swap-asap}, \emph{farthest neighbor} (FN) \textsc{swap-asap}, and \emph{random} \textsc{swap-asap}. These policies are distinguished by how they pair memories within nodes for the purposes of entanglement swapping. Intuitively, the SN \textsc{swap-asap} policy prioritizes the creation of high-fidelity end-to-end links, and FN \textsc{swap-asap} prioritizes lowering the waiting time for the creation of end-to-end links. These policies are also ``quasi-local'', requiring some knowledge of the network state but not full, global knowledge.
    
    \item \textit{The impact of classical communication costs.}
    
        In a realistic entanglement distribution setting with imperfect memories, considering classical communication (CC) costs is imperative in order to provide a realistic performance analysis. In \textit{Performance evaluation}, we incorporate the CC costs associated with our policies, and thereby address (Q1). In particular, we show that our multiplexing-based FN and SN \textsc{swap-asap} policies outperform fully-local policies that do not require any CC. We thus show that the benefit gained by acquiring information about the network's state is retained even when the communication costs of acquiring such information is accounted for.
        
        We also address (Q2) by showing that the FN and SN \textsc{swap-asap} policies can outperform the doubling policy in parameter regimes corresponding to high channel losses, small coherence times of quantum memories, and a small number of multiplexing channels and/or a large number of nodes.

    \item \textit{Policies with entanglement distillation.}
    
        In \textit{Distillation-based policies}, we also add entanglement distillation to our FN and SN \textsc{swap-asap} policies, and then address (Q3). We find that if the fresh elementary links have high fidelity, using entanglement distillation turns out to be detrimental. On the other hand, when fresh elementary links have low fidelity, and more generally in resource-constrained scenarios, entanglement distillation can be more useful. Nonetheless, even in such cases, we observe examples in which distillation can in fact be detrimental.
        
        In terms of the relative ordering of distillation and swapping, if links have low generation probability or low initial fidelity, then distilling first is more advantageous. When the elementary link generation probability is high, it is better to swap first and then distill.

    \item \textit{Design and analysis of real-world repeater chains.}
    
        In \textit{Proposed Experimental implementation}, we outline a proof-of-principle experimental realization of a quantum network equipped with multiplexing policies using a high-dimensional bi-photon frequency comb (BFC), and we map the relevant experimental parameters to their counterpart parameters within our modeling framework. Then we assess the performance of our multiplexing-based policies within these parameter regimes for two quantum memory platforms, diamond vacancy and rare-earth metal ions. In particular, we consider the number of repeater nodes needed to achieve a desired end-to-end waiting time and fidelity over a distance of 100~km. 
        
\end{enumerate}

\subsection*{Linear chain quantum network with multiple channels}
\label{sec:model}

Throughout this work, we consider entanglement distribution in a linear chain network of quantum repeaters. In this section, we present our theoretical model for entanglement distribution in such networks; see also Fig.~\ref{fig:schematic}.

\begin{itemize}
    \item  \textit{Nodes:} The linear chain is made up of $n$ nodes. Every node contains multiple quantum memories. Specifically, every node has $2n_{ch}$ quantum memories.
    Given the linear nature of our networks, $n_{ch}\in\{1,2,\dotsc\}$ corresponds to the number of channels connecting the nearest-neighbor nodes.  (In Fig.~\ref{fig:schematic}, $n_{ch}=3$.) Every memory has a finite coherence time of $m^{\star}\in\{0,1,\dotsc\}$ time steps, which specifies the maximum number of (discrete) times steps that qubits can be stored in the memory.

    \item \textit{Elementary and virtual links:} Entanglement sources establish physical/elementary links between nearest-neighbor quantum memories with probability $p_{\ell} \in [0,1]$. A probabilistic Bell state measurement creates a virtual link between distant quantum memories with success probability $p_{sw} \in [0,1]$. These values are fundamentally limited by the technology used to create the system, e.g., linear optics makes $p_{sw}\leq 0.5$, and $p_{\ell}$ is determined by properties of the source, loss characteristics of the sources, etc. These considerations are relevant for our proposed experimental implementation, and we also provide details on how to determine $p_{\ell}$ and $m^{\star}$ from physical properties of the system.

    We keep track of the fidelity of both elementary and virtual links by assigning every link an age. The age $m$ of a link starts from $m=m_0$, when it is first created, with $m_0$ being a measure of its initial fidelity with respect to a perfect Bell state. The age then increases in discrete steps, such that $m \in\{0,1,2,3...\}$. Once the age of the link reaches $m^{\star}$, the link is discarded.  
    Throughout this work, we consider a particular Pauli noise model for decoherence of the memories (see ``Methods''), such that the fidelity $f(m)$ of a link that has age $m$ is given by
    \begin{equation}\label{eq-fidelity_function_main}
        f(m) = \frac{1}{4}(1 + 3\e^{-\frac{m}{m^\star}}).
    \end{equation}

    The age of a link formed by a successful entanglement swapping operation is determined by adding the ages of the two corresponding links: if $m_1$ and $m_2$ are the ages of the two links, then if the entanglement swapping is successful the age of the new link is $m'=m_1+m_2$. This ``addition rule'' for the age of virtual links is a direct consequence of the Pauli noise model that we use, and a proof can be found in Ref.~\cite[Appendix~C]{haldar2023fastreliable}.

    \item \textit{Entanglement distillation:} Multiple low-fidelity links between any two nodes can be distilled with probability $p_{ds} \in [0,1]$ to form a higher-fidelity link. Throughout this work, we consider the BBPSSW protocol~\cite{BBP96}, 
    which allows for distillation of two entangled links to one. We provide a summary of the protocol in ``Methods''. Briefly, if $f_1$ and $f_2$ are the fidelities of the two links being distilled, then the success probability $P_{\text{distill}}(f_1,f_2)$ of the protocol and the resulting fidelity $F_{\text{distill}}(f_1,f_2)$ are given by
    \begin{align}
        P_{\text{distill}}(f_1,f_2)&=\frac{8}{9}f_1f_2-\frac{2}{9}(f_1+f_2)+\frac{5}{9},\label{eq-distill_p}\\
        F_{\text{distill}}(f_1,f_2)&=\frac{1-(f_1+f_2)+10f_1f_2}{5-2(f_1+f_2)+8f_1f_2}.
    \end{align}
    We show in ``Methods'' that the age after distillation with the BBPSSW protocol is given by the following formula:
    \begin{equation}
            %m'=\left\lceil m^\star\log\left(\frac{15-6(f(m_1)+f(m_2))+24f(m_1)f(m_2)}{32f(m_1)f(m_2) - 2(f(m_1)+f(m_2))-1}\right)\right\rceil,
            m'=\left\lceil m^\star\log\left(\frac{15-6(f_1+f_2)+24f_1f_2}{32f_1f_2 - 2(f_1+f_2)-1}\right)\right\rceil,\label{eq-age_distilled}
    \end{equation}
    where $f_1\equiv f(m_1)$, $f_2\equiv f(m_2)$, $m_1$ and $m_2$ are the ages of the two links being distilled, the function $m\mapsto f(m)$ is defined in Eq.~\eqref{eq-fidelity_function_main}, and $\lceil x \rceil$ is the smallest integer greater than or equal to $x$.

\end{itemize}

Entanglement distribution protocols progress in a series of discrete time steps, based on the Markov decision process model developed and used in Ref.~\cite{haldar2023fastreliable}, which we refer the reader to for further details on the model. In every time step, the following events occur.
\begin{enumerate}
    \item Check the number of active links to the right of every node (except for the right-most node). If there are inactive links, request the corresponding elementary link.
    
    \item Check the number of active links to the left and right of every node, except for the end nodes. If more than one link on either side is active, rank them based on the policy, as described in detail in \textit{Quasi-local multiplexing policies} below. If the doubling policy is being used, pair the links for entanglement swapping only if they are the same length. Perform the entanglement swapping operations in the ranked order. Entanglement swapping is attempted only if the sum of the ages of the two links is less than $m^\star$. 
    
    \item If entanglement distillation is part of the policy, attempt to distill all links between the same two nodes until either one or no link survives (also see Supplementary Note 2). If the policy calls for distilling first, then swapping, interchange steps 2 and~3.
    
    \item Increase the age of all active links by one time unit, accounting for decoherence. If classical communication (CC) overheads are to be accounted for, then add the CC time to the age of every active link; see below for details. Discard any link which has age greater than $m^\star$.
\end{enumerate}

The time steps that we consider can be converted into physical times, and also the ages can be converted to fidelity values, by incorporating the network hardware and design parameters. Examples of such a translation are given in \textit{Proposed experimental implementation} and further details are provided in ``Methods''.

\subsubsection*{Classical communication}
\label{subsec:CC_rules}
Since entanglement distribution at the elementary link level is not deterministic, it must be heralded via classical signals. This heralding sets a natural time scale for entanglement distribution through a network. The heralding time is equal to the classical communication time between two adjacent nodes. For the purposes of our simulation, this heralding time becomes the duration of one time step. Since entanglement swapping is also often probabilistic, and entanglement distillation requires two-way classical communication \cite{bennett1996concentrating} between nodes, two distinct classes of entanglement distribution approaches can be envisaged.
\begin{itemize}
\item A \emph{local} approach, where all elementary link generation and entanglement swapping attempts are made agnostic to the success or failure of prior swaps. This would then mean that all elementary links (once generated successfully) are retained up to the cutoff time. When an entanglement swap has failed, resulting in both participating links to become inactive, the surrounding nodes must act as if the link is still active until it has reached its expected cutoff. All classical communication about successful swaps is then done at the end of the protocol. See ``Methods'' for details about how the end of the protocol is established. Such an approach is fully local and therefore all policy decisions are taken by individual nodes, without any collaboration or exchange of information. The benefit is that no classical communication overheads exist except for heralding of elementary links, and therefore every time step has duration equal to the heralding time only. The cost is instead paid by having to wait for a link to reach its cutoff time, even though it might have become inactive long before that. Distillation can be only included for freshly generated elementary links in a fully local scheme, since as links get older their lengths cannot be predicted locally (since results of swaps are not known).

\item A \emph{global} approach, on the other hand, allows end-to-end classical communication amongst all nodes in every time step, such that the duration of every time step is equal to $(n-1) \times \Delta t$, where $\Delta t$ is the CC time across an elementary link. The waiting time is subsequently largely dominated by the CC overheads, putting a huge burden on the coherence time requirements of the memories. The benefit, of course, is that all nodes can now communicate with each other and take decisions collaboratively and in an adaptive fashion.
\end{itemize}

In previous works such as Refs.~\cite{IVSW22,haldar2023fastreliable}, global policies were shown to optimize the average waiting time and fidelity; however, classical communication costs were not included. Furthermore, in Ref.~\cite{haldar2023fastreliable}, a ``quasi-local'' approach to entanglement distribution policies was proposed, in which there is knowledge of the network state up to some length scales only, and not globally. In this case, an advantage over fully local policies could still be obtained. We now investigate in this work whether the advantages in waiting time and fidelity gained by global and quasi-local knowledge and collaboration can be retained when CC overheads are taken into account.

Below is a summary of how we have taken classical communication into account in our model~\footnote{Let us make a brief remark. In a full simulation, such as those which can be done via quantum network simulators such as NetSquid~\cite{netsquid}, QuISP~\cite{quisp}, SeQUeNCe~\cite{sequence}, etc., time stamps are assigned to each quantum operation and tasks are queued, such that classical communication costs are intrinsically kept track of. Our simulations focus on optimization of policies, similar in spirit to \emph{rule sets} used by many of the above mentioned simulators; see also Refs.~\cite{MDV19,Mat19}.
}.
\begin{itemize}
    \item We approximate the CC time by adding $t_{cc}$ time steps to every network evolution step, where $t_{cc}$ is equal to the length of the longest link (number of nodes) involved in any entanglement swap in that step times the CC time between two adjacent nodes of the chain. This allows enough time for the classical communication for all entanglement swaps attempted in that MDP step, and at the end of such a CC-accounted MDP step, all nodes are aware of which nodes they are connected to and the ages of those links (which can also be classically communicated). In most cases this estimate acts as an upper bound to the CC cost, since in practice nodes could perform the next required action as soon as they receive the CC tagged for or relevant to them. See ``Methods'' for details.
    
    \item CC costs associated with entanglement distillation are added based on the length of the longest links that are distilled in a time step. For local policies, distillation is only allowed for freshly prepared elementary links (if they are not perfect), and thus the CC cost is always equal to one time step, if distillation is performed, between any two adjacent nodes after heralded entanglement generation.
\end{itemize}

\begin{figure*}
    \centering
    \includegraphics[width = 0.9\textwidth]{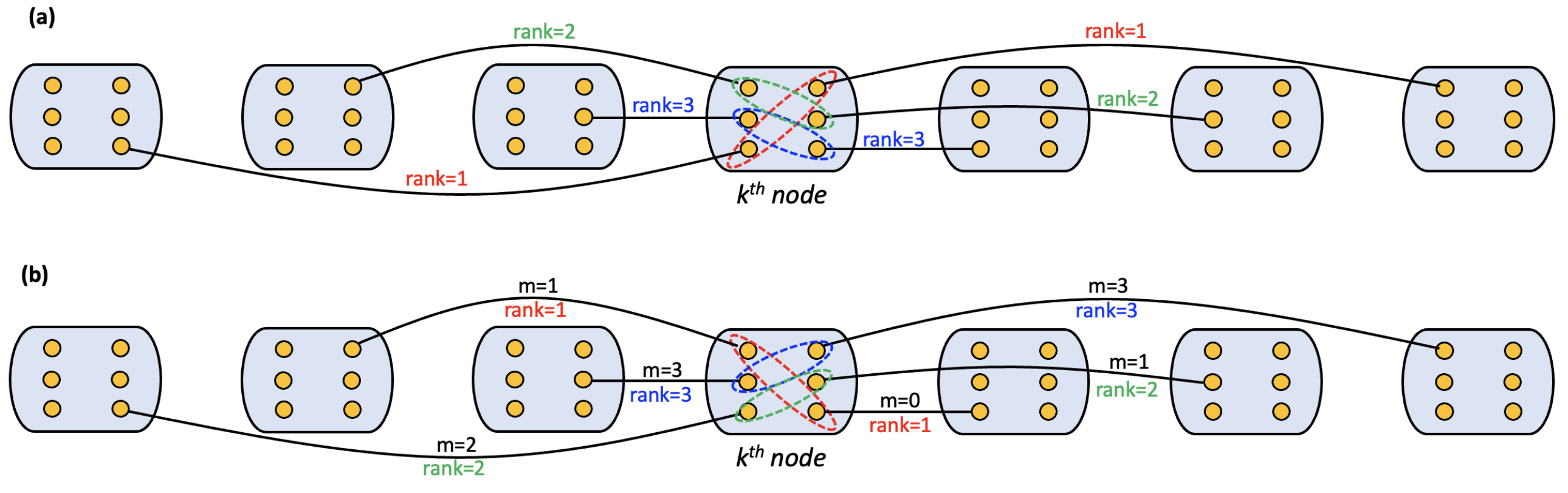}
    \caption{\textbf{Our proposed entanglement swapping policies for multiplexed repeater chains.} (a) Illustrative example of the FN \textsc{swap-asap} policy. Links on both sides of the $k^{\text{th}}$ node are ranked separately on the basis of their lengths. The longest link is given a rank of 1, while shortest is ranked 3. For entanglement swapping, we pair links of rank 1, which would result in the longest possible link. Next links of rank 2 are connected, and so on. (b) In the case of the SN \textsc{swap-asap} policy, link rankings are based on their ages.}
    \label{fig:FN}
\end{figure*}

\subsection*{Quasi-local multiplexing policies} \label{sec:policies}

We now introduce our two multiplexing policies for entanglement distribution in a linear chain quantum network. In linear quantum networks without multiplexing (i.e., $n_{ch}=1$ in Fig.~\ref{fig:schematic}), one of the simplest and best studied swapping protocols is swap-as-soon-as-possible (\textsc{swap-asap}). As the name suggests, the \textsc{swap-asap} policy dictates that entanglement swapping should be attempted as soon as two memories in a given node become active, i.e., as soon as an entanglement swapping policy becomes possible. The \textsc{swap-asap} policy has been shown to be better than fixed nesting or doubling policies for a large set of parameter regimes in linear networks consisting of a single channel between nodes~\cite{kamin2022exactrateswapasap,SSv19,SvL21}. 

The direct application of \textsc{swap-asap} in multiplexed linear chains is unspecified due to the increased number of degrees of freedom. Indeed, when multiple links are available between nodes of the network, there are many possibilities for how to perform entanglement swapping. There is also the possibility to perform entanglement distillation between nodes that share several links. We therefore define our policies based on the following considerations.

\begin{enumerate}
    \item If entanglement swapping of multiple pairs of links is possible, how should we group the pairs? We could, for example, perform entanglement swapping on link pairs that result in virtual links between \emph{farthest} nodes, or the ones that result in \emph{strongest}---highest fidelity---virtual links. Based on this consideration, we define the following two generalizations of \textsc{swap-asap} for multiplexed linear chains.
        \begin{enumerate}
            \item The \textit{Farthest Neighbor (FN) \textsc{swap-asap} policy:}~This policy prioritizes the creation of virtual links between faraway nodes. Accordingly, we rank the links based on their lengths (farthest links are ranked first), and then entanglement swapping is performed by pairing the links with the same rank. This policy is illustrated with an example in Fig.~\ref{fig:FN}(a). Intuitively, this policy prioritizes minimizing the average waiting time.
            
            \item The \textit{Strongest Neighbor (SN) \textsc{swap-asap} policy:}~This policy prioritizes the creation of virtual links with high fidelity. Accordingly, we rank the links based on their ages (links with the lowest age are ranked first), and then entanglement swapping is performed by pairing links with same rank; see Fig.~\ref{fig:FN}(b). Intuitively,  this policy aims to create the strongest or highest-fidelity links between the end nodes of the network.
        \end{enumerate}
        For comparison, we also define the \textit{random \textsc{swap-asap}} policy. In the random \textsc{swap-asap} policy, links are randomly paired to perform entanglement swapping operations, without any consideration of their length and/or age. We also define the \textit{parallel \textsc{swap-asap}} to be the non-multiplexed version of the usual \textsc{swap-asap} policy.

    \item How should we distill multiple links to one? For this, we introduce the \textit{distill-as-soon-as-possible} (\textsc{distill-asap}) policy. This is a sequential policy, which can be thought of as a hybrid of the banded and pumping policies~\cite{DBC99} (see Supplementary Note 2 for details), and very similar in spirit to the greedy policy defined in Ref.~\cite{VLMN09}. We sort the links in increasing order of their ages, and then pair them in increasing order. After one round of pairwise distillation, successfully distilled links are paired again and distillation is attempted again, until only one distilled link or no active link remains. This approach therefore uses the benefits of the banded policy, by pairing links of very similar ages, and at the same time, by realizing the restrictions introduced by decoherence and finite number of channels, by pumping fresh links into the queue of available links as soon as possible.
    
    \item What should be the preference when it comes to the order of entanglement distillation and swapping---should we distill first and then swap, or the other way around? We call these two distinct options the \textsc{distill-swap} and \textsc{swap-distill} policies. Both of these may be used in conjunction with the SN and FN entanglement swapping policies defined above, leading to four distinct policy combinations.

\end{enumerate}

The FN and SN \textsc{swap-asap} policies that we have introduced have classical communication (CC) overheads. Indeed, whenever an entanglement swapping decision is made, all the memories at the node are assumed to have knowledge of which nodes/memories they are connected to. This information is needed to decide the order of entanglement swaps based on age (SN) or length (FN) of the links. Thus, every node has some ``quasi-local'' knowledge of the network state but not the ``global'' network state. 

We demonstrate via Monte Carlo simulations of the underlying Markov decision process that our quasi-local SN and FN \textsc{swap-asap} multiplexing policies can outperform (in terms of waiting time and end-to-end fidelity) the non-multiplexed, parallel \textsc{swap-asap} policy (as we would expect), and also the random \textsc{swap-asap} policy. We present these results in Supplementary Note~1 (also see Supplementary Note~4 for some alternative quasi-local multiplexing policies). In the next section, we present results that demonstrate, more importantly, that our policies can outperform fully local multiplexing policies, as well as the well-known and widely-used nested-purification-and-doubling protocol.

Before proceeding, we remark that very recently in Ref.~\cite{KS_23_multiplex}, the idea of using fidelity-based ranks to match links for entanglement swapping (similar in spirit to our SN \textsc{swap-asap} policy) was proposed, but the focus was on second-generation repeaters with error-correction capabilities. Our focus here is on first-generation repeaters. See Refs.~\cite{MATN15,MLK+16} for the definitions of first- and second-generation quantum repeaters.

\subsection*{Performance evaluation}\label{sec:CC}

We now evaluate the performance of our policies described above with respect to the average waiting time and average age of the end-to-end link. All of our results were obtained using Monte Carlo simulations, and details about the simulations can be found in ``Methods''.

\begin{figure}
    \centering
    \includegraphics[width=0.65\columnwidth]{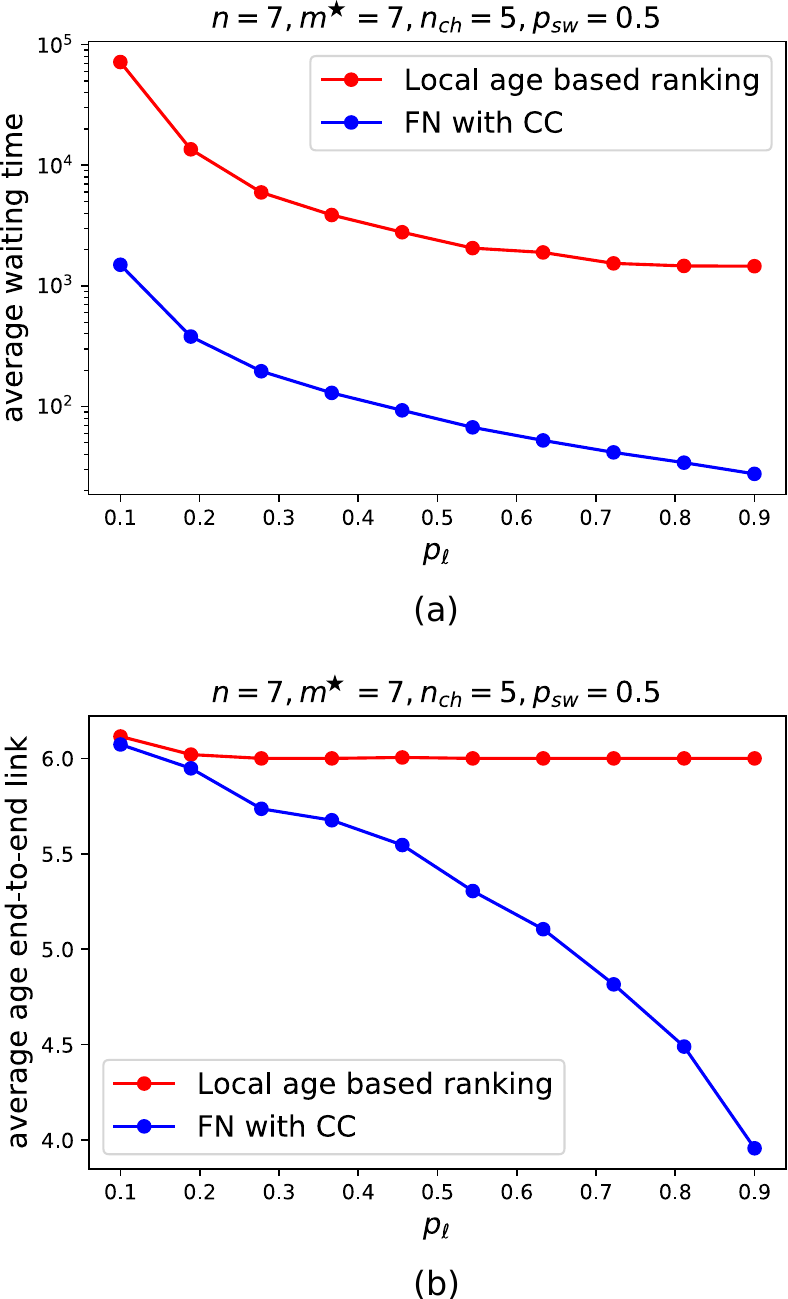}
    \caption{\textbf{Quasi-local policies outperform fully-local policies.} Average waiting time (a) and average age of an end-to-end link (b) as a function of the elementary link success probability when classical communication times are included. We compare the quasi-local FN \textsc{swap-asap} policy with the fully-local age-based \textsc{swap-asap} policy. The FN \textsc{swap-asap} policy provides smaller values for both figures of merit. Similar trends can also be seen for the SN \textsc{swap-asap} policy.}
\label{fig:CC_FOM}
\end{figure}

In order to benchmark our quasi-local policies fairly against fully-local policies, some modifications need to be made to the way the network evolves under the fully-local policies. These modifications are necessary, because when entanglement swaps are probabilistic, if links need to be restarted after a failed Bell state measurement, but before their actual memory cutoff, some CC will be needed. Therefore, if a policy has to be fully-local and free of any CC overheads, nodes must be agnostic to entanglement swapping failures. All nodes, once they have a heralded entangled link, must retain them until their anticipated memory cutoff age $m^\star$. Once this restriction is made, the only available information to nodes that can be used to decide the rank of states is the perceived or local ages of its memories. In other words, every link in the network now has two ages: one is its real age, as determined by probabilistic swaps in the evolution ``in real-time''; the second is its \emph{perceived} age, which is the age that the nodes that hold the link assume without any knowledge of the outcomes of entanglement swaps. When the perceived ages are used to make decisions in the network, there is no CC overhead during the evolution. We call such a policy \emph{local age-based} (ranking) multiplexed \textsc{swap-asap} policy. For such a policy, the CC cost is given by only one round of end-to-end communication time, in the last time step.

From Fig.~\ref{fig:CC_FOM}, it is clear that substantial advantage for both the waiting time (around two orders of magnitude reduction) and fidelity is obtained by using our policies compared to the local age-based \textsc{swap-asap} policy, even when CC times are taken into consideration.

\begin{figure}
    \centering
    \includegraphics[width=0.65\columnwidth]{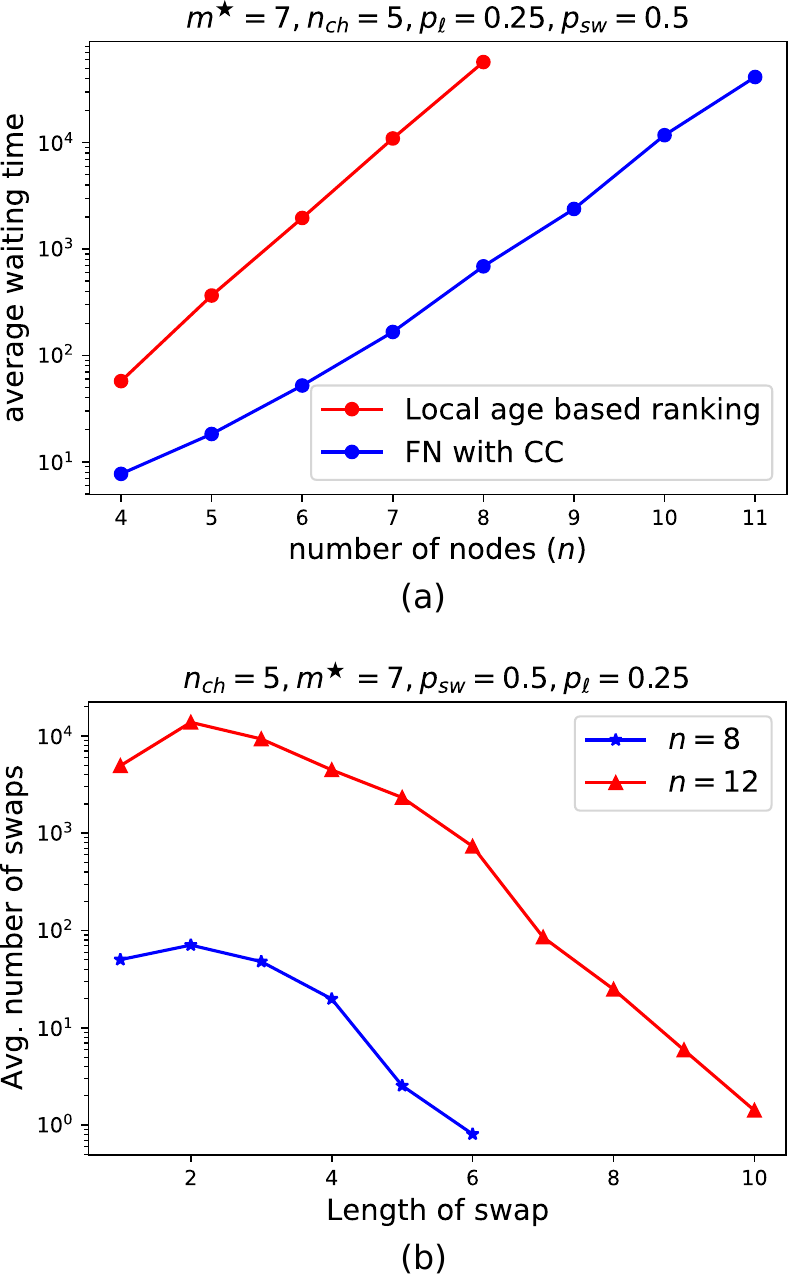}
    \caption{\textbf{Advantage of quasi-local policies is scalable and grows with increasing number of nodes.} (a) Average waiting time as a function of the number of nodes when classical communication times are included. We compare the quasi-local FN \textsc{swap-asap} policy with the fully-local age-based \textsc{swap-asap} policy, and we find that the FN \textsc{swap-asap} policy provides smaller waiting times. More surprisingly, the advantage grows with an increasing number of nodes. This provides evidence of the scalability of the advantage gained by using knowledge of the network state. (b) Number distribution of lengths of entanglement swaps in the FN \textsc{swap-asap} policy. We observe an exponential suppression of entanglement swaps of the order of the length of the chain, which is an important reason for the retained advantage of this policy, even with classical communication included.}
\label{fig:CC_nodes_FOM}
\end{figure}

A more surprising observation is the fact that even as the number of nodes increases, the advantage of FN \textsc{swap-asap} over the local age-based ranking policy remains strong---in fact, the advantage increases with increasing number of nodes. In Fig.~\ref{fig:CC_nodes_FOM}(a), we show the average waiting times for the local age-based ranking and the FN \textsc{swap-asap} policies as a function of the number of nodes. Intuitively, we might have guessed that as the number of nodes increases, so too would the CC cost, because longer and longer swaps would be required to distribute entanglement between the end points of the chain, making quasi-local policies worse off in terms of average waiting times. It is therefore surprising that not only a large improvement using the FN and SN \textsc{swap-asap} policies is obtained for long chains, but further that the improvement increases with the growing scale (number of nodes) of the network.

Nonetheless, we can explain this non-intuitive result as follows. First, we remind the reader that although fully-local policies have the advantage of no CC overheads, they suffer in terms of average waiting time due to the agnostic entanglement swapping. Indeed, waiting to restart links that are already dead is detrimental to the average waiting time. Second, because we use quasi-local policies, the CC costs in a given time step are upper-bounded by the length of the longest entanglement swapping operation, rather than the length of the entire repeater chain. (Recall that we define the length of an entanglement swapping operation as the length, in terms of number of nodes, of the longer of two links involved in the swap.) In a typical time-evolution of a network, from completely disconnected to an end-to-end connected link, for the FN \textsc{swap-asap} policy the number of entanglement swapping operations of length equal to the number of nodes is roughly $O(1)$, i.e., roughly constant with respect to the number of nodes. This is seen in Fig.~\ref{fig:CC_nodes_FOM}(b). The average number distribution of the length of entanglement swaps shows an exponentially suppressed tail. % (linear fall on the logarithmic scale).
Therefore, there are very few time steps in which the CC cost is given by the entire length of the repeater chain. This justifies and supports the use of quasi-local policies over global policies for large networks, because in the latter, the CC cost is equal for each time step and given by the entire length of the repeater chain, while for quasi-local policies the CC cost can be considerably less.

\subsubsection*{Comparison with doubling}
\label{sec:comparison_CC}

Let us now compare our policies with those in existing literature. In particular, we compare our policies with one widely studied in previous works, namely, the \emph{doubling} (or nesting) policy; see, e.g., Refs.~\cite{BDC98,DBC99,RPL09,RFT+09,bratzik2013repeatersdistill,VLMN09}. In this policy, which only applies to repeater chains with $2^N$ elementary links, the lengths of links are always doubled. As an example, consider the following situation. Suppose two adjacent elementary links are formed, an entanglement swap is performed successfully, and a link of length two is obtained. Now, in contrast to the \textsc{swap-asap} policy, this length-two link cannot be extended by swapping it with links of arbitrary lengths---links of length two must be swapped with links of length two only---which means that we must wait for another link of length two adjacent to the original link to be created before an entanglement swap can be attempted. In the meantime, the first produced virtual link keeps aging. We can thus see that this doubling policy is more restricted compared to \textsc{swap-asap}. Accordingly, it was shown in Ref.~\cite{SvL21} that for low elementary link success probabilities and moderate swapping success probabilities (such as 50\%, as in the case of linear optics), the \textsc{swap-asap} policy outperforms the doubling policy, although infinite coherence times for memories was assumed in that work.

\begin{figure}
    \centering
    \includegraphics[width=0.65\columnwidth]{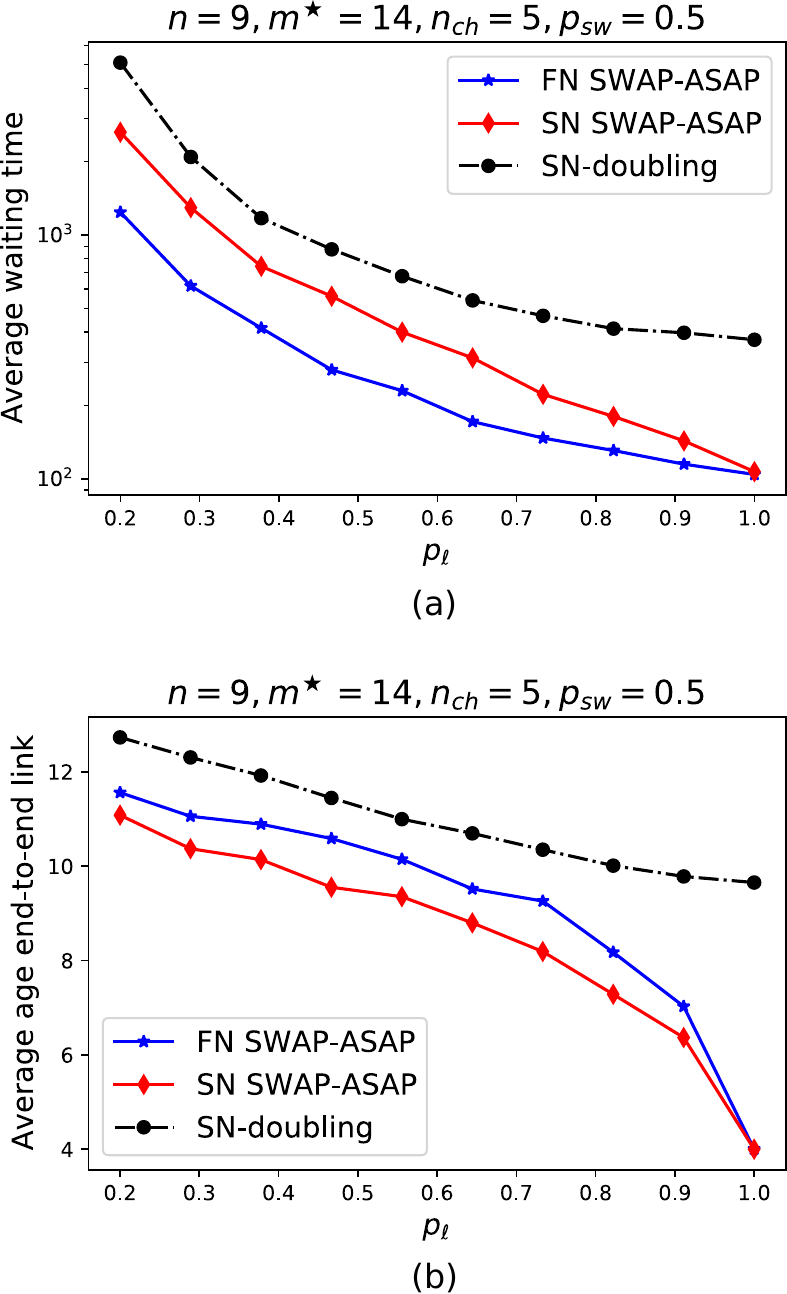}
    \caption{\textbf{Strongest Neighbor and Farthest Neighbor \textsc{swap-asap} policies outperform the doubling policy.} Comparison of (a) average waiting times and (b) average age of the end-to-end link for the FN and SN \textsc{swap-asap} policies and the SN \textsc{doubling} policy. FN \textsc{swap-asap} outperforms both SN \textsc{swap-asap} and SN \textsc{doubling} in terms of average waiting time, and SN \textsc{swap-asap} performs the best in terms of average age of the end-to-end link, for all elementary link success probabilities ($p_{\ell}$).}
    \label{fig:quasi-local_compare}
\end{figure}

Now, within the multiplexing scenario, the doubling policy can be implemented by ranking links based on age when more than one option of entanglement swapping is available. Thus, we call this policy \emph{Strongest-Neighbor doubling}, or SN \textsc{doubling}. This doubling policy is also quasi-local, in the same sense as SN and FN \textsc{swap-asap}, because a node only needs to know the length of its links, i.e., it needs to know which nodes it is connected to.

Our aim now is to see if, within the multiplexing setting, and with finite coherence times for the quantum memories, our FN and SN \textsc{swap-asap} policies can outperform the SN \textsc{doubling} policy. Our results are shown in Fig.~\ref{fig:quasi-local_compare}. We see that FN \textsc{swap-asap} outperforms both SN \textsc{swap-asap} and SN \textsc{doubling} in terms of average waiting time. At the same time, SN \textsc{swap-asap} performs better than SN \textsc{doubling}, thus further strengthening our intuition that \textsc{swap-asap} policies outperform \textsc{doubling} policies in practical resource-constrained scenarios. We also see that SN \textsc{swap-asap} performs the best in terms of average age of the end-to-end link, for all elementary link success probabilities $p_{\ell}$. This is expected, because SN \textsc{swap-asap} prioritizes the creation of younger links.

\begin{figure*}
    \centering
    \includegraphics[width=0.85\textwidth]{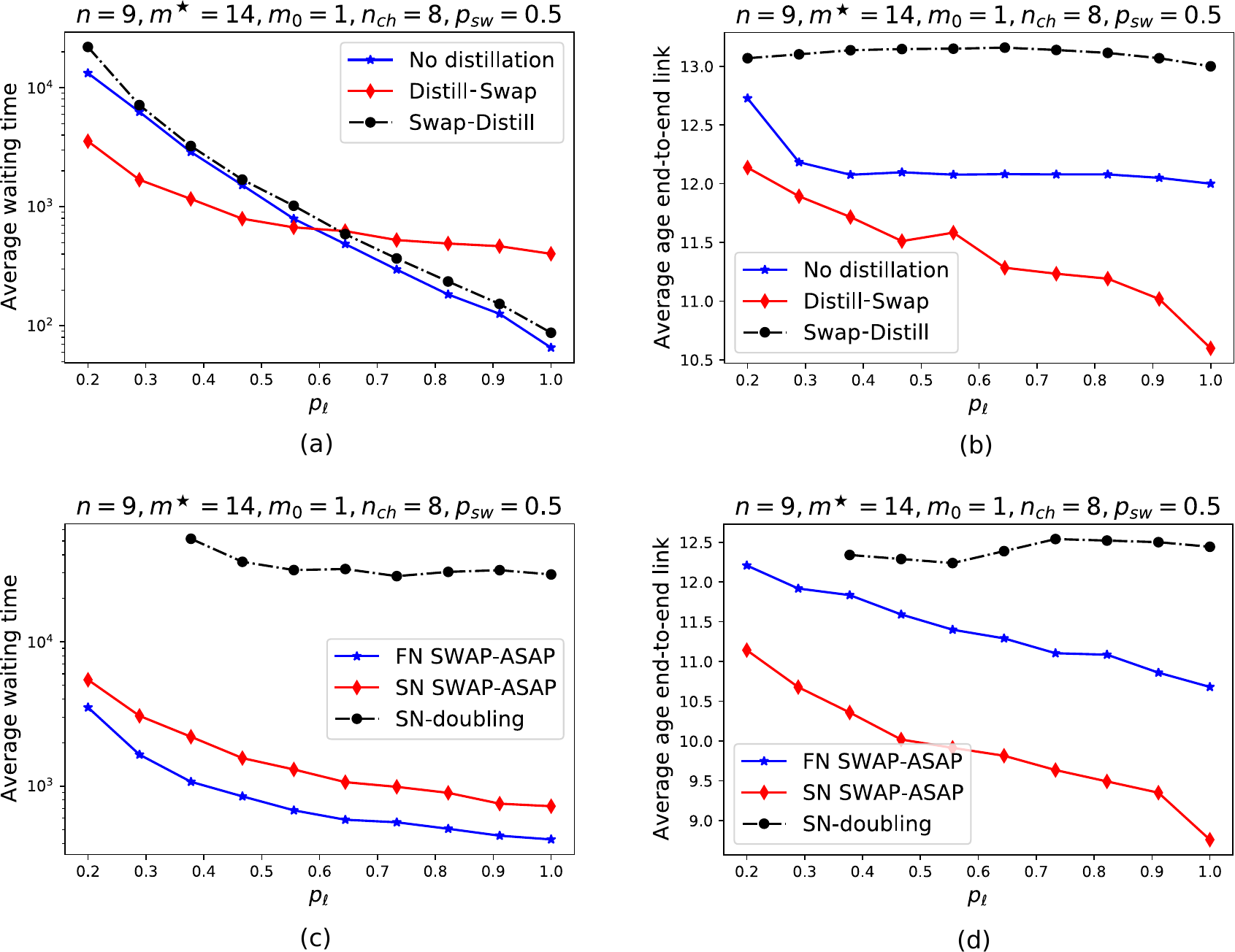}
    \caption{\textbf{Performance of various quasi-local policies with distillation.} Comparison of average waiting times (a) and average age of the end-to-end link (b) for the FN \textsc{swap-asap} policy with and without distillation. For the average waiting time, \textsc{distill-swap} outperforms the policy without distillation for low elementary link success probabilities, but it is not useful to distill when the success probability is high. At the same time, we see that it is never useful to \textsc{swap-distill}, both in terms of the average waiting time and the age of the end-to-end link. The CC costs of distillation are much higher when long virtual links are distilled, which provides intuition for the inefficacy of the \textsc{swap-distill} policy. Average waiting times (c) and average age of end-to-end link (d) for the \textsc{distill-swap} versions of FN \textsc{swap-asap}, SN \textsc{swap-asap} and the SN \textsc{doubling} policies when classical communication costs are included. FN \textsc{swap-asap} is the best policy for reducing the average waiting time and SN \textsc{swap-asap} for average age of the youngest end-to-end link.
    }
\label{fig:distill_CC}
\end{figure*}

\subsubsection*{Distillation-based policies}
Now we present results for distillation-based policies when CC overheads are included. As mentioned earlier in \textit{Classical communication}, the CC cost of distillation for any quasi-local policy is given by the length of longest links that are used for distillation. For local policies, distillation is only allowed at the elementary link level, and only when the links are freshly prepared. We continue using the \textsc{distill-asap} policy described in \textit{Quasi-local multiplexing policies}. The CC overheads make the already probabilistic BBPSSW protocol even more costly. Furthermore, the CC costs of distilling longer (virtual) links is higher than distilling elementary links, and thus the question of whether to distill or not and whether to swap first or to distill first depends crucially on the CC overheads. We show the average waiting times for FN \textsc{swap-asap} policy without distillation and both distillation-based policies (\textsc{swap-distill} and \textsc{distill-swap}) in Fig.~\ref{fig:distill_CC}(a,b). 

\begin{comment}
\begin{figure}
    \centering
    %\includegraphics[width = 0.8\columnwidth]{figures/Doubling_comparison_CC_distill_new.pdf}
    %\includegraphics[width = 0.8\columnwidth]{figures/Doubling_comparison_CC_distill_fid_new.pdf}
    \includegraphics[width=0.65\columnwidth]{fig7.pdf}
    \caption{\textbf{[figure title]} Comparison of (a) average waiting times and (b) average age of end-to-end link, \textsc{distill-swap} versions of FN \textsc{swap-asap}, SN \textsc{swap-asap} and the SN \textsc{doubling} policies when classical communication costs are included. FN \textsc{swap-asap} is the best policy for reducing the average waiting time and SN \textsc{swap-asap} for average age of the youngest end-to-end link.}
\label{fig:distill_quasi_local_compare_CC}
\end{figure}
\end{comment}

Finally, in Fig.~\ref{fig:distill_CC}(c,d), we again compare the different distillation-based versions of quasi-local policies amongst themselves, viz., FN \textsc{swap-asap}, SN \textsc{swap-asap} and the SN \textsc{doubling} policies. We see that, with distillation (\textsc{distill-swap} version) and the relevant CC costs added, the trends remain similar to Fig.~\ref{fig:quasi-local_compare}, i.e., FN \textsc{swap-asap} is still the best policy for reducing the average waiting time and SN \textsc{swap-asap} for average age of the youngest end-to-end link.

The results of the numerical simulations presented in this section show that including entanglement distillation in multiplexing policies introduces many new considerations, such as the fact that, in some parameter regimes, it is better to distill first and in others the other way around. In the next section, we take a deeper dive into these factors surrounding the use of entanglement distillation.

\subsection*{Unraveling the role of entanglement distillation policies for multiplexing}
\label{sec:DS_vs_SD_simulations}

The BBPSSW entanglement distillation protocol considered in this work non-deterministically converts two links with low fidelity into one link with a higher fidelity. This immediately raises the following question: ``Is distillation useful for all network parameter regimes?'' Due to its non-deterministic nature, in certain parameter regimes, the slight increase in fidelity obtained by distillation might not be worth the potential loss of links, in case it fails. Furthermore, even without the consideration of CC costs, there exists an underlying asymmetry between the way entanglement swapping and distillation protocols work. In the former, the success probability is independent of the ages of the links, and thus it does not discriminate between longer (generally older) and shorter links, whereas the latter's success crucially depends on the ages of the input links. This asymmetry leads to the the following question: ``Should we distill and then swap or the other way around?'' This consideration leads to two distinct policies, namely \textsc{distill-swap} and \textsc{swap-distill}, with different behaviors in different parameter regimes, as we show in Fig.~\ref{fig:distill_CC}. The aim of this section is to pursue the mechanisms behind these different behaviors and unravel them layer by layer. We make one simplifying assumption in this section, which is that classical communication costs associated with entanglement swapping and distillation are ignored. This is motivated by two reasons: firstly, for short chains, CC costs are usually ignored in existing literature on entanglement distribution policies~\cite{IVSW22, RL22, KS_switch1, KS_switchGKP,kamin2022exactrateswapasap}, thus this assumption allows us to fairly compare the performance of our policies with those in the existing literature. Secondly, ignoring CC simplifies our model and makes some analyses more transparent. At the same time, the conclusions we come to in this section hold qualitatively even when CC is accounted for.

We begin by noting that entanglement distillation is useful and worthwhile only when it increases the fidelity of the distilled link as compared to the fidelities of the original two links. In particular, the fidelity after distillation should be greater than the highest-fidelity link of the two being distilled, i.e., $F_{\text{distill}}(f_1,f_2)>\max\{f_1,f_2\}$. In terms of ages, this translates to the following condition for entanglement distillation with the BBPSSW protocol to be useful: $m'<\min\{m_1,m_2\}$, where $m'$ is the age after distillation, as given by Eq.~\eqref{eq-age_distilled}.
If $m'\geq\min\{m_1,m_2\}$, then we might as well use the youngest of the two links rather than attempt entanglement distillation and risk its failure. In particular, if one of the ages is 0, corresponding to a perfect Bell pair, then distillation should not be performed.

Another important consideration is that both links being distilled should be entangled. If one of the links is not entangled, then it is better to simply discard the unentangled link and keep the entangled one. We elaborate on this point in ``Methods'' when discussing the BBPSSW distillation protocol. With this in mind, let us note that under the Pauli noise model that we consider, the fidelity according to Eq.~\eqref{eq-fidelity_function_main} is bounded from below by $0.25$ (in the limit $m\to\infty$), and it is equal to approximately $0.5259$ when $m=m^{\star}$. The entanglement threshold for the states that we consider is $0.5$, meaning that the link is entangled if and only if its fidelity strictly exceeds $0.5$, and thus the age at which the link is no longer entangled is $m=\ceil{m^{\star}\log(3)}\approx \ceil{1.09m^{\star}}$ (also see ``Methods''). This means that, depending on the value of $m^{\star}$, it can be useful to distill links that are older than $m^{\star}$ time steps. However, the probability of successful distillation falls with increasing ages of the links, and hence using a very old link in conjunction with a very young link increases the likelihood of wasting the high-quality young link. Thus, while distillation might be useful in this scenario, it is not necessarily worthwhile. Also, for the moderate values of $m^\star$ considered in our simulations, we choose to discard links at $m^\star$ time steps, instead of the strict limit of $\ceil{m^{\star}\log(3)}$, even when distillation is considered.

We now present simulation results for a seven-node chain with seven channels. We also provide analytical results for some simpler cases in Supplementary Note 3. A somewhat large value of the cutoff is chosen ($m^\star = 24$) to show some important scaling behaviour with increasing age $m_0$ of the fresh (newly created) links. Fig.~\ref{fig:sim_distill_m0} shows the average waiting time and average age of end-to-end links as a function of the ages of the elementary links $m_0$ when they are fresh, i.e., newly created. For both \textsc{swap-distill} and \textsc{distill-swap}, we have fixed the entanglement swapping policy to FN \textsc{swap-asap}.

\begin{figure}
    \centering
    \includegraphics[width=0.65\columnwidth]{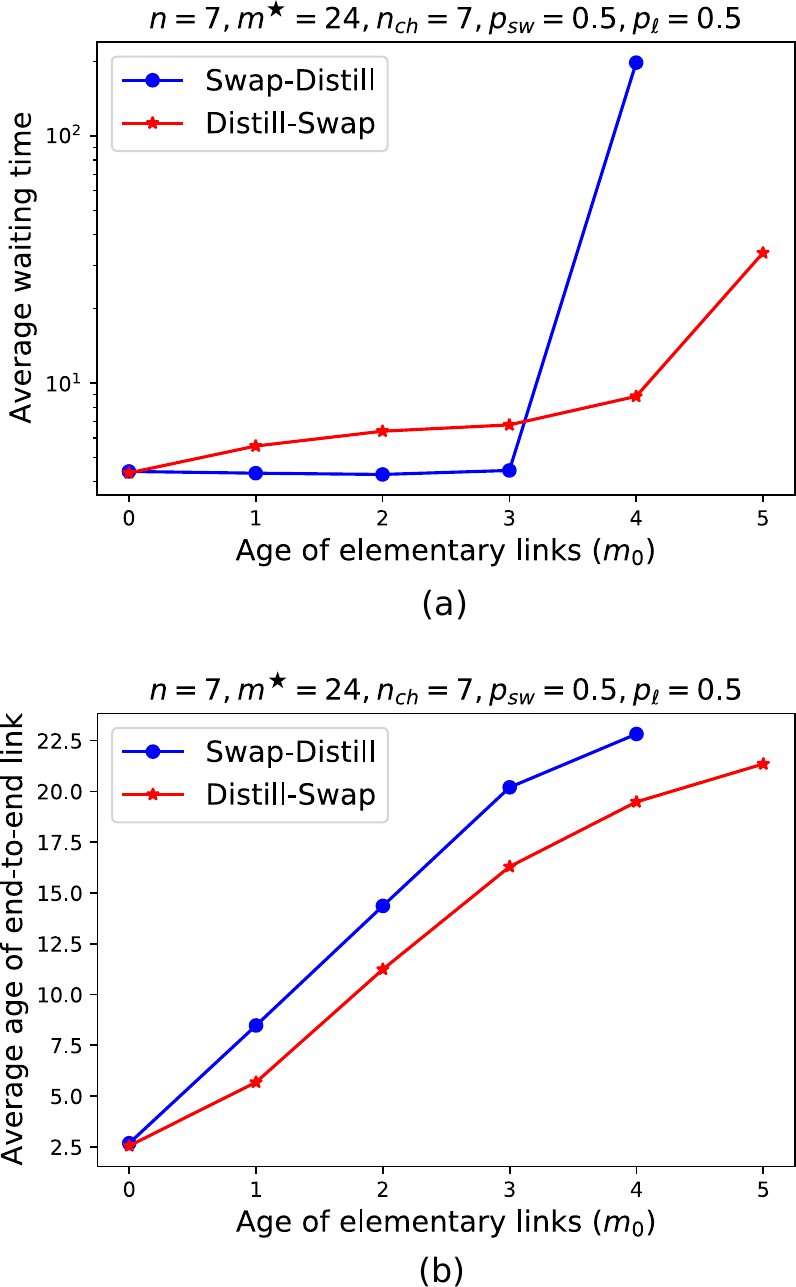}
    \caption{\textbf{Should we swap-then-distill or distill-then-swap?} Average waiting time (a) and average age of end-to-end links (b) as a function of the age $m_0$ of newly-created elementary links. The entanglement swapping policy is fixed to FN \textsc{swap-asap}. It is clear that \textsc{swap-distill} is the policy of choice in terms of reducing the waiting time when fresh elementary links have high fidelity (low age); otherwise, \textsc{distill-swap} performs better.}
\label{fig:sim_distill_m0}
\end{figure}

It is clear that \textsc{swap-distill} is the policy of choice in terms of reducing the waiting time when fresh elementary links have high fidelity (low age); otherwise, \textsc{distill-swap} performs better. At the same time, \textsc{distill-swap} always performs better in terms of the age of the end-to-end link. We also verify that the waiting time rapidly increases beyond a threshold value of $m_0$ for both \textsc{swap-distill} and \textsc{distill-swap}, the threshold being smaller for \textsc{swap-distill}. Beyond the $m_0$ values reported in Fig.~\ref{fig:sim_distill_m0}, end-to-end entanglement generation time increases so much in the simulations that we are unable to find any reasonable estimate of an average value. Thus, for all practical purposes, the waiting time is essentially infinite. This is also indicated by the average ages approaching $m^\star$.

Some estimates of the $m_0$-thresholds up to which the two distillation policies are effective can be made using some rather simple arguments.
If the fresh elementary links have age $m_0$,
%as seen in Fig.~\ref{fig:swdl_vs_dlsw_pumping},
the age of the distilled link saturates to a minimum value $m^d_{\min}(m_0)$ which is a function of $m_0$ (see, e.g., the pumping policy in Supplementary Note 2). In order to get an end-to-end link via entanglement swapping of such distilled links, the following condition must be satisfied: 
\begin{equation}
    (n-1) \cdot m^d_{\min}(m_0) \leq m^\star.
\end{equation} 
For the configuration in Fig.~\ref{fig:sim_distill_m0}, the above condition translates to $m^d_{\min} \lessapprox  4$, which in terms of $m_0$ is $m_0 \leq 5$, and this is the value around which we see that the average waiting time diverges. 

In the case of \textsc{swap-distill}, a typical trajectory to get an end-to-end link would be to first perform $(n-2)$ swaps to get end-to-end links of age $(n-1)m_0$, which are then distilled to links of age $m^d_{\min}((n-1) \cdot m_0)$. The threshold condition thus takes the following form: 
\begin{equation}
    m^d_{\min}((n-1) \cdot m_0)) \leq m^\star,
\end{equation}
but since distillation is ineffective when both links involved are above the cutoff age, this condition is satisfied whenever
\begin{equation}
     (n-1) \cdot m_0 \leq m^\star .
\end{equation}
This also explains why the threshold is lower for \textsc{swap-distill}, since $m^d_{\min}(m_0) \leq m_0$. Again, for the configuration in Fig.~\ref{fig:sim_distill_m0}, this translates to $m_0 \lessapprox 4$, which agrees well with the simulation results. Furthermore, the threshold in the case of the \textsc{swap-distill} policy is identical to a policy without distillation. 

\begin{figure*}
    \centering
    \includegraphics[width=0.85\textwidth]{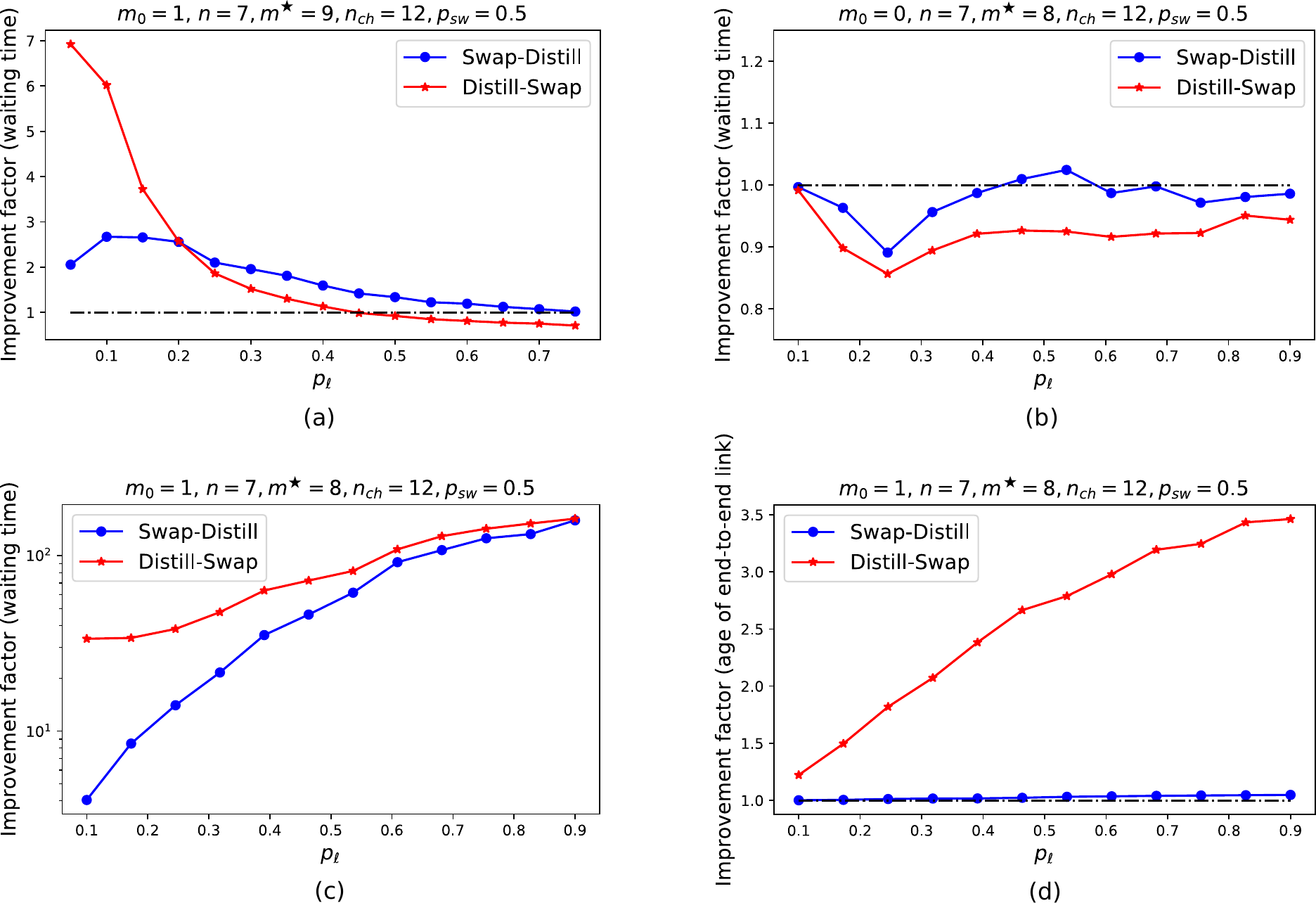}
    \caption{\textbf{Is distillation useful?} Average waiting time improvement factor for distillation with the \textsc{swap-distill} and \textsc{distill-swap} policies (with FN \textsc{swap-asap} entanglement swapping). (a) When $p_{\ell}$ is low, \textsc{distill-swap} performs better, and when $p_\ell$ is high \textsc{swap-distill} performs better. These are general trends and in concurrence with (c), where for low $p_\ell$ \textsc{distill-swap} is better but as $p_\ell$ increases \textsc{swap-distill} catches up. Furthermore, for large $p_\ell$, distillation is either not very useful, and can even be detrimental. Distillation is not very useful, in fact mostly detrimental, when (b) fresh elementary links are perfect ($m_0 = 0$, $100\%$ fidelity), but when fresh elementary links are imperfect (c) ($m_0=1$, approx.~$90\%$ fidelity), a dramatic reduction in waiting time can be obtained by performing distillation. (d) Improvement factor for average age of youngest end-to-end link using distillation over non-distillation based (FN) policy. \textsc{swap-distill} provides negligible advantage but \textsc{distill-swap} provides substantial advantage.}
    \label{fig:swdl_vs_dlsw_sims}
\end{figure*}

\begin{comment}
\begin{figure}
    \centering
    \includegraphics[width = 0.8\columnwidth]{imp_pl_distillation_m1_swdl_vs_dlsw.pdf}
    \caption{Average waiting time improvement factor for distillation with the \textsc{swap-distill} and \textsc{distill-swap} policies (with FN \textsc{swap-asap} entanglement swapping). When $p_{\ell}$ is low, \textsc{distill-swap} performs better, and when $p_\ell$ is high \textsc{swap-distill} performs better. These are general trends and in concurrence with the Fig.~\ref{fig:distill_advantage}~(center), where for low $p_\ell$ \textsc{distill-swap} is better but as $p_\ell$ increases \textsc{swap-distill} catches up. Furthermore, for large $p_\ell$, distillation is either not very useful, and can even be detrimental.
    }
    \label{fig:swdl_vs_dlsw_sims}
\end{figure}
\end{comment}

We further elaborate on our answer to the question of whether to distill first or to swap first, and whether or not distillation is useful in Fig.~\ref{fig:swdl_vs_dlsw_sims}. To do this, we plot the waiting time \textit{improvement factor} for distillation, which we define to be the ratio of the average waiting time without distillation to the average waiting time with distillation. Apart from the quality of fresh links (the parameter $m_0$), as shown in Fig.~\ref{fig:sim_distill_m0}, the choice between the two distillation orderings also depends on the other network parameters. As an illustration, we show that for the parameters $(n, n_{ch}, m^\star, p_\ell)$ chosen in Fig.~\ref{fig:swdl_vs_dlsw_sims}~(a), when $p_{\ell}$ is low, \textsc{distill-swap} performs better, and when $p_\ell$ is high \textsc{swap-distill} performs better. This observation can be understood in light of the analytical results in Supplementary Note~3 and Fig.~\ref{fig:sim_distill_m0}. Indeed, we observe that for low $p_{\ell}$, links are typically older when they are ready to be swapped, and we have seen in the previous discussions in this section that \textsc{distill-swap} performs better in this case, because long entanglement swaps can only be possible once the ages of the participating links have been reasonably reduced. Also, distilling after entanglement swapping is highly improbable to be successful because the ages are going to only increase due to the swapping. On the other hand, when links have a high $p_\ell$, \textsc{swap-distill} performs better, because we observe that such links are typically younger when they are ready to be swapped, and they are much more conducive to successful swapping. Furthermore, in this case, the non-deterministic swaps become the rate determining step, and hence having more opportunities to perform entanglement swapping is more important. Distilling links first reduces the number of such opportunities. Furthermore, we see that for large $p_\ell$, distillation is either not very useful or even detrimental. Note that these trends were also qualitatively seen in Fig.~\ref{fig:distill_CC}, when CC costs were accounted for. 

\begin{comment}
\begin{figure}
    \centering
    \includegraphics[width = 0.6\columnwidth]{imp_pl_distillation_m0.pdf}
    \includegraphics[width = 0.6\columnwidth]{imp_pl_distillation_m1.pdf}
    \includegraphics[width = 0.6\columnwidth]{imp_pl_distillation_m1_age.pdf}
    \caption{Average waiting time improvement factor for distillation with the \textsc{swap-distill} and \textsc{distill-swap} policies (with FN \textsc{swap-asap} entanglement swapping).
    Distillation is not very useful, in fact mostly detrimental, when (top) fresh elementary links are perfect ($m_0 = 0$, $100\%$ fidelity), but when fresh elementary links are imperfect (center) ($m_0=1$, approx.~$90\%$ fidelity), a dramatic reduction in waiting time can be obtained by performing distillation. (Bottom) Improvement factor for average age of youngest end-to-end link using distillation over non-distillation based (FN) policy. \textsc{swap-distill} provides negligible advantage but \textsc{distill-swap} provides substantial advantage.}
\label{fig:distill_advantage}
\end{figure}
\end{comment}

The discussion above leads us to further explore the question pertaining to the usefulness of distillation-based policies. %We are yet to show whether there is any advantage in performing distillation.
Distillation is a probabilistic process, and even when it succeeds it leads to a reduction in the number of active links. Hence, it is not obvious that a policy without distillation, such SN or FN  \textsc{swap-asap} alone, cannot perform better than its distillation-based counterpart. Let us look at two cases: when fresh elementary links are perfect, i.e., $m_0 = 0$, and when they are imperfect, say $m_0 = 1$ for concreteness. In Fig.~\ref{fig:swdl_vs_dlsw_sims}(b--d) we look at a seven-node chain with $n_{ch}=12$, $p_{sw}=0.5$, and $m^\star=8$ ($m_0 = 1$ corresponds to roughly $90\%$ fidelity in this case). We choose FN \textsc{swap-asap} as our swapping policy. We show the average waiting time improvement factor for distillation-based policies, defined as the ratio of average waiting time without distillation and with distillation as a function of the elementary link success probability $p_\ell$. An improvement factor of greater than 1 indicates an advantage in using distillation and vice versa. We see in Fig.~\ref{fig:swdl_vs_dlsw_sims}(b) that distillation is mostly detrimental when freshly produced elementary links are perfect (100\% fidelity). Both \textsc{swap-distill} and \textsc{distill-swap} increase the average waiting time, for all values of the elementary link probability. Further, the disadvantage is higher in the case of \textsc{distill-swap}. This can understood by the following observation. The primary role of distillation in the case when elementary links are perfect is to reduce the age of long or older links obtained as a consequence of waiting for other links and/or via entanglement swaps. The longest links tend to be the oldest ones in the network, due to the $m_1+m_2$ age-update rule for entanglement swapping. If we distill first, then initially, since most links are perfect, distillation is never invoked and after entanglement swaps, when distillation is needed, we have to wait an extra time step before distillation is attempted. This extra time step reduces the success probability of distillation of already old links, following Eq.~\eqref{eq-distill_p}. This effect is even more magnified when CC costs are added, since swaps will now lead to even older links. This is also the reason why \textsc{swap-distill} has some positive improvement in some parameter regimes. Although not shown in Fig.~\ref{fig:swdl_vs_dlsw_sims}(b), there exists a small parameter regime of low $m^\star, p_\ell$ and $n_{ch}$ in which \textsc{swap-distill} leads to a small positive improvement (3-5\%). The choice of parameters in the figure is motivated by the intent to show the most prominent and overarching trends. On the other hand, when elementary links are imperfect, distillation is essential in increasing the scale, fidelity and throughput of the network, indicated by the significant reduction in waiting time seen in Fig.~\ref{fig:swdl_vs_dlsw_sims}(c). When the elementary link probability is low, \textsc{distill-swap} performs much better, and when it is is high, \textsc{swap-distill} catches up in terms of improvement. This is in agreement with the results shown in Fig.~\ref{fig:sim_distill_m0} and \ref{fig:swdl_vs_dlsw_sims}(a). In terms of the end-to-end link fidelity \textsc{swap-distill} provides negligible advantage but \textsc{distill-swap} provides substantial advantage as shown in Fig.~\ref{fig:swdl_vs_dlsw_sims}(d). It is noteworthy that the trends of improvement factor variation with increasing elementary link probability change between Fig.~\ref{fig:swdl_vs_dlsw_sims}(a) and Fig.~\ref{fig:swdl_vs_dlsw_sims}(c) substantially with a small change in $m^\star$ from $8$ to $9$. This happens because $m^\star = 9$ lies at the threshold near which distillation becomes beneficial, hence the high sensitivity to $m^\star$. As an aid to intuition, recall that if no distillation is used, then to establish an end-to-end link, when newly-created links have age $m_0=1$ in a chain of seven nodes, the smallest cutoff time that ensures an end-to-end entangled link is $m^\star=6$, and note that $m^\star = 8,9$ are very close to this limit.

\subsubsection*{Main messages for distillation}

In this section, we have addressed the following questions. 
\begin{itemize}
    \item \emph{To distill or not to distill?}
        
            \textit{Yes} (generally), when fresh elementary links are not perfect. In this case, distillation leads to a significant improvement in both the average waiting time and the average age of the end-to-end link. Nonetheless, the decision to distill links is not straightforward, since in extremely resource constrained scenarios, such as very low $p_\ell$ or $p_{sw}$ it might again become disadvantageous to distill (see \textit{Figures of merit and repeater chain design for existing memory platforms}), even when fresh links are imperfect.
            
            \textit{No}, when fresh elementary links are perfect. In this case, distillation offers either negligible advantage or is in fact detrimental for most parameter regimes. 
        
            More generally, the more resource constrained the scenario, i.e., in terms of network's hardware parameters, the more useful distillation is. The observations made in this section show that detailed simulations are indispensable for analysis of the role of distillation in repeater chains.

    \item \emph{To distill first or to swap first?} 

        If links have low $p_\ell$ or low initial fidelity, distilling first is more advantageous, in the opposite case it is better to swap first.
    
        More generally, in resource constrained scenarios, i.e., in terms of network's hardware parameters, it is more useful to distill first and then swap. Further, classical communication overheads only lead to quantitative changes in these trends.

\end{itemize}

%\section{Experimental Implementation} \label{sec:exp}
%\section{Experimental proposal}
\subsection*{Proposed experimental implementation} \label{sec:exp}

\begin{figure*}
    \centering
    \includegraphics[width = 0.8\textwidth]{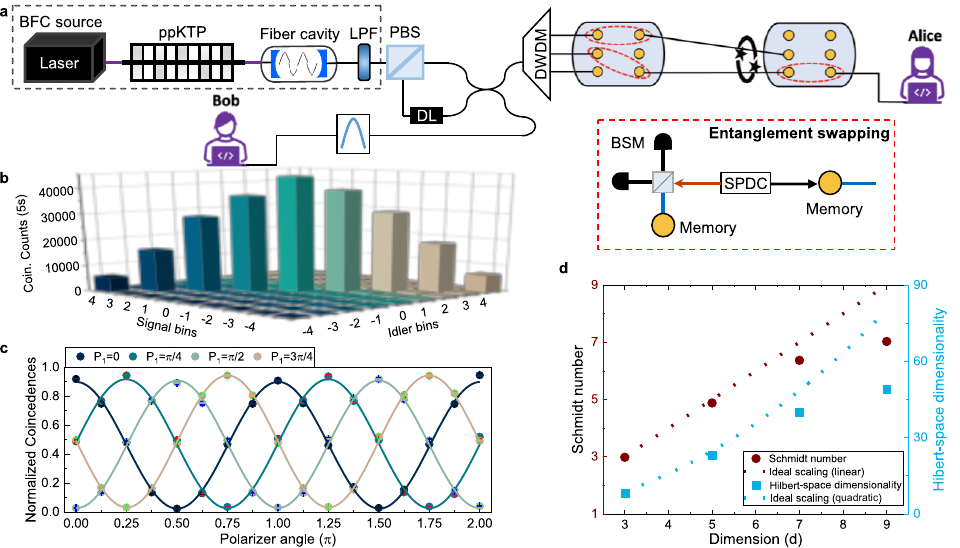}
    \caption{\textbf{Proposed experimental implementation of a multiplexed repeater chain.} (a) Proposed implementation of a multiplexed repeater chain using a high-dimensional biphoton frequency comb (BFC). (b) Measured frequency correlation matrix of the BFC for the multiplexed elementary link that we implemented. (c) Polarization entanglement measurements of the implemented BFC. (d) Schmidt number and Hilbert-space dimensionality scaling for different dimensions in the implemented~BFC.}
    \label{fig:exp_imp}
\end{figure*}

Based on the multiplexed model introduced in Fig.~\ref{fig:schematic}, in this section we propose a proof-of-principle experimental proposal of a quantum network with multiplexing using a high-dimensional biphoton frequency comb (BFC) in Fig.~\ref{fig:exp_imp}. 

As part of this proposal, we present results of an implementation of a multiplexed elementary link using a high-dimensional BFC. First, the BFC state is prepared by passing the spontaneous parametric down-conversion (SPDC) photons generated from a type-II periodically-poled KTiOPO4 (ppKTP) waveguide through a fiber Fabry–P\'{e}rot cavity (FFPC)~\cite{Chang2021, Chang2024, Xie2015}. Due to the type-II phase-matching condition, the signal and idler photons can be efficiently separated by a polarization beam-splitter (PBS), enabling deterministic BFC generation without post-selection. The signal and idler photons are then temporally overlapped at a fiber beamsplitter with orthogonal polarizations, by using a tunable delay line (DL) and polarization controllers~\cite{Chang2021}. This configuration generates polarization entangled BFC state
\begin{multline}
    \ket{\psi} = \sum_{m=-N}^{N} \ket{\omega_p/2 + m\Delta\Omega}_1 \ket{\omega_p/2 - m\Delta\Omega}_2 \\ \otimes (\ket{H}_1\ket{V}_2 + \ket{V}_1\ket{H}_2),
\end{multline}
where $2N+1$ is the number of cavity lines passed by an overall bandwidth-limiting filter (the number of frequency modes), $\Delta \Omega$ is the free spectral range of the FFPC, and $\omega_p$ is the pump frequency~\cite{Xie2015}. This scalable frequency-polarization hyperentangled BFC state~\cite{Chang2021} facilitates the multiplexing network policy by transmitting polarization qubits via different frequency channels to quantum network nodes, where entanglement swapping is performed in polarization basis. Specifically, within the context of our theoretical model, $n_{ch} = 2N + 1$. As shown in Fig.~\ref{fig:exp_imp}(a), one output channel of the beamsplitter is directed to Bob, and the other is connected to a commercial dense wavelength-division multiplexing (DWDM) module for de-multiplexing and sent to quantum network nodes, consisting of multiple quantum memories. Each frequency channel is assigned to couple to one quantum memory at the quantum network node. The entanglement swapping inside a node is assisted by an auxiliary SPDC source, where the signal photons are sent to Bell-state measurement and idler photons are sent to another memory to establish a fiber link to the next node.

Fig.~\ref{fig:exp_imp}(b) shows the measured frequency correlation matrix of a BFC, using a FFPC of 50 GHz free spectral range (FSR) and finesse of 10. We measured 9 correlated frequency-bin pairs between signal and idler photons with a type-II SPDC source of $\approx$ 250 GHz phase-matching bandwidth. Fig.~\ref{fig:exp_imp}(c) shows the measured high-quality polarization fringes for the BFC state before de-multiplexing, which shows a raw fringe visibility up to 94.1 $\pm$ 0.5 $\%$ and violates the Clauser-Horne-Shimony-Holt (CHSH) inequality by 19 standard deviations with score $S = 2.666 \pm 0.034$. This polarization state yields a state fidelity of $93.3 \pm 0.1\%$ with respect to a Bell state $\frac{1}{\sqrt{2}}(\ket{H}_1\ket{V}_2 + \ket{V}_1\ket{H}_2)$. We note that, ideally, the polarization entanglement fidelity remains the same after multiplexing. However, due to the falloff of the frequency correlation matrix, there will be deviation for the frequency-bin pairs away from the central bin~\cite{Chang2021,cheng2023highdimensional}. Such high-fidelity polarization entanglement facilitates the entanglement swapping between quantum memories. 

We perform the Schmidt mode decomposition of the measured joint spectral intensity~\cite{Chang2021}, and summarize the extracted Schmidt mode number ($K$), and the corresponding Hilbert space dimension ($K^2$). In particular, $K$ represents the effective orthogonal modes in the system. When reconstructing the density matrix, we lower bound the Schmidt mode number, which scales linearly with the dimension $d$ of the biphoton frequency comb. The Hilbert-space dimensionality can then be estimated by $K^2$. In Fig.~\ref{fig:exp_imp}(d), we plot the Schmidt mode number and the Hilbert-space dimension as a function of the dimension $d$, and we compare it with the ideal, theoretical scaling. We observe that the Schmidt mode number is $K=2.98$ with $d = 3$, and $K = 7.03$ when we increase the dimension to 9. Such BFC state provides up to seven effective frequency modes for multiplexing. The deviation from ideal scaling in Fig.~\ref{fig:exp_imp}(d) mainly comes from the falloff from the sinc function of the SPDC spectrum, which can be circumvented by using a flat-top broadband SPDC source. The Hilbert-space dimension can be further increased by using a SPDC source with broader spectrum bandwidth and FFPC with smaller FSR. We also note that by passing only the signal photons through a FFPC, the idler photons will still exhibit a comb-like behavior~\cite{Chang2023, cheng2023highdimensional}. Such a scheme can provide higher photon flux due to less filtering loss, and with potentially higher secure key rates for quantum key distribution~\cite{Sarihan2019, Cheng2021, Sarihan2023, Chang2023_iop}. Although we have already experimentally implemented the multiplexed source required for the proposed experimental implementation, integration of the scheme with quantum memories, implementing entanglement swapping, and distillation of links has not been done yet. At the same time, there has been great progress in each of these key steps for the construction of a long-distance quantum network.

\subsubsection*{Calculation of model parameters}\label{sec:appendix1}

In this section, we show how to translate physical parameters into the values of $p_{\ell}$, $m^{\star}$, and $\Delta t$ within our theoretical framework.

\paragraph*{Elementary link success probability $p_{\ell}$.}~In general, we have that 
\begin{equation}
    p_{\ell} = \eta_{r} \eta_{\ell} \eta^2,
\end{equation}
where $\eta_{\ell}=\exp(-\frac{\ell}{12\text{ km}})$ is the photon-loss contribution, with $\ell$ being the length of the elementary link. For $\ell=40\text{ km}$, we have $\eta_{\ell}\approx 0.036$. The factor of $\frac{1}{12\text{ km}}$ in the exponent corresponds to the 3.63~dB loss in the 10~km channel reported in the experiment in Ref.~\cite{Chang2021_2}. In addition, $\eta=0.69$ is the Debye--Waller factor, the ideal efficiency of the memory (fraction of photons captured in the memory out of all the incident photons)~\cite{DWF}, for a rare-Earth metal based quantum memory \cite{memory_numbers}, and $\eta_r = 0.79$ (approximately 1~dB) is the loss from other residual factors, which could include detector efficiencies, optical component loss, etc.~\cite{Chang2021}. Altogether, we have that $p_{\ell} = \eta_{r} \eta_{\ell} \eta^2 \approx 0.0134$. In the case of emissive memories like those based on Diamond vacancies, an effective Debye--Waller factor of 0.5 is assumed, since a linear-optical Bell state measurement is needed at each node to teleport the state of the incoming SPDC photon into the memory.

\paragraph*{Simulation time step $\Delta t$.}~This is given by the elementary link distance $\ell$ and the source rate per channel as 

\begin{equation}
    \Delta t = \max\Big\{\frac{1}{R}, \frac{n\ell}{c} \Big\}\,\text{seconds},
\end{equation}
where $n$ is the refractive index of the communication medium, $\ell$ is length of the elementary link, and $c$ is the speed of light in vacuum and $R$ is the rate of ebit pairs generated by the dimmest channel of the multiplexed SPDC source (see Fig.~\ref{fig:exp_imp})~(b).

\paragraph*{Memory cutoff $m^{\star}$.}~Using the memory coherence time of $T_2 \approx 1\, \rm{ms}$, as for a rare-Earth metal based quantum memory \cite{memory_numbers}, we can find the memory cutoff $m^{\star}$ as
\begin{equation}
    m^\star = \frac{T_2}{\Delta t}.
\end{equation}
For $\ell = 40~\text{km}$, we have $m^\star = 5$.
 \paragraph*{Fresh elementary link age $m_0$.}~This is given by the fidelity $f_s$ of the polarization entangled state produced by the SPDC source (which we assume to be the same for all channels, for simplicity) and the fidelity $f_0$ of the entangled state when it is freshly absorbed (emitted) by the quantum memory. The fidelity of the elementary link is thus given by $f_e \coloneqq f_sf_0^2$. This can be converted to the initial age of the elementary links using Eq.~\eqref{eq-fidelity_to_age}. Note that this also depends on $m^\star$.

For example, if the elementary link distance $\ell$ is varied from $5$~km to $25$~km, $p_\ell$ varies from $0.25$ to $0.05$, and $\Delta t$ goes from $0.025$~ms to $0.125$~ms. The latter in turn implies that $m^\star$ varies from $40$ to $8$, assuming a $T_2$ of approximately $1$~ms.

\subsubsection*{Figures of merit and repeater chain design for existing memory platforms}\label{sec:memories} 

Let us now address the following question. Given a particular number $n_{ch}$ of channels for multiplexing, and given other physical parameters corresponding to channel losses, quantum memory coherence times, etc., what is the optimal number of repeater nodes for a linear quantum repeater chain spanning a certain distance, and what would be the rate and fidelity of end-to-end entanglement distribution for this optimal setting? Intuitively, we do expect such an optimal number of nodes to exist, because as the number of nodes increases, even though the individual links become smaller and thus $p_\ell$ and $m^\star$ both increase, the non-deterministic nature of entanglement swapping and/or the ``age addition'' rule, will at some point begin to adversely affect the waiting time as the number of nodes increases. Furthermore, in a (source-)rate limited setting, $m^\star$ also saturates to a maximum value with decreasing elementary link length.

\begin{figure*}
    \centering
    \includegraphics[width=0.85\textwidth]{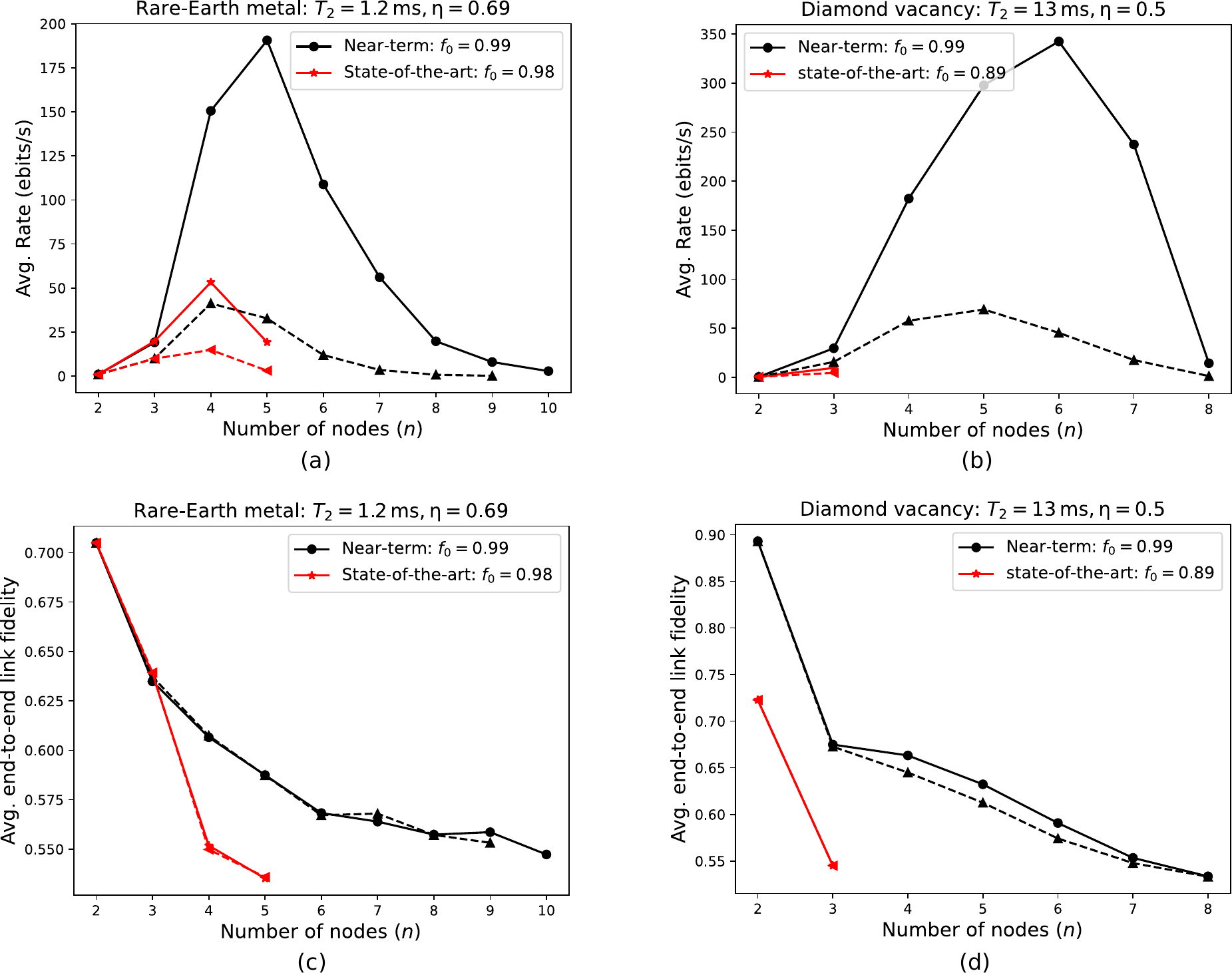}   
    \caption{\textbf{Performance of the proposed experimental implementation using our multiplexing policies.} Average rate of end-to-end entanglement distribution (a,b), including classical communication overheads, and average fidelity of the end-to-end link (c,d) for entanglement distribution along a 100~km repeater chain with five channels with rare-Earth ($\text{Pr}^{3+}$ ion) (a,c) and diamond vacancy (b,d) based memories using the FN \textsc{swap-asap} policy. The dotted (solid) lines represent $p_{sw} = 0.5 (p_{sw} = 1)$. Both state-of-the-art and near-term values for the fidelity $f_0$ of fresh entangled states generated (absorbed) by memory platforms are shown. Using three to five repeaters between the end nodes seems to be optimal for both memory platforms in terms of waiting times. 
    }
\label{fig:RE_Dia_FOM}
\end{figure*}

To answer the question posed above, let us set the end-to-end distance to be 100~km. We also use the FN \textsc{swap-asap} policy for illustrative purposes, and we also include the classical communication overheads, but we leave out distillation-based policies. The physical parameters are chosen to reflect the proposed experimental implementation above. They are as follows:
\begin{itemize}
    \item Number of multiplexing channels $n_{ch}$: We choose $n_{ch} = 5$, since although the SPDC spectrum in Fig.~\ref{fig:exp_imp}(b) can be divided into nine channels, both the ebit rates of polarization entangled states and their fidelity with respect to the maximally entangled Bell states fall as one moves away from the center of the spectrum. 
    \item Source rate $R$: We choose $R = 5000$ ebits/s. Since the original SPDC source is divided into different frequency channels, the number of ebits per channel is at least lower by a factor of $n_{ch}$. Furthermore, as shown in Fig.~\ref{fig:exp_imp}(b) the sinc function of the SPDC spectrum makes the ebit rates fall as one moves away from the central channel. Thus, to make a conservative estimate, we choose $R$ to be rate corresponding to the dimmest channel. This also allows us to choose the same rate for all channels (consistent with our simulation assumptions that all channels are identical).
    \item Fidelity of state at source $f_s$: We choose $f_s = 0.93$, as obtained in the experiment mentioned above. Recall that the fresh elementary link fidelity is given by $f_e = f_s f_0^2$.
\end{itemize}
Since quantum memories and entanglement swapping have not yet been implemented in the experiment, we choose two paradigmatic values for the swapping success probability, which are $p_{sw} = 0.5$ (linear optics) and $p_{sw} = 1.0$ (perfect gates on solid-state memories). For the quantum memories, we choose two quantum memory platforms for our analysis, a rare-Earth metal based memory ($\rm{Pr}^{3+}$ ions) and a diamond vacancy based memory. The state-of-the-art physical parameters for both the Rare-Earth and Diamond-based quantum memories are taken from Ref.~\cite[Table~1]{memory_numbers}.
Given the above SPDC source and memory parameters, we can calculate all the relevant simulation parameters following the prescription given above.
We plot the average rate and fidelity of end-to-end entanglement distribution both for near-term and state-of-the-art parameter values in Fig.~\ref{fig:RE_Dia_FOM}. We tabulate the values for the fidelity and entanglement distribution rate for the two memory platforms in the optimal setting, along with the hardware parameters, in Table~\ref{tab:memories}. The optimal number of nodes for the near-term (state-of-the-art) settings in terms of the waiting time seems to be around five to seven (two to four) for both memory platforms, which have $T_2$ times on the order of a few milliseconds. In the rate optimal setting (number of nodes), we anticipate entanglement distribution rates for the near-term (state-of-the-art) parameters ranging from a few Hertz to tens of Hertz, with an average end-to-end link fidelity of 60-70\% (53-55\%). We also see that as the swapping success probability is increased, especially close to the optimal choice of number of repeaters, there is a substantial increase in the rates, although the end-to-end fidelity does not substantially improve. This is also an artifact of choosing the farthest neighbor policy, which prioritizes forming long links. Also, the results of both platforms indicate that the end-to-end link fidelity decreases with an increasing number of nodes. 

Finally, we look at the performance of SN \textsc{doubling} and distillation-based multiplexing policies for the linear repeater chain considered above. Significant differences between \textsc{swap-asap} and \textsc{doubling} policies appear for $n>5$ (four links or more). For example, for a nine-node chain using a rare-Earth metal-based memory and 50\% swapping success probability, the SN \textsc{doubling} policy leads to around 60\% higher waiting time compared to FN \textsc{swap-asap}. Similar improvement using the \textsc{swap-asap} policy can also be seen for the Diamond vacancy-based memory.

Since the elementary links have (nearly) unit fidelity for the rare-Earth metal-based memory, distillation is detrimental as shown by the results in Fig.~\ref{fig:swdl_vs_dlsw_sims}. For the diamond vacancy-based memory with state-of-the-art initial state fidelity of 89\%, we also find that for the network parameters under consideration, distillation is in fact detrimental. This observation can be attributed to the higher CC overheads compared to the small increase in fidelity that distillation offers in this case. For example, the maximum rates available (for optimal number of repeaters) with deterministic entanglement swapping and using FN-based \textsc{distill-swap} policy is around 1~Hz, which is at least five times smaller than the rates available without distillation.

\begin{table}
\renewcommand{\arraystretch}{1.2}
\centering
\begin{tabular}{|c|c|c|c|c|c|c|c|c|}
\hline
\multicolumn{9}{|c|}{State-of-the-art parameters} \\
\hline
  Memory & $T_2$ & $\eta$ & $f_e$ & $p_\ell$ & $m^\star$ & $m_0$ & $R$ & $f$ \\
\hline\hline
Rare-Earth  & 1.2 ms & 0.69 & 0.89 & 0.023 & 7 & 1 & 53 Hz & 0.55 \\
\hline
Diamond & 13 ms & 0.5 & 0.74 & 0.003 & 52 & 22 & 10 Hz & 0.54 \\
\hline
\multicolumn{9}{|c|}{Near-term parameters} \\
\hline\hline
Rare-Earth  & 1.2 ms & 0.69 & 0.91 & 0.046 & 7 & 0 & 190 Hz & 0.59 \\
\hline
Diamond & 13 ms & 0.5 & 0.91 & 0.04 & 78 & 9 & 342 Hz & 0.59 \\
\hline
\end{tabular}

\caption{\textbf{Optimized figures of merit for a repeater chain using frequency multiplexing and quantum memories.} Hardware parameters for two quantum memory platforms (rare-earth ion and diamond vacancy) \cite[Table~1]{memory_numbers}\cite{SiV_mem,RE_mem}, and their respective optimal entanglement distribution rates ($R$) and end-to-end fidelities ($f$), for a 100~km repeater chain, based on the data in Fig.~\ref{fig:RE_Dia_FOM}. $T_2$ is the coherence time of the memory, $\eta$ is the absorption efficiency, i.e., Debye--Waller factor (DWF). in the case of Rare-Earth memory, for emissive memories like the Diamond vacancy-based memory an effective DWF of 0.5 is assumed. $f_e$ is the fidelity of a fresh elementary link (which includes the fidelity of the SPDC source state $f_s$ and the initial entanglement generated or absorbed by the quantum memories $f_0$). For the near-term case we consider $f_0 = 0.99$. Since, solid-state memories are considered, we assume swapping to be deterministic. Also, shown are the parameters $p_\ell$, $m^\star$, and $m_0$ of our theoretical model.}
\label{tab:memories}
\end{table}

\section*{Conclusions}
\label{sec:conclusions}

Near-term, resource-constrained quantum networks require the use of \textit{hardware-aware} and \textit{network state-aware} policies, in order to achieve high performance in terms of end-to-end waiting times and fidelities. These policies should be quasi-local, in order to reduce the impact of classical communication (CC) costs on performance. A careful assessment of the trade-off between CC costs and performance is crucial for the identification of optimal entanglement distribution policies for different hardware parameter regimes, particularly for first-generation quantum repeaters.

In this paper, we have taken steps in these directions. We have presented practical, quasi-local multiplexing-based policies for long-distance entanglement distribution using quantum repeaters with multiple memories. We call these policies farthest neighbor (FN) \textsc{swap-asap} and strongest neighbor (SN) \textsc{swap-asap}, adapting swap-as-soon-as-possible (\textsc{swap-asap}) policies to multiplexing-based linear networks. These policies go beyond fully local policies, such that the nodes have knowledge of states of the other nodes in the chain, but not necessarily full, global knowledge.

We have shown that not only do our quasi-local multiplexed policies retain their advantage over fully-local policies when CC costs are included, but this advantage can also increase with an increasing number of nodes. This is a surprising and counter-intuitive result, and shows that policies that use some knowledge of the network state (but not necessarily full, global knowledge) can enhance network performance, even when the CC costs associated with such knowledge are accounted for. The advantage attained by our quasi-local policies is rooted in the fact that they only require CC between connected regions of the chain; therefore, for most time steps, end-to-end CC is not required. This is an important conclusion from the point of view determining useful policies beyond the fully-local ones, especially for large quantum networks.
We also benchmarked our policies against the widely studied doubling policy. We have shown via simulations that the FN and SN \textsc{swap-asap} policies can yield a considerable advantage over the doubling policy, both in terms of reducing the average end-to-end waiting time and increasing the fidelity of the end-to-end link. These advantages occur in the most relevant parameter regimes for near-term quantum networks, which correspond to the most resource-constrained settings.

We then considered policies with entanglement distillation, and we proposed a new policy that we call \textsc{distill-asap}. This policy combines the benefits of existing distillation policies, like the banded, greedy, and pumping approaches, and thus is able to outperform the doubling policy nested with distillation. We also provide answers to two important policy questions related to distillation: When is distillation useful, and when it is useful, should we distill then swap, or the other way around? In this direction, the next question of immediate interest for future work is: How much to distill? Indeed, in this work, we only considered entanglement distillation policies that take $N$ links and distill them to one, but one could instead distill $N$ links to some $K$ links greater than one, and such policies have been the study of recent work~\cite{N2K_goodenough}.

Finally, we performed an experimental demonstration of multiplexing using a high-dimensional biphoton frequency comb (BFC), which would form the backbone of our proposed experimental implementation of a linear-chain quantum network with multiplexing capabilities. We then assessed the anticipated performance of such a network over 100~km for two concrete memory platforms, namely rare-earth ion and diamond vacancy based quantum memories, when using our multiplexing-based policies. 

Moving forward, it would be interesting to assess the optimality of the policies presented in this work. Optimal policies for multiplexed quantum repeater chains can be obtained using reinforcement learning (RL), using the methods developed in Refs.~\cite{IVSW22,reiss2022deep,haldar2023fastreliable}. By adding the appropriate classical communication costs to such policies, we could assess not only the optimality of the FN and SN \textsc{swap-asap} policies presented in this work, but also we could determine whether quasi-local policies in general can outperform fully-global ones, when CC costs are accounted for.

Throughout this work, we have also considered average values of the waiting time and fidelity of the end-to-end link. The behavior of the higher moments of these quantities, and in particular the distributions of these quantities, is an important consideration that has an impact on how well the average estimates the real behavior. We anticipate that this will involve new techniques, and is an interesting direction for future work.

\section*{Methods}
\subsection*{Noise model and decoherence} \label{sec:appendix3}
    
    We consider the following Pauli channel noise model for qubit decoherence~\cite{sarvepalli2009asymmetriccodes,ghosh2012surfacecodedecoherence}:
    \begin{equation}
        \mathcal{N}_{m_1^{\star},m_2^{\star}}(\rho)=p_I\rho+p_X X\rho X+p_Y Y\rho Y+p_Z Z\rho Z,
    \end{equation}
    where $X,Y,Z$ are the single-qubit Pauli operators, defined as
    \begin{equation}
        X\coloneqq\begin{pmatrix} 0 & 1 \\ 1 & 0 \end{pmatrix},\quad Y\coloneqq\begin{pmatrix} 0 & -\I \\ \I & 0 \end{pmatrix},\quad Z\coloneqq\begin{pmatrix} 1 & 0 \\ 0 & -1 \end{pmatrix},
    \end{equation}
    and the probabilities $p_I,p_X,p_Y,p_Z$ are defined as
    \begin{align}
        p_I&=\frac{1+\e^{-\frac{1}{m_2^{\star}}}}{2}-\frac{1-\e^{-\frac{1}{m_1^{\star}}}}{4},\\
        p_X&=\frac{1-\e^{-\frac{1}{m_1^{\star}}}}{4},\\
        p_Y&=\frac{1-\e^{-\frac{1}{m_1^{\star}}}}{4},\\
        p_Z&=\frac{1-\e^{-\frac{1}{m_2^{\star}}}}{2}-\frac{1-\e^{-\frac{1}{m_1^{\star}}}}{4},
    \end{align}
    where $m_1^{\star},m_2^{\star}\in\{1,2,\dotsc\}$. This channel is the Pauli-twirled version of the concatenated amplitude damping and dephasing channels.

    Let us also define the two-qubit Bell states as follows:
    \begin{align}
        \ket{\Phi^{\pm}}&\coloneqq\frac{1}{\sqrt{2}}(\ket{0,0}\pm\ket{1,1}),\quad \Phi^{\pm}\coloneqq\ketbra{\Phi^{\pm}}{\Phi^{\pm}},\\
        \ket{\Psi^{\pm}}&\coloneqq\frac{1}{\sqrt{2}}(\ket{0,1}\pm\ket{1,0}),\quad \Psi^{\pm}\coloneqq\ketbra{\Psi^{\pm}}{\Psi^{\pm}}.
    \end{align}

    \paragraph*{Decoherence of an entangled qubit pair.} Now, let us suppose that the initial state of an entangled qubit pair is the perfect maximally-entangled Bell state $\Phi^+$. Then, after $m\in\{1,2,3,\dotsc\}$ time steps, it is straightforward to show that the decohered entangled state is equal to
    \begin{align}
        &(\mathcal{N}_{m_1^{\star},m_2^{\star}}\otimes\mathcal{N}_{m_1^{\star},m_2^{\star}})^{\circ m}(\Phi^+)\nonumber\\
        &\quad=\frac{1}{4}\left((1+\e^{-2m/m_1^{\star}}+2\e^{-2m/m_2^{\star}})\Phi^+\right.\nonumber\\
        &\qquad \left.+(1+\e^{-2m/m_1^{\star}}-2\e^{-2m/m_2^{\star}})\Phi^-\right.\nonumber\\
        &\qquad \left.+(1-\e^{-2m/m_1^{\star}})\Psi^+\right.\nonumber\\
        &\qquad \left.+(1-\e^{-2m/m_1^{\star}})\Psi^-\right). \label{eq-decohered_entangled_state}
    \end{align}
    For a proof, we refer to Ref.~\cite[Appendix~E]{haldar2023fastreliable}. This implies that the fidelity of the state after $m$ time steps is $\frac{1}{4}(1+\e^{-2m/m_1^{\star}}+2\e^{-2m/m_2^{\star}})$.

    In order to connect with the results presented in ``Results'', for a given value of $m^{\star}$ as presented there, let us take
\begin{equation}\label{eq-m1_star_m2_star_assumptions}
    m_1^{\star}=2m^{\star},\quad m_2^{\star}=2m^{\star}.
\end{equation}
The state in \eqref{eq-decohered_entangled_state} is then
\begin{align}
    \sigma(m)&\coloneqq(\mathcal{N}_{2m^{\star},2m^{\star}}\otimes\mathcal{N}_{2m^{\star},2m^{\star}})^{\circ m}(\Phi^+)\\
    &=\frac{1}{4}(1+3\e^{-\frac{m}{m^{\star}}})\Phi^++\frac{1}{4}(1-\e^{-\frac{m}{m^{\star}}})(\Phi^-+\Psi^++\Psi^-)\\
    &=f(m)\Phi^++\frac{1-f(m)}{3}(\Phi^-+\Psi^++\Psi^-),\label{eq-decohered_entangled_state_isotropic}
\end{align}
where the fidelity $f(m)$ after $m$ time steps is equal to
\begin{equation}\label{eq-fidelity_function}
    f(m)\coloneqq\frac{1}{4}(1+3\e^{-\frac{m}{m^{\star}}}).
\end{equation}
Note that we have chosen the value of $m_2^{\star}$ in Eq.~\eqref{eq-m1_star_m2_star_assumptions} such that, at the cutoff time $m^{\star}$, the fidelity of the entangled qubit pair is $f(m^{\star})=\frac{1}{4}(1+3/e)\approx 0.5259$. We emphasize that our results can be applied to other choices of $m_1^{\star}$ and $m_2^{\star}$---in particular, choices that could take dephasing as the dominant source of noise~\cite{ourari2023telecomphotonsErion}, in order to make connections to prior works~\cite{RPL09,rozpedek2018parameterregimesrepeater,rozpedek2019neartermrepeaterNV,reiss2022deep,kamin2022exactrateswapasap}.

We can invert the fidelity function in Eq.~\eqref{eq-fidelity_function}, such that for a given fidelity $F$, the corresponding age of the qubits is given by
\begin{equation}\label{eq-fidelity_to_age}
    f^{-1}(F)=\left\lceil-m^{\star}\log((4F-1)/3)\right\rceil,
\end{equation}
for all $F\in(f(m^{\star}),1)$, where we take the ceiling $\ceil{\cdot}$ because we want an integer for the age.

We also remark that the decohered entangled state in Eq.~\eqref{eq-decohered_entangled_state} is a Bell-diagonal state of the form
\begin{equation}\label{eq-gen_iso_state}
    \rho_{(a,b,c)}\coloneqq(a+b)\Phi^++(a-b)\Phi^-+c\Psi^++c\Psi^-,
\end{equation}
with $a=\frac{1}{4}(1+\e^{-2m/m_1^{\star}})$, $b=\frac{1}{2}\e^{-2m/m_2^{\star}}$, and $c=\frac{1}{4}(1-\e^{-2m/m_1^{\star}})$. As such, its entanglement can be characterized entirely by the quantity $a+b$, namely, the fidelity of the state with respect to $\Phi^+$. In particular, the state is entangled if and only if this fidelity is strictly greater than $\frac{1}{2}$; see, e.g.,~\cite[Chapter~2]{Kha21}. Therefore, under the assumptions in Eq.~\eqref{eq-m1_star_m2_star_assumptions}, such that the fidelity is given by Eq.~\eqref{eq-fidelity_function}, we find that the state in Eq.~\eqref{eq-decohered_entangled_state} is entangled for all time steps $m\in\{0,1,2,\dotsc,m^{\star}\}$. Furthermore, inserting $F=\frac{1}{2}$ in Eq.~\eqref{eq-fidelity_to_age}, we find that the age at which entanglement is lost is $\ceil{m^{\star}\log(3)}$.

\subsection*{Entanglement distillation with the BBPSSW protocol}
%\label{sec:ent_distill_appendix}

    In this work, the entanglement distillation protocol we consider is the BBPSSW protocol introduced in Ref.~\cite{BBP96}. This protocol, involving the distillation of two entangled pairs to one, consists of a CNOT gate applied locally to each pair of qubits, followed by measuring the local target qubits of the CNOTs in the Pauli-$Z$ basis. If $f_1$ and $f_2$ are the fidelities of the two entangled qubit pairs, then the success probability of the distillation procedure is
    \begin{equation}\label{eq-distill_succ_appendix}
        P_{\text{distill}}(f_1,f_2)=\frac{8}{9}f_1f_2-\frac{2}{9}(f_1+f_2)+\frac{5}{9},
    \end{equation}
    and the fidelity of the resulting state is
    \begin{align}
        F_{\text{distill}}(f_1,f_2)&=\frac{1}{P_{\text{distill}}(f_1,f_2)}\left(\frac{10}{9}f_1f_2-\frac{1}{9}(f_1+f_2)+\frac{1}{9}\right)\\
        &=\frac{1-(f_1+f_2)+10f_1f_2}{5-2(f_1+f_2)+8f_1f_2}.
        \label{eq-distill_f}
    \end{align}
    Note that these formulas apply only when the input states to the distillation protocol are so-called \textit{isotropic states}, given by $f\Phi^++(\frac{1-f}{3})(\Phi^-+\Psi^++\Psi^-)$, with $f\in[0,1]$ being the fidelity. Also, distillation should be considered useful only when it can improve upon the fidelity of the best of the two links. In other words, in order for distillation to be useful, we should have $F_{\text{distill}}(f_1,f_2)>\max\{f_1,f_2\}$. In Fig.~\ref{fig:distill_useful_region}, we plot the $(f_1,f_2)$ region defined by this inequality. General restrictions on useful entanglement distillation protocols can be found in Ref.~\cite{zang2024nogousefulpurification}.
    
    From Ref.~\cite{BBP96}, we have the following relations, indicating how the Bell states $\Phi^{\pm}$ and $\Psi^{\pm}$ are transformed by the BBPSSW protocol.
    \begin{center}
    \renewcommand{\arraystretch}{1.2}
    \begin{tabular}{c||c|c|c|c}
          & $\Phi^+$ & $\Phi^-$ & $\Psi^+$ & $\Psi^-$  \\ \hline
          $\Phi^+$ & $\Phi^+$ & $\Phi^-$ & $0$ & $0$ \\ \hline
          $\Phi^-$ & $\Phi^-$ & $\Phi^+$ & $0$ & $0$ \\ \hline
          $\Psi^+$ & $0$ & $0$ & $\Psi^+$ & $\Psi^-$ \\ \hline
          $\Psi^-$ & $0$ & $0$ & $\Psi^-$ & $\Psi^+$ 
    \end{tabular}
    \end{center}
    Here, the rows correspond to the first link, and the columns correspond to the second link. From this, it follows that the unnormalized state corresponding to success of the BBPSSW protocol, when both input states are of the form \eqref{eq-decohered_entangled_state_isotropic}, is
    \begin{widetext}
    \begin{align}
        \sigma(m_1)\otimes\sigma(m_2)&\mapsto \left(f(m_1)f(m_2)+\frac{1-f(m_1)}{3}\frac{1-f(m_2)}{3}\right)\Phi^++\left(f(m_1)\frac{1-f(m_2)}{3}+\frac{1-f(m_2)}{3}f(m_1)\right)\Phi^-\nonumber\\
        &\qquad\qquad\qquad\qquad+2\frac{1-f(m_1)}{3}\frac{1-f(m_2)}{3}(\Psi^++\Psi^-)\\
        &\mapsto\left(\frac{10}{9}f(m_1)f(m_2)-\frac{1}{9}(f(m_1)+f(m_2))+\frac{1}{9}\right)\Phi^+\nonumber\\
        &\qquad\qquad\qquad\qquad+\left(-\frac{10}{9}f(m_1)f(m_2)+\frac{1}{9}(f(m_1)+f(m_2))+\frac{8}{9}\right)\frac{1}{3}(\Phi^-+\Psi^++\Psi^-),
    \end{align}
    \end{widetext}
    where in the last line we applied the isotropic twirling map $X\mapsto\Tr[\Phi^+X]\Phi^++\frac{1-\Tr[\Phi^+X]}{3}(\Phi^-+\Psi^++\Psi^-)$~\cite[Example~7.25]{Wat18_book}, which is required in order to use the BBPSSW protocol recursively. Therefore, after normalization, the state is
    \begin{multline}
        F_{\text{distill}}(f(m_1),f(m_2))\Phi^+\\+\frac{1-F_{\text{distill}}(f(m_1),f(m_2))}{3}(\Phi^-+\Psi^++\Psi^-).
    \end{multline}
    We observe that this is a state of the form $\sigma(m')$ in \eqref{eq-decohered_entangled_state_isotropic}, for some age $m'$. We can determine the age $m'$ corresponding to the state after distillation via $m'=f^{-1}(F_{\text{distill}}(f(m_1),f(m_2)))$, with the expression for $f^{-1}$ given in Eq.~\eqref{eq-fidelity_to_age}. In particular, the updated age $m'$ is given by:
    \begin{equation}
        m'=\left\lceil m^\star\log\left(\frac{15-6(f(m_1)+f(m_2))+24f(m_1)f(m_2)}{32f(m_1)f(m_2) - 2(f(m_1)+f(m_2))-1}\right)\right\rceil .
    \end{equation}

    \begin{figure}
        \centering
        \includegraphics[scale=0.6]{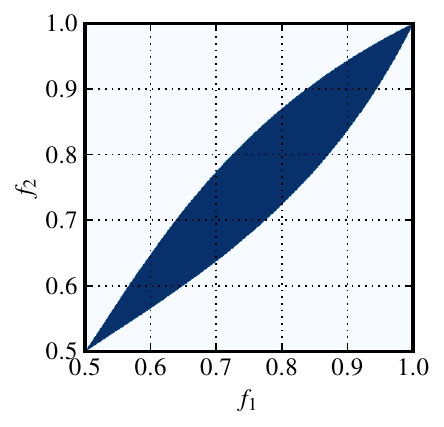}
        \caption{\textbf{Usefulness of distillation with the BBPSSW protocol.} The dark blue shaded region indicates the values of the initial fidelities, $f_1$ and $f_2$, for which the fidelity after distillation with the BBPSSW protocol~\cite{BBP96} exceeds the fidelity of the best of the two links, i.e., when $F_{\text{distill}}(f_1,f_2)\geq\max\{f_1,f_2\}$.}
        \label{fig:distill_useful_region}
    \end{figure}

\subsection*{Classical communication costs for local, quasi-local and global policies}
\label{sec:CC_appendix}

In this section, we aim to discuss some of the details related to our model for calculating the classical communication costs for different policies. As was mentioned in ``Results'', based on the amount of information available to nodes when making policy decisions, policies can be classified as local, quasi-local, and global. We remind the reader that in our model, for all policies, no CC is allowed within one MDP step. Consequently, if at the beginning of an MDP step it is determined that multiple interdependent swaps or distillation attempts are to be made, the failure of one leads to the failure of all dependent swaps or distillation attempts.

We begin by discussing local policies. Distillation, if performed, is always performed at the elementary link level when the links are freshly prepared, and thus no extra CC cost has to be added for distillation. For entanglement swapping, we considered that all the CC cost is paid only at the end of the protocol, since nodes determine their policy decisions or actions based on the locally perceived ages of their links. One concrete implementation through which this can be achieved is the following. Consider a central control or information processing node. All the Bell measurement outcomes (success or failure, and the outcomes obtained in case of success) are continuously transmitted by each node to this central unit. The central node waits until it receives a CC signal indicating a successful swap that leads to end-to-end entanglement generation. When such a signal is received, it relays this information to the end nodes, indicating the end of the protocol, along with the instructions about the list of local operations (Pauli rotations) they must perform to get the final end-to-end virtual link or links. Thus, no node waits for CC from any other node while taking decisions, but the end nodes know when to perform local operations to extract the end-to-end state. End nodes need not know beforehand when to stop, because perceived ages are always lower than the real ages of links, and hence there is no risk of them discarding their links before the CC arrives from the central unit indicating the end of one round of the protocol. Of course, the assumption of an additional central processing node is just an aid to understanding; this role could very well be played by, e.g., the end nodes, or the central node of the chain itself. 

For quasi-local policies, we considered that the CC cost for every MDP step is equal to the length (number of nodes) of the longest link involved in a swap multiplied by the elementary link CC time (heralding time). We stress here that this is a simple but useful estimate for the CC cost, and the actual CC cost will depend on the exact time evolution of the network and the concrete protocol used to share and process the classical information emerging from different nodes of the network at different times. Each node sees a constantly evolving picture of the entire network, as new classical communication signals reach it, and thus can decide to make its decisions based on any of these perceived snapshots. A concrete CC policy determination can only be done aided by a quantum network simulator such as those developed in Refs.~\cite{quisp, netsquid}, etc. The aim of our work is not to replace such a simulator, but to understand the merits and demerits of different classes of policies based on their use of network-state knowledge. Quasi-local policies utilize the fact that each node, using CC, can determine the nodes that it is connected to and the ages of its links. Ranking of links for swaps is done based on this knowledge. Thus, the most important question we must answer is the following: How do links know when to stop anticipating more CC signals and make a decision to perform subsequent swaps?

\begin{figure*}
    \centering
    \includegraphics[width=0.85\textwidth]{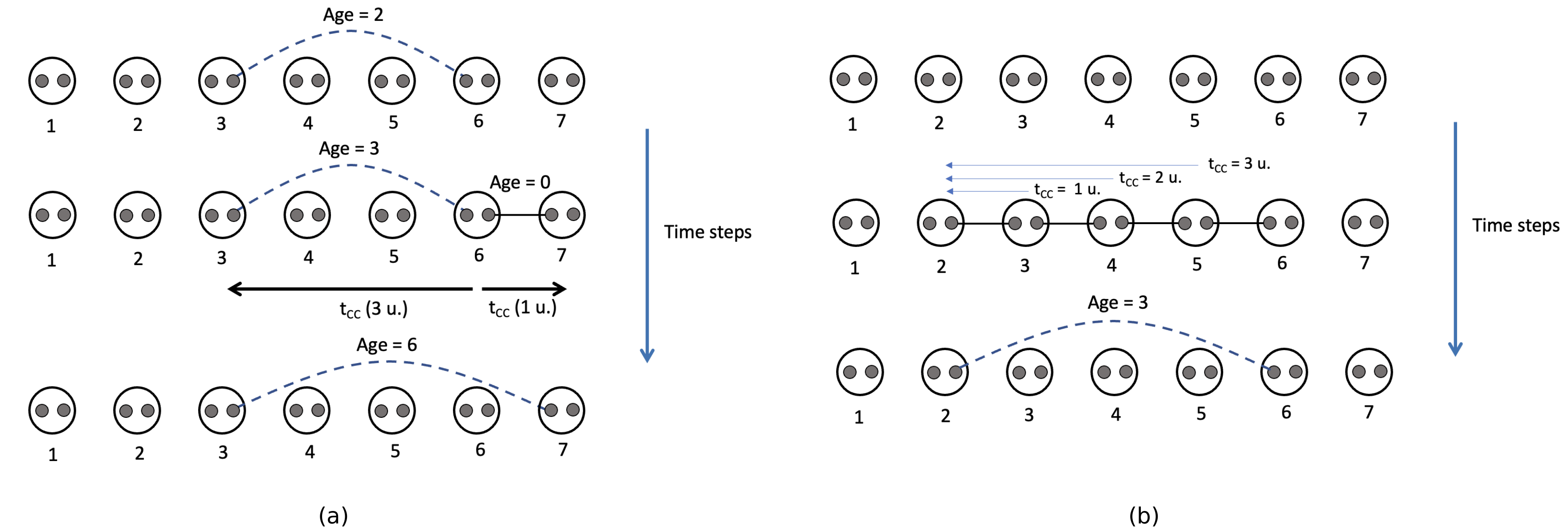}
    \caption{\textbf{Classical communication costs of entanglement swapping.} (a) A generic network evolution trajectory. The CC cost is estimated in our model as the length of the longest link involved in a swap in a given time-step. The swap at node 6 leads to a CC cost of 3 units, where 1 unit is equal to the heralding time for the elementary link. (b) Evolution of a section of the chain with no connections initially. The end nodes stop waiting for CC if there is a gap in the stream of CC packets they receive.}
    \label{fig:cc_example}
\end{figure*}

Consider the example of a generic network state as shown in Fig.~\ref{fig:cc_example}(a). We ignore distillation in this discussion, the CC cost for which is much easier to estimate and has been explained in ``Results''. Let us focus on the actions of node 6 and isolate one of the channels between the nodes of the chain. Node 6 knows at the beginning of the MDP step that it is connected to node 3. Now it will try to generate an elementary link with node 7. As soon as such a link is generated, it will perform a swap. Some of the time evolution trajectories that can emerge are as follows.
\begin{itemize}
    \item Case 1: If the swap at node 6 fails, it will know that either the swap it attempted failed or a swap at node 3 must have been attempted and failed in the meantime. In any case, it is free to restart its memory. As per our model, a swap at node 3 or node 6 will lead to a CC cost of three time steps, because that was the longest link involved in those swaps. Therefore, we overestimate the CC cost in this case. That is, in our policies, we unnecessarily keep node 6 waiting for two extra time steps.
    
    \item Case 2: The swap at 6 succeeds, in which case node 6 will know that it must wait for at least three time steps, since this will give it enough time to transmit information to nodes 3 and 7 that they now share a link. In this case, our estimate of CC cost is correct. Of course, node 6 could have started to attempt generating new elementary links and performed subsequent swaps in the meantime, but in our policies we make the choice to not do so. Node 6 restarts only after the required CC time has elapsed. 
    
    \item Case 3: In the time that the link between nodes 6 and 7 is produced, if the link between nodes 2 and 3 had also been produced, this information would be available at node 3. Now swaps will be attempted both at node 3 and 6. In this case, the longest link involved in a swap will be between nodes 3 and 7 (or equivalently between 2 and 6), and hence the CC cost is four time steps. This will give enough time for node 3 to communicate to nodes 2 and 7 that they are now connected.
\end{itemize}
Of course, the question now arises: why should nodes 2 and 7 wait for four time steps and not try to perform swaps as soon as possible? The answer is again similar in spirit to the reasoning above. When the elementary link between 2 and 3 is heralded, 3 can send a simultaneous CC signal to 2 indicating that it already has a connection with node 6 or node 7, as the situation may be at the instant the CC signal is generated. This will inform node 2 that if it attempts a swap and it is successful, this indicates that the swap at node 3 must have also been successful. Thus, it must wait for four time steps to find out whether it connects to node 6 or node 7. Our estimate of CC cost is therefore also correct in this case.

More generally, the cases of network evolution where many interdependent or connected swaps are attempted in the same time step are extremely complex, with many possible time evolution trajectories and different CC costs associated with each. But following the principle that CC can be shared transitively between nodes as soon as a new member joins a connected portion of the chain, many of these complexities can be resolved and a concrete swapping policy can be established. Taking such an approach, and as illustrated via the example above, the CC cost estimated by taking the length of the longest link involved in a swap is an upper bound on the real CC cost.  

To elaborate on such a scheme further, we give another example in Fig.~\ref{fig:cc_example}(b). We consider a portion of the network in which there are no links to begin with and show how a connected segment of the chain is established, with each node aware of its connections.

At the beginning of the protocol, each node tries to generate an elementary link. Suppose all of these attempts succeed. Now, each node will attempt swaps simultaneously, because the information about their respective elementary link generation success is available in one time step (heralding time). Now, if one of the swaps fails all of them must fail, and this information will be available at each node, allowing them to restart their memories in the next time step. This will also mark the end of one round of the protocol for this portion of the chain and no extra CC will be required. The other possibility is that all the swaps have succeeded, in which case the internal nodes will be free to restart without the need for any CC. They will just send their measurement outcomes to their next node, that is the node they were connected to at the beginning of the outcome, this CC signal will now be relayed along the chain with more information being added at each node until all this information reaches both the end nodes (see Fig.~\ref{fig:cc_example}(b)). One of the end nodes can keep utilizing this classical information to perform the local operations necessary to create an end-to-end entangled link. (Which one? We can have a set convention that the node on the left, i.e., one receiving CC signals from the right of it, will do Pauli rotations.) As more and more CC has reached them, they get informed about the success and failure of swaps that occurred farther and farther away, and they get to know that they now share a link with some far-off node that has a certain age. The end nodes will stop waiting for more CC once there is a gap of one time step in the CC packets reaching them. This will indicate that they have received all CC from the connected region of the chain.

\subsection*{Error estimates for Monte-Carlo simulations}
\label{sec:Monte_Carlo_details}
All simulation results in this work were produced by performing Monte Carlo simulations. Therefore, all average waiting time and fidelity (end-to-end link age) values have some statistical errors. For most hardware parameter regimes considered in the work, averages of the figures of merit were obtained by performing 20 batches of 1000 runs each. One run is considered as starting from a fully disconnected chain and evolving the network until at least one end-to-end link is generated. The average value reported is the mean of means of all the 20 batches. The standard deviation of this mean (of means) falls with increasing number of batches. The choice of number of batches was made to keep the standard deviation less than 5\% of the reported mean values. Therefore, statistical error (due to finite statistics generated by the Monte Carlo simulations) is less than 5\% for all values reported in our work. 
Here we would also like to point out that higher moments of the waiting times and fidelities are also an important figure of merit, e.g., if the variance of the waiting time is very large, having a small average does not guarantee that most runs will take a small time. The full probability distributions can also be obtained using our simulations, but this was not the focus of our work. Existing works have looked at these distributions in detail, for example Ref.~\cite{SSv19}.

\section*{Data availability}
The data generated through the simulations performed in this work and the experimental data that support the findings of this study are available on request from the corresponding authors.

\section*{Code availability}
The custom code developed to perform the simulations in this work that support the findings of this study are available on request from the corresponding authors.

\section*{Author contributions}
SH, SK, HL, and PB conceived the main idea of the work. SK, PB, and SH developed the theory and performed the computations.
SK and SH performed the analytical calculations. XC and KC performed the experiments supervised by CW. BK encouraged SH to investigate the role of classical communication.
SK, BK, CW, and HL supervised the findings of this work. All authors discussed the results and contributed to the writing of the manuscript. 

\section*{Competing interests}
The authors declare no competing interests.

\begin{acknowledgments}
The authors thank Thomas Searles and Sanjaya Lohani for helpful discussions. This work was supported by the Army Research Office Multidisciplinary University Research Initiative (ARO MURI) through the grant number W911NF2120214. SH also acknowledges support from the RCS program of Louisiana Boards of Regents through the grant LEQSF(2023-25)-RD-A-04. PB and HL also acknowledge the support the US-Israel Binational Science Foundation. SK acknowledges financial support from the German BMBF (Hybrid).
\end{acknowledgments}

\bibliography{refs_main_new}

\onecolumngrid

\section*{Supplementary Note 1: Properties of our multiplexing policies}
\label{sec:fom}

In this section, we evaluate our multiplexing policies with respect to the end-to-end waiting time and age of the youngest link. We exclude classical communication costs here.

\begin{figure}
    \centering
    \includegraphics[width = 0.45\textwidth]{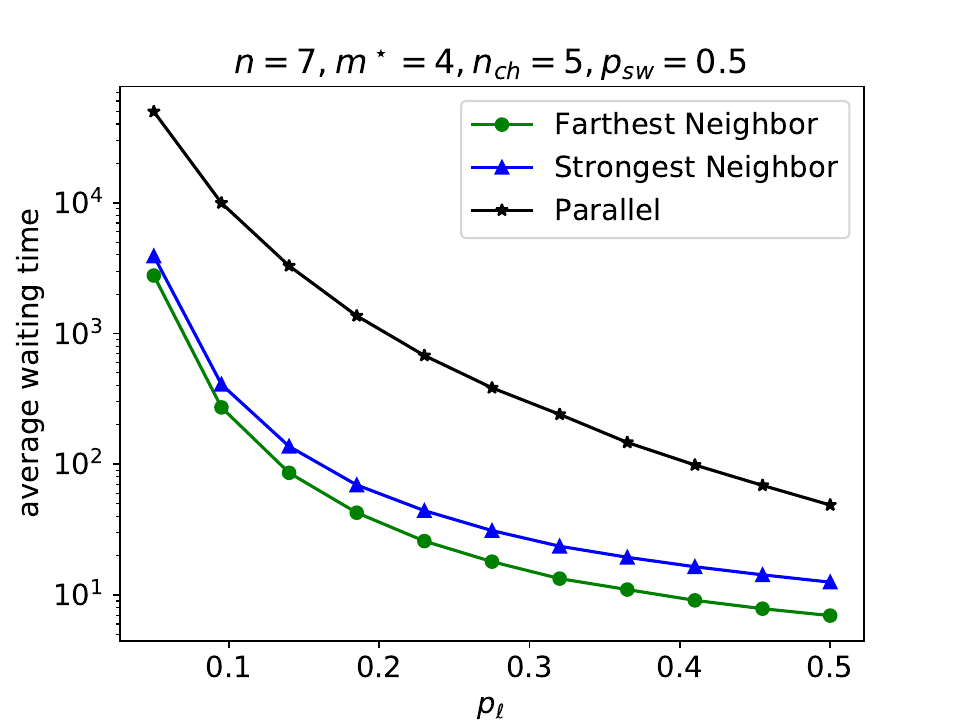}\quad
    \includegraphics[width = 0.45\textwidth]{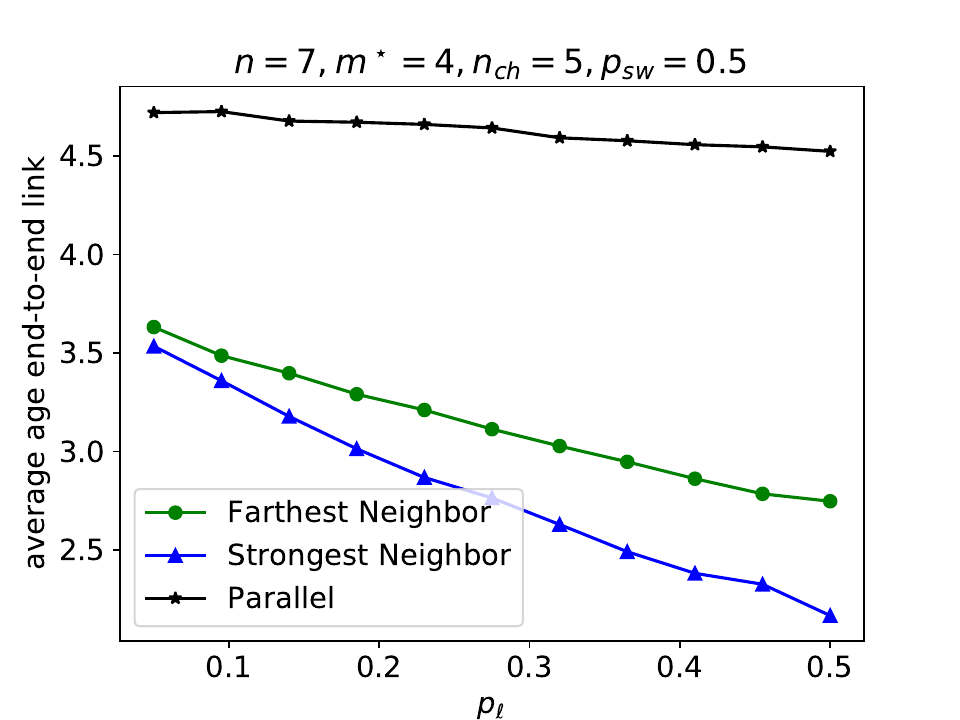}
    \caption{\textbf{Performance of our multiplexing policies against parallel \textsc{swap-asap}.} Average waiting time and average age of the youngest end-to-end link for a repeater chain with $n=7$ nodes, according to the multiplexed (SN and FN \textsc{swap-asap}) and non-multiplexed (parallel \textsc{swap-asap}) policies. 
    The multiplexed versions can enormously reduce average waiting time compared to parallel \textsc{swap-asap}, and at the same time increase end-to-end link fidelities. SN is the better choice for fidelity enhancement, while FN reduces the waiting time more significantly.}
    \label{fig:SN_FN_parallel}
\end{figure}

Intuitively, based on their definitions, the FN \textsc{swap-asap} policy aims to create the longest possible links, with the aim of creating an end-to-end link in the smallest possible number of time steps. Thus, the FN \textsc{swap-asap} policy qualitatively prioritizes minimizing the average waiting time.
On the other hand, the SN \textsc{swap-asap} policy aims to create the strongest or highest-fidelity links between the end nodes of the network. In order to confirm this intuitive picture, and to ascertain the benefit of multiplexing with these two policies, we benchmark the FN and SN policies against $n_{ch}$ copies of the usual \textsc{swap-asap} run in \textit{parallel}. In Fig.~\ref{fig:SN_FN_parallel}, we show that the multiplexed versions can indeed, as expected, enormously reduce the average waiting time compared to the parallel version of \textsc{swap-asap}, and at the same time provide higher end-to-end fidelities.

In Fig.~\ref{fig:SN_FN_parallel}, and for all of the plots shown in this paper, the parameter values that we choose are based on what is computationally feasible, but they are also based on what we expect with actual hardware, as we outline in the main text. The choice of network parameters such as $p_\ell, m^\star$ etc. in specific examples is also motivated by the intent to show the most prominent and overarching trends.

\begin{figure}
    \centering
    \includegraphics[width = 0.45\textwidth]{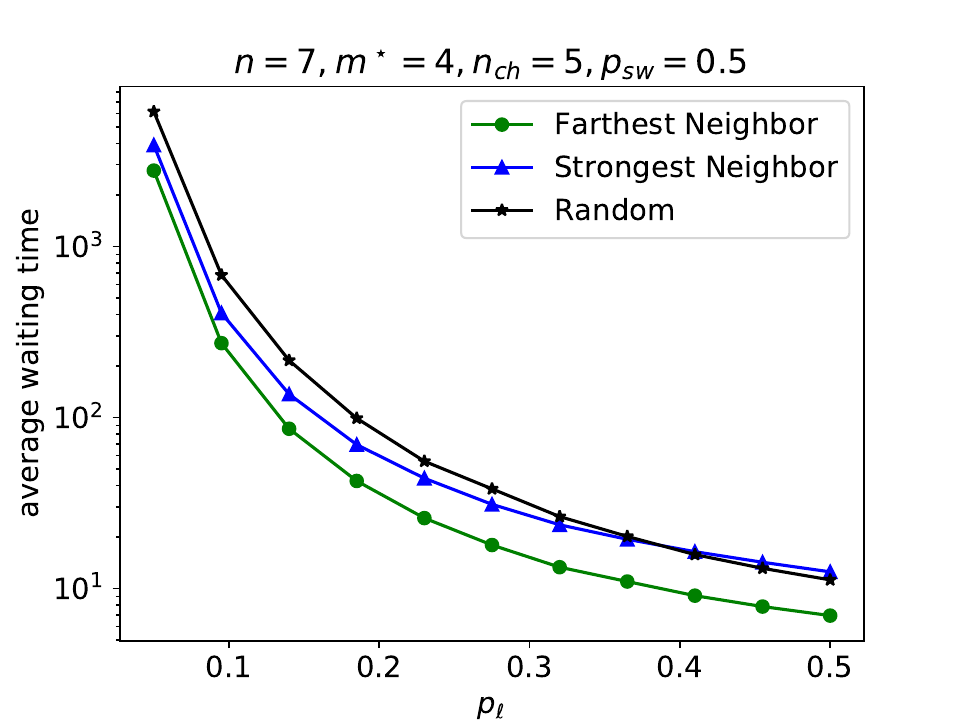}\quad
    \includegraphics[width = 0.45\textwidth]{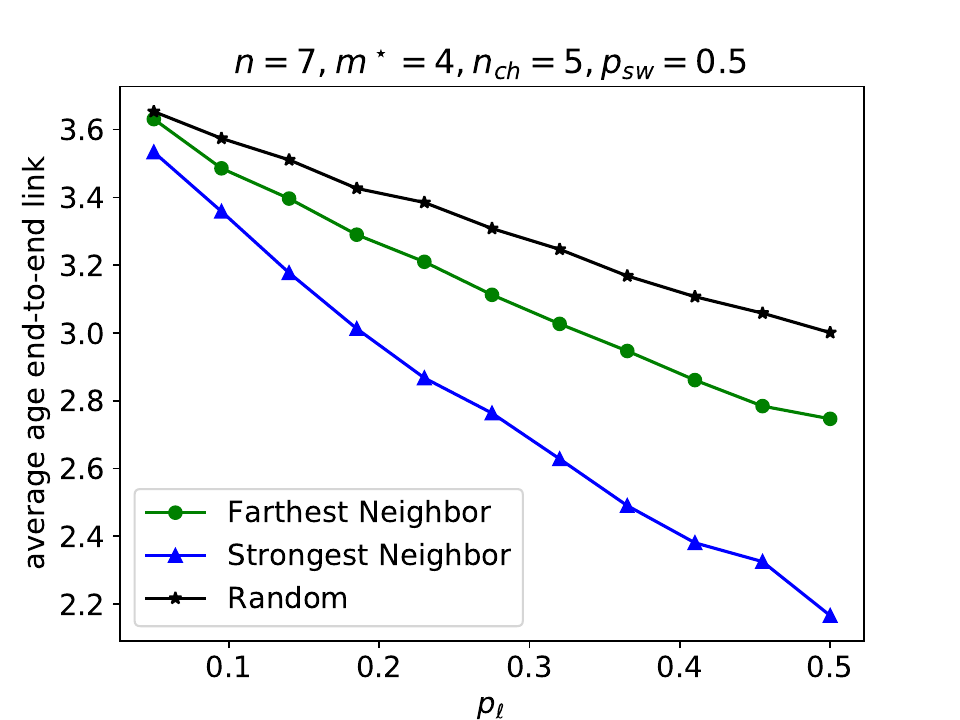}
    \caption{\textbf{Comparison among different multiplexing policies.} Average waiting time and average age of the youngest end-to-end link for a repeater chain with $n=7$ nodes, according to the FN, SN, and random \textsc{swap-asap} multiplexing policies. SN and FN both outperform the random policy, indicating the benefit of using knowledge of the network state. SN is the better choice for fidelity enhancement, while FN reduces the waiting time more significantly.
    }
    \label{fig:SN_FN_Random}
\end{figure}

The advantage of our SN and FN policies is not only based on the trivial advantage one would expect by shifting from a parallel policy to a multiplexing policy. Even amongst multiplexing policies, the SN and FN policies can be advantageous. Indeed, the fact that we assign ranks to different links and use this knowledge to perform entanglement swapping also provides considerable advantage. To show this advantage, we compare SN and FN multiplexing policies against the random \textsc{swap-asap} multiplexing policy defined above. As expected, both SN and FN policies perform better than the random policy, as we show in Fig.~\ref{fig:SN_FN_Random}. We also see that the FN policy provides lower average waiting times compared to the SN policy, and the SN policy provides a higher-fidelity end-to-end link compared to the FN policy, verifying our intuition.

\begin{figure}
    \centering
    \includegraphics[width = 0.65\textwidth]{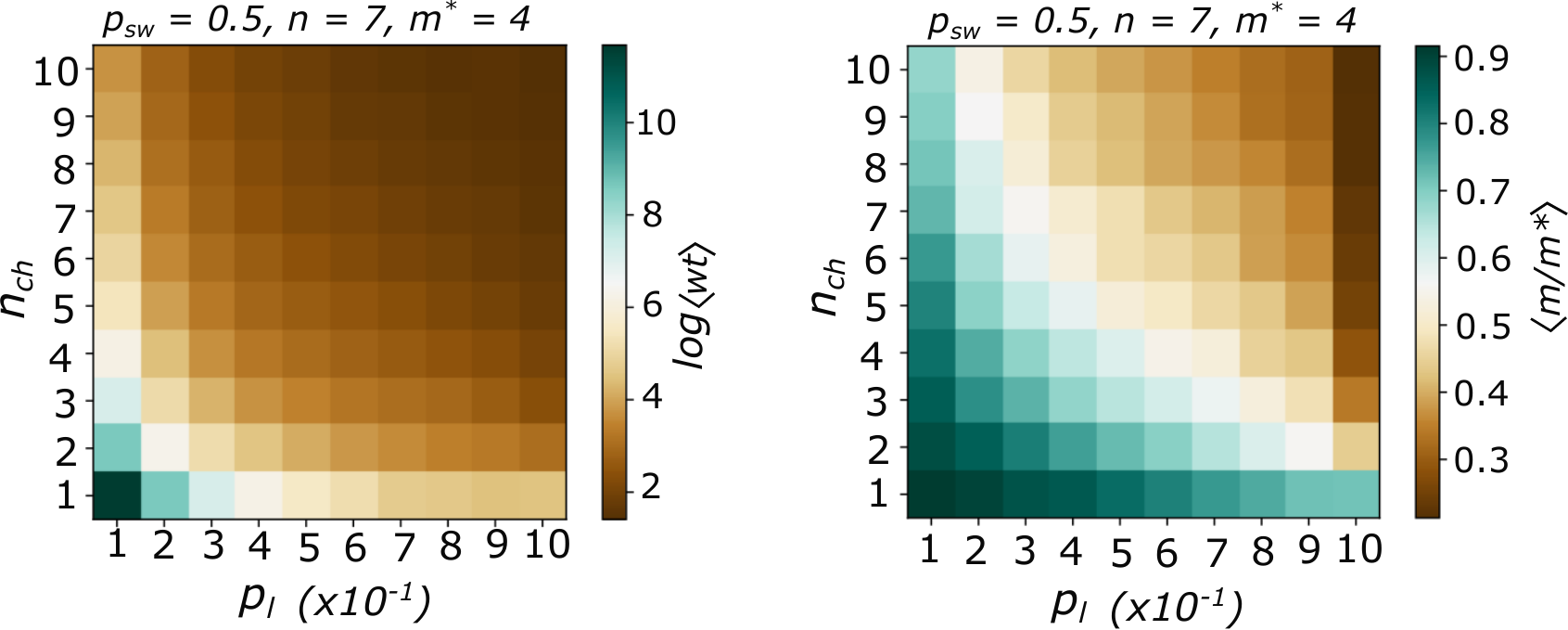}
    \caption{\textbf{Multiplexing can beat the detrimental effect of lower link success probability.} Heat map of average waiting time (left) (denoted $\left<wt\right>$) and ratio of the average age of youngest end-to-end link to the cutoff (right) as a function of elementary link success probability and number of channels. Contours of fixed values run along diagonals from the top-left to the bottom-right, showing that lower link probabilities can be compensated for by increasing the number of channels. Both the average waiting time and the average age of the youngest end-to-end link reduce as the number of channels increases, showing the usefulness of multiplexing. For the waiting time results we use the FN policy, and for the average age results we use the SN policy.}
\label{fig:FN_n_ch_p_l}
\end{figure}

Next, we look at how increasing the number $n_{ch}$ of channels affects the average waiting time of the repeater chain. Increasing $n_{ch}$ increases the effective $p_\ell$ of the chain according to $p_{\ell}\to 1-(1-p_{\ell})^{n_{ch}}$; therefore, in Fig.~\ref{fig:FN_n_ch_p_l}, approximately constant values of the waiting time run along the diagonal from the top-left to the bottom-right, indicating that as the elementary link probability goes down, in order to keep the average waiting time and end-to-end fidelity constant, we should increase the number of channels. In other words, for a given average waiting time or end-to-end fidelity requirement, low elementary link probabilities can be compensated for by increasing the number of multiplexing channels.

\begin{figure}
    \centering
    \includegraphics[width = 0.65\textwidth]{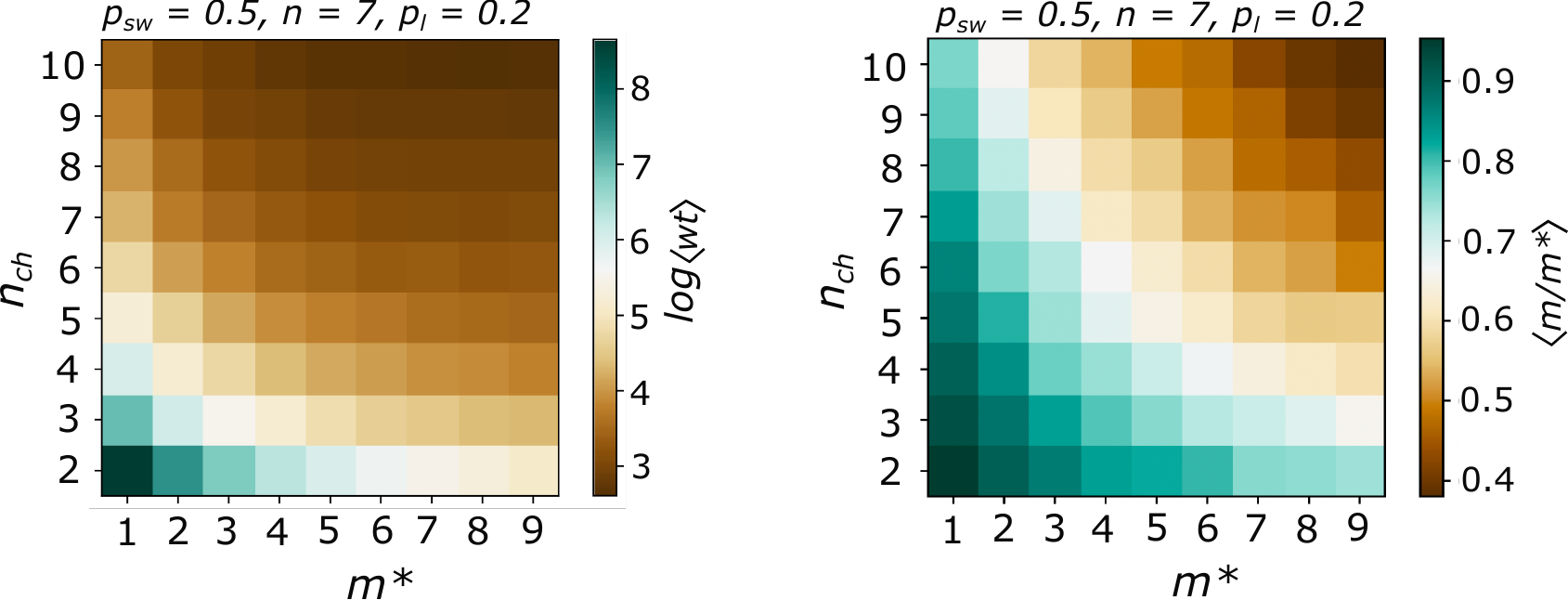}
    \caption{\textbf{Multiplexing can beat the detrimental effect of lower memory coherence time.}Heat map of average waiting time (left) (denoted by $\left<wt\right>$) and ratio of the average age of youngest end-to-end link to the cutoff (right) as a function of link cutoff age and number of channels. Contours of fixed values run along diagonals from the top-left to the bottom-right, showing that lower cutoffs, which translates to memories with smaller coherence times, can be compensated for by increasing the number of channels. For waiting time results we use the FN policy and for average age results we use SN policy.
    }
\label{fig:FN_n_ch_m}
\end{figure}

In Fig.~\ref{fig:FN_n_ch_m}, we fix the elementary link and entanglement swapping success probabilities while varying $n_{ch}$ and the coherence time $m^\star$. Here we see a similar trend, a low value of the memory coherence time can again be compensated for by increasing the number of channels for both figures of merit.

Of course, FN, SN, and random are not the only three possible options for multiplexing-based policies. For example, one could develop a ranking scheme based on a combination of age and distance, or one could perform a search over all possible pairings of links, looking for the optimal pairings in terms of some figure of merit (such as the sum of lengths of all virtual links formed). With respect to such optimal ranking schemes, it is not clear how close the SN and FN schemes are to being optimal; they are, however, simple and intuitive policy choices that require no optimization in terms of pairing of links for entanglement swapping, and hence are practical for near-term implementation. We defer the discussion of other policy alternatives and comments on the optimality of SN and FN \textsc{swap-asap} to Supplementary Note 4.

Let us now compare our policies with the doubling policy, described in the main text. 

\begin{figure}
    \centering
    \includegraphics[width = 0.45\textwidth]{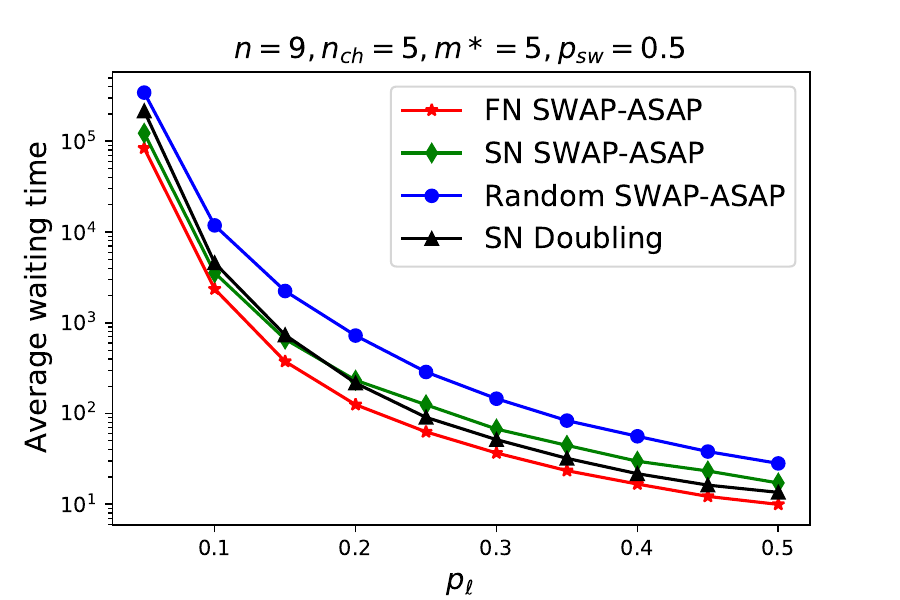}\quad
    \includegraphics[width = 0.45\textwidth]{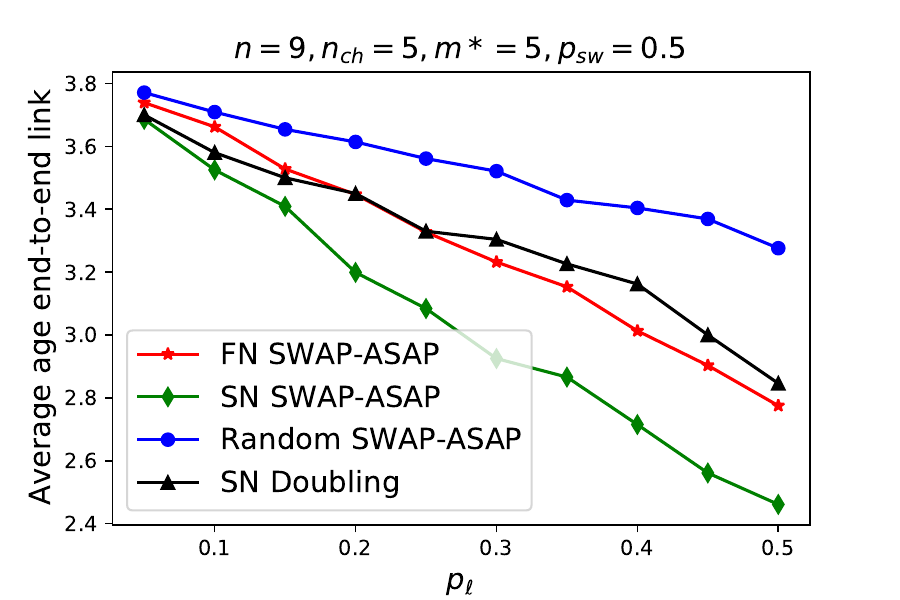}
    \caption{\textbf{Performance of multiplexing policies compared to the \textsc{doubling} policy.} Average waiting time (left) and average age of the youngest end-to-end link (right) for a repeater chain with $n=9$ nodes (eight elementary links), according to the FN, SN, and random \textsc{swap-asap} policies, compared to the SN \textsc{doubling} policy. For small values of $p_{\ell}$, both the FN and SN \textsc{swap-asap} policies outperform the doubling policy in terms of waiting time. The SN \textsc{swap-asap} policy outperforms doubling for the average age for all values of $p_{\ell}$ up to $p_{\ell}=0.5$.
    }
    \label{fig:doubling}
\end{figure}

Our FN and SN \textsc{swap-asap} policies can outperform a doubling based multiplexing policy. This is what we show in Fig.~\ref{fig:doubling}. In terms of the average waiting time, the FN \textsc{swap-asap} policy outperforms the SN \textsc{doubling} policy for all values of link success probability $p_\ell$ up to $p_{\ell}=0.5$, including very low values of $p_{\ell}$, which is often most relevant in practice. It is important to mention here that this result is not true for all parameter regimes; for high values of $p_\ell$ and $m^\star$, the doubling policy can actually outperform our \textsc{swap-asap} policies. Our results are consistent with the results in previous works. Specifically, our results generalize those in Ref.~\cite{SvL21} to the case of finite coherence times. 

At the same time, our results also show the importance of ranking the links in an appropriate way, reflected in the fact that the \emph{random} \textsc{swap-asap} policy and the SN \textsc{swap-asap} policies actually perform worse than the doubling policy in terms of average waiting time. We see very similar results in terms of the average age of the youngest end-to-end link. With respect to this figure of merit, the SN \textsc{swap-asap} policy outperforms all other policies discussed here, including the SN \textsc{doubling} policy.

\section*{Supplementary Note 2: Entanglement distillation with the BBPSSW protocol for more than two links}\label{sec:ent_distill_appendix}
    Recall that if $f_1$ and $f_2$ are the fidelities of the two links being distilled using the BBPSSW protocol, then the success probability $P_{\text{distill}}(f_1,f_2)$ of the protocol and the resulting fidelity $F_{\text{distill}}(f_1,f_2)$ are given by
    \begin{align}
        P_{\text{distill}}(f_1,f_2)&=\frac{8}{9}f_1f_2-\frac{2}{9}(f_1+f_2)+\frac{5}{9},\label{eq-distill_p}\\
        F_{\text{distill}}(f_1,f_2)&=\frac{1-(f_1+f_2)+10f_1f_2}{5-2(f_1+f_2)+8f_1f_2}.\label{eq-distill_f}
    \end{align}
    
    When there are multiple (more than two) links that can be distilled, and we consider the pairwise BBPSSW protocol, the question of what is the optimal \textit{distillation policy} naturally arises, i.e., what is the optimal pairing of links when performing entanglement distillation of more than two links. This question has been addressed in existing literature; see, e.g., Refs.~\cite{DBC99,VLMN09}. Here, we consider the following policies.% as named in Ref.~\cite{VLMN09}.
\begin{itemize}
    \item \textit{Banded policy:} Referred to as ``scheme A'' in Ref.~\cite{DBC99}, this involves distillation of entangled pairs that have the same age. Assuming that $n_{ch}$ is a power of two, and assuming that all initial $n_{ch}$ pairs have the same fidelity, they are all put into pairs and a first round of $\frac{n_{ch}}{2}$ distillation attempts is made. The resulting links are all of the same fidelity and are paired again such that, in the next round, $\frac{n_{ch}}{4}$ distillation attempts are made. This procedure continues for $\log_2(n_{ch})$ rounds, at the end of which we have one link remaining. Assuming all of the initial links have fidelity $f_0$, the fidelity after $r$ rounds is given by the recurrence relation $f_r=F_{\text{distill}}(f_{r-1},f_{r-1})$. It has been shown in Ref.~\cite{DBC99} that, as long as the initial links are all entangled ($f_0>\frac{1}{2}$), this policy can achieve arbitrarily high fidelities with an increasing value of $n_{ch}$.
    
    \item \textit{Pumping policy:} Referred to as ``scheme C'' in Ref.~\cite{DBC99}, this policy proceeds as follows. Two links of fidelity $f_0$ are distilled, then a fresh link of fidelity $f_0$ is used to distill the previously distilled link, and so on. This policy is tailored to a practical situation in which fresh links are used to distill the surviving links of previous time steps, and thus the links being distilled do not necessarily have the same fidelity, as in the banded policy. If the initial fidelity satisfies $f_0>\frac{1}{2}$, the fidelity $f_r$ after $r\in\{1,2,\dotsc\}$ rounds of the pumping policy is given by the recurrence relation $f_r=F_{\text{distill}}(f_{r-1},f_0)$. Below, in Theorem~\ref{thm-pumping_recurrence_soln}, we provide a closed-form, analytical solution to this recurrence relation.

\end{itemize}

\subsection*{The pumping policy}\label{sec:pumping_analytical}

Let us recall that the fidelity of the pumping policy after $r\in\{1,2,\dotsc\}$ rounds is defined by the recurrence relation $f_r=F_{\text{distill}}(f_{r-1},f_0)$, where $f_0$ is the initial fidelity. Using Eq.~\eqref{eq-distill_f}, we can write this recurrence relation as
\begin{equation}\label{eq-pumping_recurrence_relation}
    f_r=\frac{a_1f_{r-1}+a_2}{a_3f_{r-1}+a_4},\quad a_1\coloneqq 10f_0-1,\quad a_2\coloneqq 1-f_0,\quad a_3\coloneqq 8f_0-2,\quad a_4\coloneqq 5-2f_0.
\end{equation}
In this section, we provide a closed-form, analytical solution to this recurrence relation.

\begin{lemma}\label{lem-pumping_recurrence}
    Define the matrix $A\coloneqq\begin{pmatrix} a_1 & a_2 \\ a_3 & a_4 \end{pmatrix}$, where $a_1,a_2,a_3,a_4$ are defined in Eq.~\eqref{eq-pumping_recurrence_relation}. For $r\in\{1,2,\dotsc\}$, consider the vector
    \begin{equation}
        \vec{v}_r=\begin{pmatrix} v_{r,0} \\ v_{r,1} \end{pmatrix}\coloneqq A^r\vec{v}_0,\quad \vec{v}_0=\begin{pmatrix} f_0 \\ 1 \end{pmatrix}.
    \end{equation}
    It holds that $f_r=\frac{v_{r,0}}{v_{r,1}}$.
\end{lemma}

\begin{proof}
    We prove this by induction on $r$. For the base case of $r=1$, we have 
    \begin{equation}
        \begin{pmatrix} v_{1,0} \\ v_{1,1} \end{pmatrix}=\begin{pmatrix} a_1 & a_2 \\ a_3 & a_4 \end{pmatrix}\begin{pmatrix} f_0 \\ 1 \end{pmatrix}=\begin{pmatrix} a_1f_0+a_2 \\ a_3f_0+a_4 \end{pmatrix},
    \end{equation}
    which is consistent with Eq.~\eqref{eq-pumping_recurrence_relation}. Now, the induction step: assuming that the result holds for $r\in\{1,2,\dotsc\}$, using Eq.~\eqref{eq-pumping_recurrence_relation} we obtain
    \begin{align}
        f_{r+1}&=\frac{a_1f_r+a_2}{a_3f_r+a_4}\\
        &=\frac{a_1\frac{v_{r,0}}{v_{r,1}}+a_2}{a_3\frac{v_{r,0}}{v_{r,1}}+a_4}\\
        &=\frac{a_1v_{r,0}+a_2v_{r,1}}{a_3v_{r,0}+a_4v_{r,1}}.
    \end{align}
    From this, we see that
    \begin{align}
        \begin{pmatrix} a_1v_{r,0}+a_2v_{r,1} \\ a_3v_{r,0}+a_4v_{r,1} \end{pmatrix} &= \begin{pmatrix} a_1 & a_2 \\ a_3 & a_4 \end{pmatrix}\begin{pmatrix} v_{r,0} \\ v_{r,1} \end{pmatrix}\\
        &=A\cdot A^r\vec{v}_0\\
        &=A^{r+1}\vec{v}_0\\
        &=\vec{v}_{r+1},
     \end{align}
     which means that $f_{r+1}=\frac{v_{r+1,0}}{v_{r+1,1}}$. This completes the proof.
\end{proof}

\begin{theorem}\label{thm-pumping_recurrence_soln}
    The fidelity after $r\in\{1,2,\dotsc\}$ rounds of the pumping policy is given by $f_r=\frac{\alpha_r}{\beta_r}$,where
    \begin{align}
        \alpha_r&=-(1 - 4 f_0 + 6 f_0^2) \left(\left(2 + 4 f_0 - \sqrt{7 - 26 f_0 + 28 f_0^2}\right)^r - \left(2 + 4 f_0 + \sqrt{7 - 26 f_0 + 28 f_0^2}\right)^r\right)\nonumber\\
        &\qquad+ f_0 \sqrt{7 - 26 f_0 + 28 f_0^2} \left(\left(2 + 4 f_0 - \sqrt{7 - 26 f_0 + 28 f_0^2}\right)^r + \left(2 + 4 f_0 + \sqrt{7 - 26 f_0 + 28 f_0^2}\right)^r\right),\label{eq-pumping_fidelity_numerator}\\
        \beta_r&=-(3 - 8f_0+8f_0^2) \left(\left(2 + 4 f_0 - \sqrt{7 - 26 f_0 + 28f_0^2}\right)^r - \left(2 + 4 f_0 + \sqrt{7 - 26f_0 + 28 f_0^2}\right)^r\right)\nonumber\\
        &\qquad+ \sqrt{7 - 26 f_0 + 28 f_0^2} \left(\left(2 + 4 f_0 - \sqrt{7 - 26 f_0 + 28 f_0^2}\right)^r + \left(2 + 4 f_0 + \sqrt{7 - 26 f_0 + 28 f_0^2}\right)^r\right).\label{eq-pumping_fidelity_denominator}
    \end{align}
\end{theorem}

\begin{proof}
    We start with the eigenvalues and eigenvectors of the matrix $A=\begin{pmatrix} a_1 & a_2 \\ a_3 & a_4 \end{pmatrix}$. The eigenvalues $\lambda_{\pm}$ and corresponding eigenvectors $\vec{x}_{\pm}$ are readily found to be
    \begin{align}
        \lambda_{\pm}&=2+4f_0\pm\sqrt{7-26f_0+28f_0^2},\\
        \vec{x}_{\pm}&=\begin{pmatrix} \omega_{\pm} \\ 1 \end{pmatrix},\quad \omega_{\pm}=\frac{-3+6f_0\pm\sqrt{7-26f_0+28f_0^2}}{-2+8f_0}.\label{eq-pumping_fidelity_pf1}
    \end{align}
    In particular, we have that $A\vec{x}_{\pm}=\lambda_{\pm}\vec{x}_{\pm}$. Then, the initial vector $\vec{v}_0=\begin{pmatrix} f_0 \\ 1 \end{pmatrix}$ can be expressed in terms of the eigenvectors $\vec{x}_{\pm}$ as
    \begin{equation}
        \vec{v}_0=\begin{pmatrix} f_0 \\ 1 \end{pmatrix}=c_1\vec{x}_++c_2\vec{x}_-,\quad c_1=\frac{f_0-\omega_-}{\omega_++\omega_-},\quad c_2=1-c_1.
    \end{equation}
    We therefore find that
    \begin{align}
        \vec{v}_r&=A^r\vec{v}_0\\
        &=c_1A^r\vec{x}_++c_2A^r\vec{x}_-\\
        &=c_1\lambda_+^r\vec{x}_++c_2\lambda_-^r\vec{x}_-\\
        &=c_1\lambda_+^r\begin{pmatrix} \omega_+ \\ 1 \end{pmatrix}+c_2\lambda_-^r\begin{pmatrix}\omega_- \\ 1 \end{pmatrix},
    \end{align}
    which implies (via Lemma~\ref{lem-pumping_recurrence}) that
    \begin{equation}\label{eq-pumping_fidelity_alt}
        f_r=\frac{c_1\lambda_+^r\omega_++c_2\lambda_-^r\omega_-}{c_1\lambda_+^r+c_2\lambda_-^r}=\frac{\lambda_-^r\omega_-+c_1(\lambda_+^r\omega_+-\lambda_-^r\omega_-)}{\lambda_-^r+c_1(\lambda_+^r-\lambda_-^r)}.
    \end{equation}
    After substitution and much simplification, we obtain the desired expressions in Eq.~\eqref{eq-pumping_fidelity_numerator} and Eq.~\eqref{eq-pumping_fidelity_denominator}.
\end{proof}

\begin{corollary}
    For the sequence $\{f_r\}_r$ of fidelities defined by the pumping policy, it holds that
    \begin{equation}\label{eq-pumping_fidelity_limit_appendix}
        \lim_{r\to\infty}f_r=\frac{-3+6f_0+\sqrt{7-26f_0+28f_0^2}}{-2+8f_0}.
    \end{equation}
\end{corollary}

\begin{proof}
    We can write the expression for $f_r$ in Eq.~\eqref{eq-pumping_fidelity_alt} as
    \begin{equation}
        f_r=\frac{\left(\frac{\lambda_-}{\lambda_+}\right)^r\omega_-+c_1\left(\omega_+-\left(\frac{\lambda_-}{\lambda_+}\right)^r\omega_-\right)}{\left(\frac{\lambda_-}{\lambda_+}\right)^r+c_1\left(1-\left(\frac{\lambda_-}{\lambda_+}\right)^r\right)}.
    \end{equation}
    Now, it is straightforward to see that $\lambda_{\pm}\geq 0$ for all $f_0\in[\frac{1}{2},1]$, and furthermore that $\lambda_+>\lambda_-$. Therefore, $\frac{\lambda_-}{\lambda_+}\in[0,1)$, which means that $\lim_{r\to\infty}\left(\frac{\lambda_-}{\lambda_+}\right)^r=0$. We therefore find that $\lim_{r\to\infty} f_r=\omega_+$, and because $\omega_+$ (defined in Eq.~\eqref{eq-pumping_fidelity_pf1}) is precisely the expression on the right-hand side of Eq.~\eqref{eq-pumping_fidelity_limit_appendix}, we have the desired result.
\end{proof}

We note that the limiting fidelity on the right-hand side of \eqref{eq-pumping_fidelity_limit_appendix} is strictly less than one whenever $f_0<1$. This proves that, in contrast to the banded policy, with the pumping policy it is not possible to distill links to arbitrarily high fidelity, even with an infinite number of rounds.

\subsection*{Banding vs. pumping}\label{sec:DS_vs_SD_example_1}

To gain an understanding of the difference between the banding and pumping policies of entanglement distillation, we consider here a three-node chain with $n_{ch}$ channels between the nodes. 
In this setting, by the \textsc{swap-distill} policy, we mean that all $n_{ch}$ links are active at the same time, entanglement swapping is performed on all of them and we assume that they all succeed, and then all of the end-to-end links are distilled into one end-to-end link. By the \textsc{distill-swap} policy, we mean that all $n_{ch}$ on both sides are distilled into one link, and then entanglement swapping is performed on them and assumed to be successful.

Let us first consider the banding approach to distillation, as defined above. Once all $n_{ch}$ links are put in pairs, a first round of $\frac{n_{ch}}{2}$ distillation attempts are made; the resulting links are all of the same fidelity, and they are paired again. In the next round, $\frac{n_{ch}}{4}$ distillation attempts are made, and the process continues until we have one link remaining. Thus a total number of $\log_2(n)$ rounds are needed (supposing that $n_{ch}$ is a power of two).
\begin{figure}
    \centering
    \includegraphics[width = 0.40\textwidth]{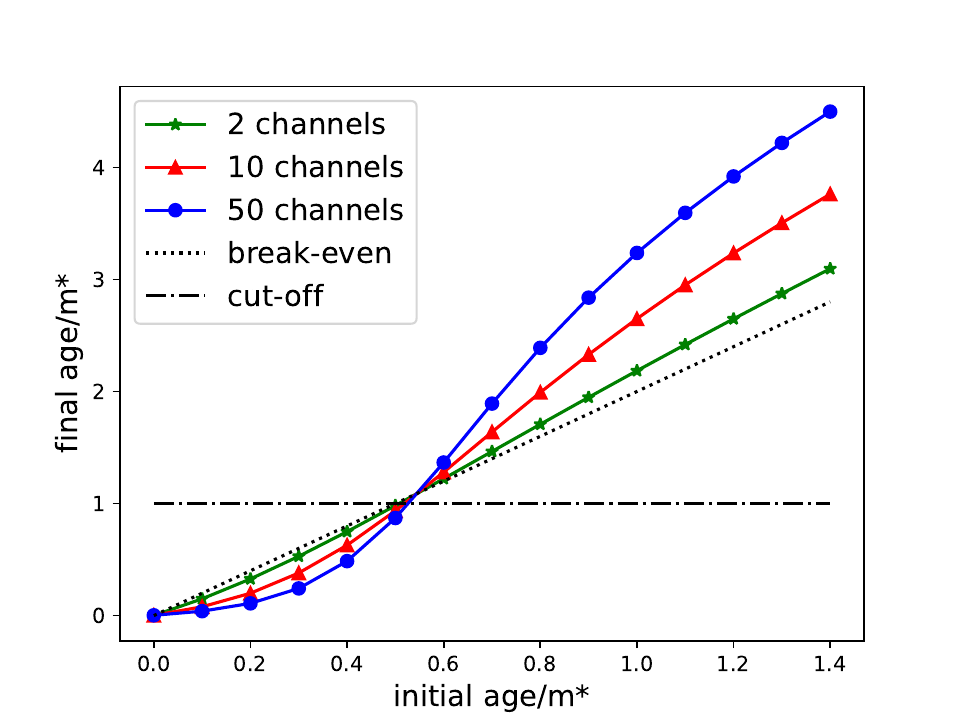}\quad
    \includegraphics[width = 0.40\textwidth]{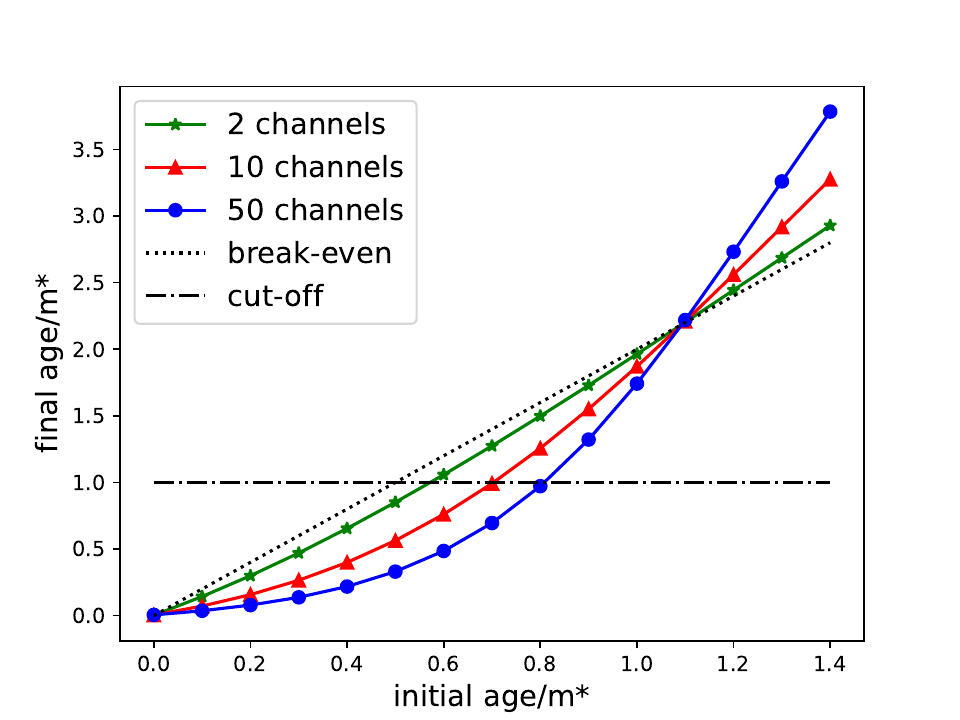}
    \caption{\textbf{Swap-then-distill versus distill-then-swap policy using banded approach to distillation.} Comparison of the \textsc{swap-distill} (left) and \textsc{distill-swap} (right) policies for a three-node chain with $n_{ch}$ links connecting the neighboring nodes, as described in \textit{Quasi-local multiplexing policies} of main text. Here we use the banded policy to perform distillation. 
    The \textsc{swap-distill} policy is not useful beyond initial age $m^\star/2$, whereas \textsc{distill-swap} can extend this threshold to $m^\star$.}
\label{fig:swdl_vs_dlsw}
\end{figure}

Results for the banding approach are shown in Fig.~\ref{fig:swdl_vs_dlsw}. We see that in the case of \textsc{swap-distill}, the cut-off of $m^\star$ is breached by the end-to-end link if the initial links are older than $\approx 0.5 m^\star$. This is so because distillation is not useful for links of age greater than the cut-off $m^\star$, and swapping adds the ages of links. In the case of \textsc{distill-swap}, the threshold is higher and equal to $m^\star$, because distillation occurs first. More importantly, we note that in both cases these thresholds are independent of the number of distillation channels used, indicated by the intersection of the different curves. On the other hand, for links of age below the threshold, using an increasing number of links, an arbitrarily high fidelity link can be obtained.

\begin{figure}
    \centering
    \includegraphics[width = 0.40\textwidth]{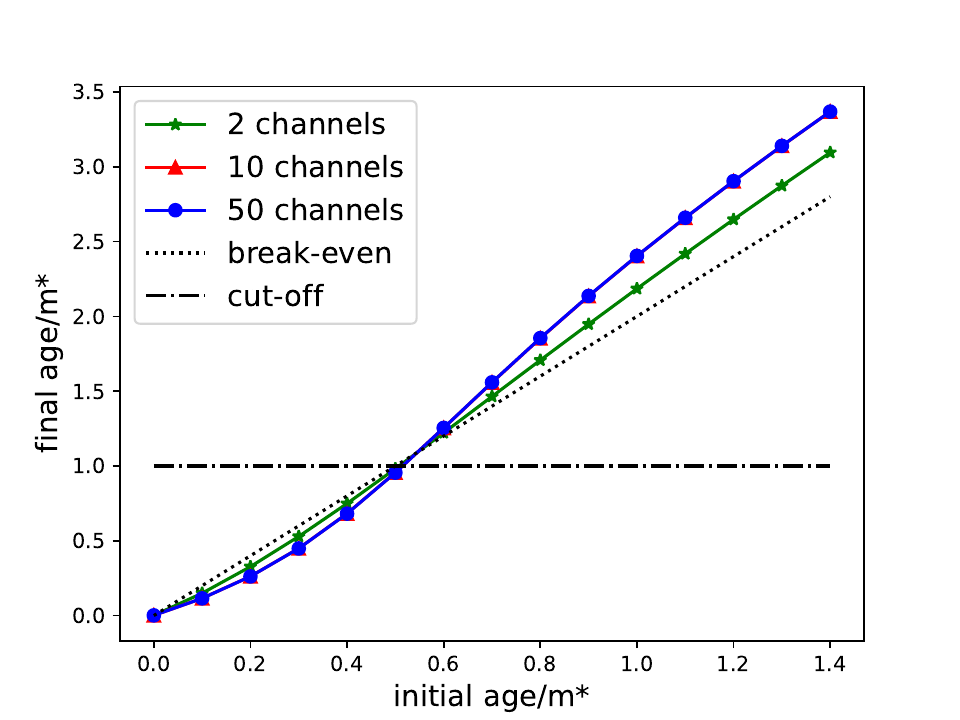}\quad
    \includegraphics[width = 0.40\textwidth]{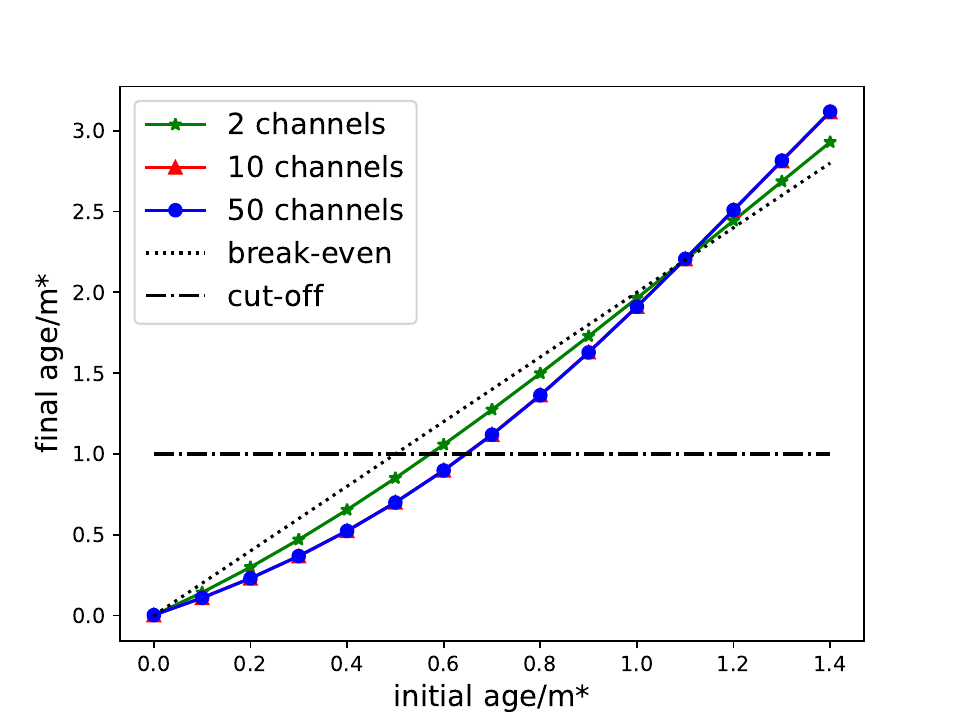}
    \caption{\textbf{Swap-then-distill versus distill-then-swap policy using pumping approach to distillation.}Comparison of the \textsc{swap-distill} (left) and \textsc{distill-swap} (right) policies for a three-node chain with $n_{ch}$ links connecting the neighboring nodes, as described above. Now, instead of the banding policy, we use the pumping policy for distillation. 
    As with banded purification, \textsc{swap-distill} is not useful beyond an initial age of $m^\star/2$, whereas \textsc{distill-swap} can extend this threshold, but now not quite up to $\approx m^\star$. This is a consequence of the fact that the pumping policy has a strict limit on the fidelity that can be achieved with an increasing number of channels, as shown in Eq.~\eqref{eq-pumping_fidelity_limit_appendix}.}
\label{fig:swdl_vs_dlsw_pumping}
\end{figure}

All of the analysis above was done considering an idealized situation in which links of the same age are available always for distillation. This is of course not necessarily true in practice. In practical settings, links will have to be used as and when they are produced. Therefore, it is important to look also at the previously mentioned pumping policy.

Results for the pumping policy are shown in Fig.~\ref{fig:swdl_vs_dlsw_pumping}. Although the trends remain very similar to Fig.~\ref{fig:swdl_vs_dlsw}, one important difference is noticeable immediately. When the pumping policy is used for distillation, the final age of the distilled link always saturates to a value approaching the limit in Eq.~\eqref{eq-pumping_fidelity_limit_appendix} as the number of channels increases. This can be seen by the convergence of the red (triangle) and blue (circle) curves for $n_{ch} = 10$ and $n_{ch} = 50$, respectively. This is in contrast to the observations in Fig.~\ref{fig:swdl_vs_dlsw}, in which we observe that increasing the number of channels progressively decreases the final age.

\section*{Supplementary Note 3: Distill-swap vs. swap-distill policy: analytical results}\label{sec:DS_vs_SD_example2}

Here we analytically compare the performance of the \textsc{distill-swap} and \textsc{swap-distill} policies in the case of three nodes. To simplify the analysis, we restrict ourselves to the cases with two and four channels, and consider all elementary links to be active simultaneously with the same fidelity. The general result for an arbitrary number of channels is difficult to calculate analytically given the large number of possible trajectories that can lead to the successful generation of at least one end-to-end link. Seeing in Supplementary Note 2 that the banded policy performs better than pumping, in this section we look exclusively at the banded policy. The performance of the policies is evaluated by the success probability of establishing at least one end-to-end entangled link and by the average fidelity of an end-to-end link.

\subsection*{Two-channel case}
Let $f_0$ be the initial fidelity of both the active links on each side given in terms of the initial age $m_0$ as $f_0=f(m_0)= \frac{1}{4}(1+3\e^{-m_0/m^\star})$.

Let us first consider the \textsc{distill-swap} policy. The probability of distillation being successful simultaneously on both sides followed by a successful entanglement swapping of the distilled links is given by
\begin{equation}
    p_{ds} = p_{sw}p_d^2 = p_{sw}((8f_0^2 - 4f_0 + 5)/9)^2,
    \label{eq-e2}
\end{equation}
where $p_d = (8f_0^2 - 4f_0 + 5)/9$ is the success probability of one distillation attempt, as given by Eq.~\eqref{eq-distill_p}. The fidelity and ages of the successfully distilled links on both sides is given by
\begin{align}
    f_d &= (1-2f_0+10f_0^2)/(5-4f_0+8f_0^2), \\
    m_d &= \lceil-m^\star\log((4f_d - 1)/3))\rceil. \label{eq-e3_4}
\end{align}
Finally, the age of virtual link produced by swapping is the sum of the ages of the elementary links consumed, i.e., $2m_d$. Therefore, the fidelity of the end-to-end link is given by:
\begin{equation}
    f_{ds} = \frac{1}{4}(1+3\e^{-2m_d/m^\star}).
    \label{eq-e5}
\end{equation}

Next, we consider the \textsc{swap-distill} policy. For this policy, two cases may arise depending on the number of successful entanglement swaps. Let $p_{sd}^{(1)}$, $f_{sd}^{(1)}$ and $p_{sd}^{(2)}$, $f_{sd}^{(2)}$ be the success probability and fidelity, respectively, of the two cases. The two cases are based on whether one or both of the entanglement swaps succeed. Then, the total probability of producing at least one end-to-end link is the sum of the two cases and is given by
\begin{equation}\label{eq-SD_vs_DS_two_channel_analytical_succ}
    p_{sd}=p_{sd}^{(1)}+p_{sd}^{(2)},
\end{equation}
and the weighted (expected) fidelity of the final distilled virtual link produced is given by
\begin{equation}\label{eq-SD_vs_DS_two_channel_analytical_fid}
    f_{sd}=\frac{p_{sd}^{(1)}f_{sd}^{(1)}+p_{sd}^{(2)}f_{sd}^{(2)}}{p_{sd}^{(1)}+p_{sd}^{(2)}}.
\end{equation}
Let us now determine analytical expressions for $p_{sd}^{(1)},f_{sd}^{(1)},p_{sd}^{(2)},f_{sd}^{(2)}$.
\begin{enumerate}
    \item Both entanglement swaps succeed. In this case, the probability that both entanglement swaps succeed, followed by a successful distillation attempt, is given by
    \begin{equation}
        p^{(2)}_{sd} = p_{sw}^2(8f_s^2 - 4f_s + 5)/9),
    \end{equation}
    where the ``$(2)$'' in the superscript of $p_{sd}^{(2)}$ indicates the number of successful swaps and $f_s$ is the fidelity of the virtual links produced after successful swapping, given by
    \begin{equation}\label{eq-f_s_appendix}
        f_s = \frac{1}{4}(1+3\e^{-\frac{2m_0}{m^\star}}).   
    \end{equation}
    The resulting fidelity, after distillation, is then 
    \begin{equation}
        f_{sd}^{(2)}=F_{\text{distill}}(f_s,f_s)=\frac{1-2f_s+10f_s^2}{5-4f_s+8f_s^2}.
    \end{equation}
    %\begin{equation}
    %    f_s = \frac{1}{4}\Bigg(1+3\exp\bigg(\frac{-2m_0}{m^\star}\bigg)\Bigg).
    %    \label{eq-e8}
    %\end{equation}
    \item Only one entanglement swap is successful. This can happen in two ways, therefore the probability of successfully producing one end-to-end link in this case is simply given by
    \begin{equation}
        p_{sd}^{(1)} = 2p_{sw}(1-p_{sw}),
        \label{eq-e9}
    \end{equation}
    because no distillation can be performed in this case, which also means that
    \begin{equation}
        f_{sd}^{(1)}=f_s,
    \end{equation}
    with $f_s$ as defined in Eq.~\eqref{eq-f_s_appendix}.
\end{enumerate}

\subsection*{Four-channel case}

Now, for the four-channel case, let us start with the \textsc{distill-swap} policy. Recall that under the banding policy, at most three rounds of distillation can be performed with four links. Based on this, there are three trajectories that can lead to one end-to-end link, and the corresponding success probabilities and fidelities are derived below. We combine them to obtain the total probability and weighted (expected) fidelity of an end-to-end link, in a manner analogous to \eqref{eq-SD_vs_DS_two_channel_analytical_succ} and \eqref{eq-SD_vs_DS_two_channel_analytical_fid}. Throughout these derivations, we let $f_0$ be the initial fidelity of all links, we let $p_d\equiv P_{\text{distill}}(f_0,f_0)$ and $f_d\equiv P_{\text{distill}}(f_0,f_0)$ be the success probability and fidelity, respectively, after the first round of distillation. Similarly, we let $p_{d_2}\equiv P_{\text{distill}}(f_d,f_d)$ and $f_{d_2}\equiv F_{\text{distill}}(f_d,f_d)$ be the success probability and fidelity, respectively, after the second round of distillation. The corresponding ages of the links after each round are denoted by $m_d\equiv f^{-1}(f_d)$ and $m_{d_2}\equiv f^{-1}(f_{d_2})$, where we recall the function $f(m)=\frac{1}{4}(1+3\exp(-m/m^{\star}))$ and its inverse $f^{-1}(F)=\ceil{-m^{\star}\log((4F-1)/3)}$.
\begin{enumerate}
    \item Only one pair of links is distilled successfully on each side, and the distilled link is successfully swapped. There are four possible ways this can occur, which means that the success probability and resulting fidelity in this case are given by
    \begin{align}
        p^{(1,1)}_{ds} &= p_{sw}(4p_d^2(1-p_d)^2), \\
        f^{(1,1)}_{ds} &= \frac{1}{4}\Bigg(1+3\exp\bigg(-\frac{2m_d}{m^\star}\bigg)\Bigg), \label{eq-e14}
    \end{align}
    where the ``$(1,1)$'' in the superscript denotes that only one distillation attempt succeeded on each side.
    
    \item All three distillation attempts succeed on one side and only one distillation attempt succeeds on the other side, followed by a successful entanglement swapping. (Note that the case of only two distillation successes on either side does not lead to an active link, and thus entanglement swapping to generate an end-to-end link is not possible.) This can happen in four ways - two ways to choose which side all three distillation will succeed and when all attempts succeed on one side, either the first or the second attempt could fail on the other side. Denoting the corresponding success probability and fidelity by $p_{ds}^{(3,1)}$ and $f_{ds}^{(3,1)}$, respectively, we have
    \begin{align}
        p^{(3,1)}_{ds} &= p_{sw}(4p_d^3(1-p_d)p_{d_2}), \\
        f^{(3,1)}_{ds} &= \frac{1}{4}\Bigg(1+3\exp\bigg(-\frac{m_d + m_{d_2}}{m^\star}\bigg)\Bigg). \label{eq-e16}
    \end{align}

    \item All three distillation attempts succeed on both sides, followed by successful entanglement swapping. In this case, denoting the success probability and fidelity by $p_{ds}^{(3,3)}$ and $f_{ds}^{(3,3)}$, respectively, we have
        \begin{align}
            p^{(3,3)}_{ds} &= p_{sw}(p_d^2p_{d_2})^2, \\
            f^{(3,3)}_{ds} &= \frac{1}{4}\Bigg(1+3\exp\bigg(-\frac{2m_{d_2}}{m^\star}\bigg)\Bigg). \label{eq-e20}
        \end{align}
\end{enumerate}
Therefore, the total success probability of producing at least one end-to-end link is given by
\begin{equation}
    p_{ds}=p_{ds}^{(1,1)}+p_{ds}^{(3,1)}+p_{ds}^{(3,3)},
\end{equation}
and the expected fidelity is given by
\begin{equation}
    f_{ds}=\frac{f_{ds}^{(1,1)}p_{ds}^{(1,1)}+f_{ds}^{(3,1)}p_{ds}^{(3,1)}+f_{ds}^{(3,3)}p_{ds}^{(3,3)}}{p_{ds}^{(1,1)}+p_{ds}^{(3,1)}+p_{ds}^{(3,3)}}.
\end{equation}

Next, we move to the \textsc{swap-distill} policy. In this case, six trajectories can lead to a successful end-to-end link generation. They are listed below.
\begin{enumerate}
    \item Only one entanglement swapping attempt succeeds, and thus no distillation is possible. There are $\binom{4}{1}=4$ possible pairings for the entanglement swapping, which means that the success probability $p_{sd}^{(1)}$ and the fidelity $f_{sd}^{(1)}$ are given by 
    \begin{align}
        p^{(1)}_{sd} &= 4p_{sw}(1-p_{sw})^3, \\
        f^{(1)}_{sd} &= f_s \equiv f(2m_0) = \frac{1}{4}\Bigg(1+3\exp\bigg(-\frac{2m_0}{m^\star}\bigg)\Bigg). \label{eq-e18}
    \end{align}
    
    \item Two entanglement swapping attempts succeed, and then one distillation attempt is made and succeeds. There are now $\binom{4}{2}=6$ possible ways to obtain two successful entanglement swaps, which means the success probability $p_{sd}^{(2)}$ and fidelity $f_{sd}^{(2)}$ in the case are given by
    \begin{align}
        p^{(2)}_{sd} &= 6p_{sw}^2(1-p_{sw})^2p_s, \\
        f^{(2)}_{sd} &= F_{\text{distill}}(f_s,f_s)=(1-2f_s+10f_s^2)/(5-4f_s+8f_s^2), \label{eq-e19}
    \end{align}
    where $p_s \equiv P_{\text{succ}}(f_s,f_s)=(8f_s^2 - 4f_s + 5)/9$ is the success probability of distilling two virtual links of age $2m_0$ (the fidelity is $f_s=f(2m_0)$, as defined in Eq.~\eqref{eq-e18}).
    
    \item If three entanglement swaps succeed (can happen in $\binom{4}{3}=4$ ways), then two distillation attempts need to be made in order to get one end-to-end link.  Two sub-cases arise
        \begin{enumerate}
            \item The first distillation attempt may fail leaving us with just one end-to-end link. 
            \item Both distillation attempts might succeed.
        \end{enumerate}
        Thus, the success probability and fidelity are a sum (and weighted sum) of these two sub-cases, and are given by:
        \begin{align}
            p^{(3)}_{sd} &= 4p_{sw}^3(1-p_{sw})(1-p_s) + 4p_{sw}^3(1-p_{sw} p_s p_{s2}), \\
            f^{(3)}_{sd} &= \frac{1}{p^{(3)}_{sd}}\Bigg(4p_{sw}^3(1-p_{sw})(1-p_s) f_s + 4p_{sw}^3(1-p_{sw} p_s p_{s2})\frac{1-(f_s + f^{(2)}_{sd})+10f_s f^{(2)}_{sd}}{5-2(f_s + f^{(2)}_{sd})+8f_sf^{(2)}_{sd}}\Bigg)
        \end{align}
        where $p_{s2} = (8f_sf^{(2)}_{sd} - 2(f_s + f^{(2)}_{sd}) + 5)/9$ is the distillation success probability for two links, one with age $2m_0$ and fidelity $f_s$ and another which has been obtained by distilling two links of age $2m_0$ and thus having fidelity $f^{(2)}_{sd}$.
    
    \item Case 4: All 4 entanglement swapping attempts succeed. Again, two sub-cases arise---one where all three subsequent distillation attempts succeed, and second where one of the two first rounds of distillation attempts fails (can happen in two ways), and there is no scope to distill further. Thus, the success probability and fidelity is a sum (and weighted sum) of these two sub-cases, and is given by:
    \begin{align}
        p^{(4)}_{sd} &= p_{sw}^4 p_s^2 p_{s3} + 2p_{sw}^4p_s(1-p_s), \\
        f^{(4)}_{sd} &= \frac{1}{p^{(4)}_{sd}}\Bigg(p_{sw}^4 p_s^2 p_{s3} (1-2f^{(2)}_{sd}+10*(f^{(2)}_{sd})^2)/(5-4f^{(2)}_{sd}+8(f^{(2)}_{sd})^2 + 2p_{sw}^4p_s(1-p_s)f^{(2)}_sd\Bigg),
    \end{align}
    where $p_{s3} = ((8f^{(2)}_{sd})^2 - 4f^{(2)}_{sd} + 5)/9)$
\end{enumerate}
Therefore, the total success probability of getting a end-to-end entangled link and the final fidelity of the link when \textsc{swap-distill} policy is used is given by:
\begin{eqnarray}
    p_{sd} &=& 4p_{sw}(1-p_{sw})^3 + 6p_{sw}^2(1-p_{sw})^2p_s + 4p_{sw}^3(1-p_{sw})(1-p_s) + 4p_{sw}^3(1-p_{sw} p_s p_{s2}) + p_{sw}^4 p_s^2 p_{s3} + 2p_{sw}^4p_s(1-p_s), \nonumber\\
    f_{sd} &=& \frac{1}{p_{sd}}\Bigg(p_{sw}(1-p_{sw})^3f_s + 6p_{sw}^2(1-p_{sw})^2p_s f^{(2)}_{sd} + (4p_{sw}^3(1-p_{sw})(1-p_s) f_s \nonumber\\ &+& 4p_{sw}^3(1-p_{sw} p_s p_{s2})\frac{1-(f_s + f^{(2)}_{sd})+10f_s f^{(2)}_{sd}}{5-2(f_s + f^{(2)}_{sd})+8f_sf^{(2)}_{sd}} + p_{sw}^4 p_s^2 p_{s3} \frac{1-2f^{(2)}_{sd}+10(f^{(2)}_{sd})^2}{5-4f^{(2)}_{sd}+8(f^{(2)}_{sd})^2} + 2p_{sw}^4p_s(1-p_s)f^{(2)}_{sd}\Bigg)
\end{eqnarray}

\begin{figure}
    \centering
    \includegraphics[width = 0.40\textwidth]{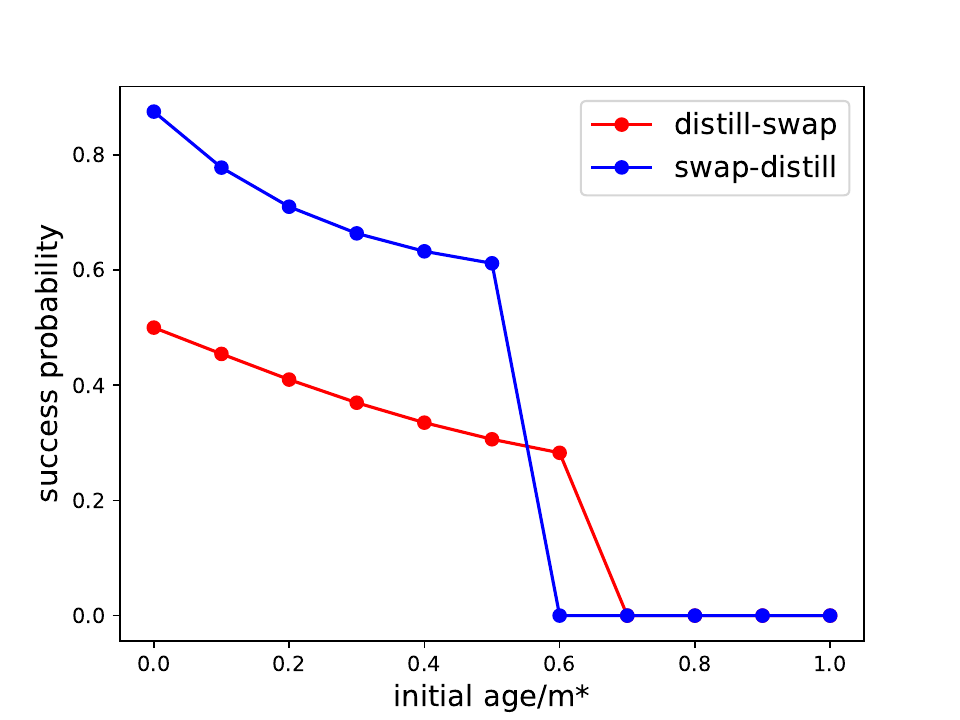}\quad
    \includegraphics[width = 0.40\textwidth]{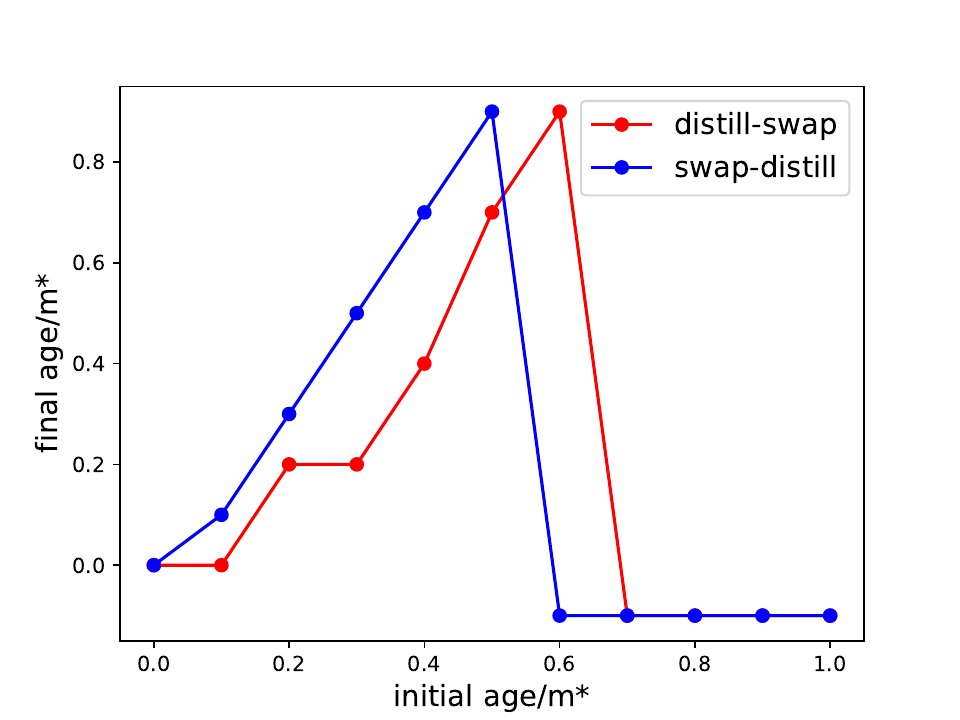}
    \caption{\textbf{Analytical results for \textsc{distill-swap} and \textsc{swap-distill} policies.} Comparison of the \textsc{distill-swap} and \textsc{swap-distill} policies for a three-node chain with four channels connecting the neighboring nodes, as described in Supplementary Note 3. We look at the probability of successfully producing at least one end-to-end entangled link (left) and the average age of that link (when successfully produced) (right). For distillation, we use the banded policy. We see again that the \textsc{swap-distill} policy is not useful beyond and initial age of $m^\star/2$, whereas \textsc{distill-swap} can extend this threshold up to $m^\star$ by using an increasing number of channels. $p_{sw} = 50\%$ for these plots.}
\label{fig:swdl_vs_dlsw_1link}
\end{figure}

In Fig.~\ref{fig:swdl_vs_dlsw_1link}, we show the results for four channels. It is clear from Fig.~\ref{fig:swdl_vs_dlsw_1link} that \textsc{swap-distill} has a higher probability of success compared to \textsc{distill-swap} when the initial age of the links is less than $m^\star/2$, beyond which it fails to succeed. Unlike distillation, success probability of entanglement swapping is independent of the quality of the used links. This is the primary reason for an asymmetry in the success rates of \textsc{swap-distill} and \textsc{distill-swap} policies. The probability of at least one entanglement swap being successful increases exponentially with the number of attempts available, whereas the probability of all distillations being successful falls with an increasing number of attempts, irrespective of the ages of the links involved. Thus, it is clear that if the aim is to generate at least one end-to-end link as quickly as possible, i.e., if the aim is to minimize waiting time, the \textsc{swap-distill} is the policy of choice. At the same time, \textsc{distill-swap} succeeds until a larger initial age threshold; thus, if the elementary links are of low fidelity to begin with, distilling first might be judicious. Also, in terms of fidelity enhancement (age reduction), \textsc{distill-swap} performs better. These conclusions are also consistent with the results in Fig.~\ref{fig:swdl_vs_dlsw} and \ref{fig:swdl_vs_dlsw_pumping}.

\section*{Supplementary Note 4: Other policies for multiplexing based on \textsc{swap-asap}} \label{sec:appendix6}

Here we explore multiplexing based \textsc{swap-asap} policies other than the two paradigmatic versions presented in the main text, namely FN \textsc{swap-asap}, which prioritizes forming the longest possible links, and SN \textsc{swap-asap}, which prioritizes forming links with high fidelity. A natural question to ask is how well do policies that are a mixture of these approaches perform. As we have seen that FN outperforms SN in terms of waiting time, but SN outperforms FN in terms of fidelity, the above question is also motivated by the aim to optimize some linear combination of waiting time and fidelity. SN and FN can be mixed in two ways, yielding two different families of policies:
\begin{enumerate}
    \item \emph{Mixed-weight} \textsc{swap-asap}: Linear combination of link length and link age can be used to assign the ranks of links for entanglement swapping. Let $R$ be the metric to decide the rank of a link. Then, a continuous parameter $a \in [0,1]$ could be used to define a family of policies as $R = a|i-j| + (1-a)m$, where $i,j$ are the end-nodes of a link and $m$ is its age. The limiting cases $a=0$ and $a=1$ correspond to the SN and FN \textsc{swap-asap} policies, respectively.
    
    \item \emph{Random-priority} \textsc{swap-asap}: Keeping the weights to be either the ages of links or their lengths, we could randomly change the priority from SN to FN between different time steps. This would also give a family of policies parameterized by a random number $r \in [0,1]$, which would determine the relative frequencies of choosing the SN and FN policies. Again, we choose $r$ such that $r=0$ and $r=1$ would denote the limiting cases, SN and FN, respectively.
\end{enumerate}
Having defined these two families of policies, it is interesting to look at their performance in terms of average waiting time and age of an end-to-end link, as a function of the parameter $a$. This is shown in Fig.~\ref{fig:mixed_policy}.

\begin{figure}
        \centering
        \includegraphics[width=0.4\textwidth]{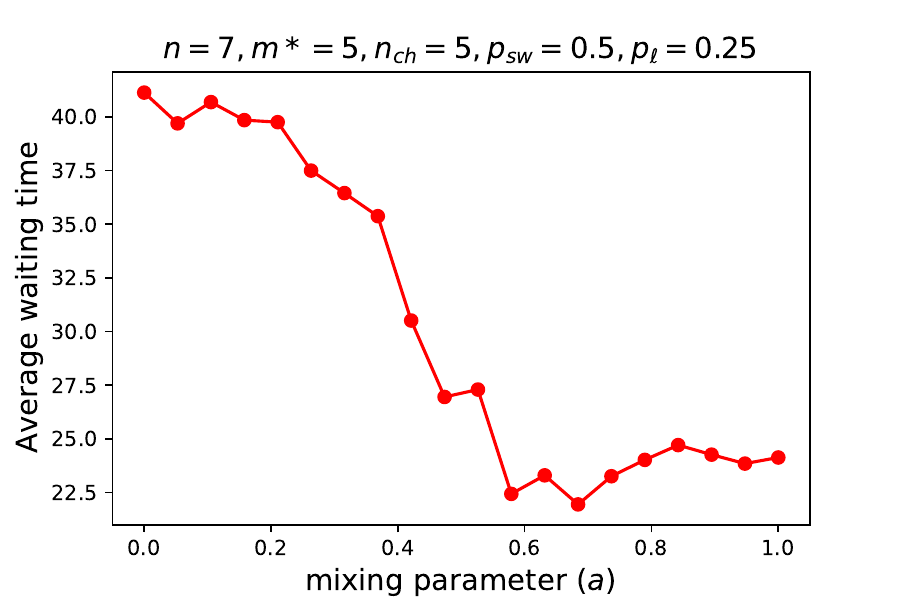}\quad
        \includegraphics[width=0.4\textwidth]{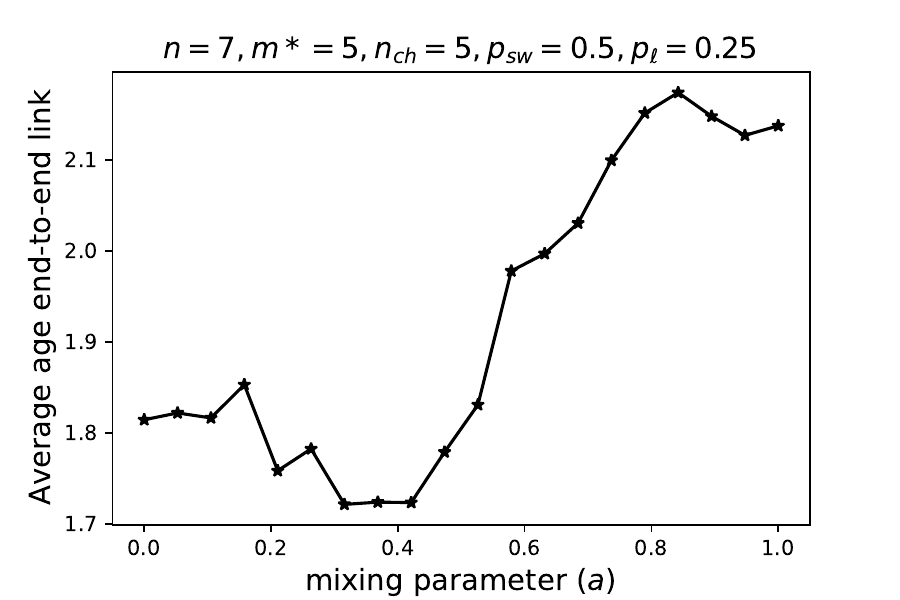}\\[1ex]
        \includegraphics[width=0.4\textwidth]{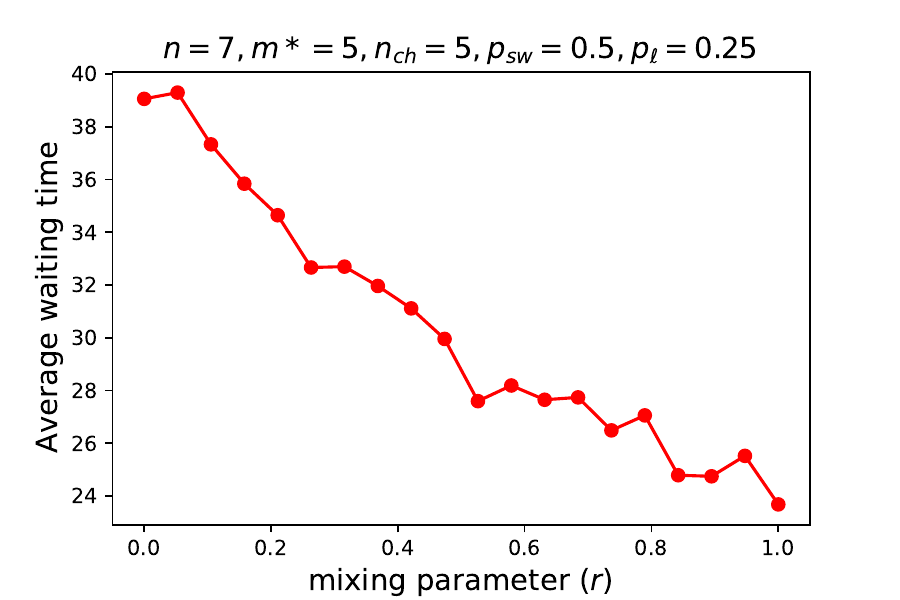}\quad
        \includegraphics[width=0.4\textwidth]{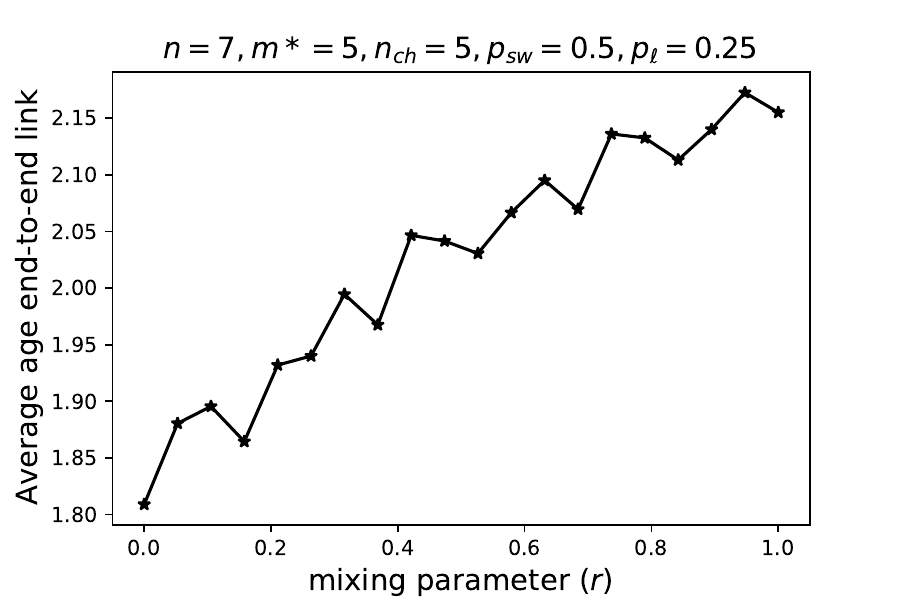}
        \caption{\textbf{Performance of other multiplexing policies.} (Top) Performance of mixed-weight \textsc{swap-asap} multiplexing policies in terms of average waiting time and age of an end-to-end link, as a function of the mixing parameter $a$. FN ($a=1$) is nearly optimal for waiting time reduction, but a nearly equal mixture of SN and FN ($a=0.4$) seems to minimize the average age of an end-to-end link. (Bottom) The figures of merit for the random-priority \textsc{swap-asap} policy show a much more monotonic trend. FN ($r=1$) is optimal for waiting time and SN ($r=0$) is optimal for average age.}
        \label{fig:mixed_policy}
\end{figure}

Here, we also want to consider the question of optimality of SN and FN \textsc{swap-asap} in terms of reducing the average age and waiting time for an end-to-end link. Of course, a full optimization could be done by explicitly defining the figure of merit as a cost function and then using tools such as reinforcement learning or linear programs to find the optimal entanglement distribution policy within our MDP framework. As was shown in previous works like Refs.~\cite{haldar2023fastreliable, SvL21}, the number of MDP states grows exponentially with the number of links and also with the cutoff. Thus, finding the optimal policy explicitly in the multiplexing case is quite challenging, and we leave that as an aim for future work. Here, we ask a limited question. Let us consider a policy that intends to match links for entanglement swapping in such a way that the total length of the virtual links (if entanglement swaps were to be successful) is maximum. We define this policy as \textit{FN optimal}, or FN-OPT policy. In other words, FN-OPT tries all pairings of links for entanglement swapping, calculates the sum of the lengths of the anticipated virtual links, and then chooses the set of pairings that maximizes this sum. Such a pairing could be different from the pairings provided by the usual FN policy, in which the longest link on each side is paired with one another, followed by the second longest, and so on. For example, if we have two links of ages $7$ and $3$ on each side on a node, then two entanglement swaps can be performed in total. Let us say that the cutoff is 10. The standard FN pairing would be $(7,7)$ and $(3,3)$. However, this would mean that the first entanglement swap will never be attempted, because we know that it will fail (as described earlier, the swapped link has age $m_1+m_2$, and we only perform swaps if $m_1+m_2<m^\star$). Therefore, the sum of lengths of anticipated virtual links in this case will be $3+3 = 6$. FN-OPT, on the other hand, would choose the pairings $(7,3)$ and $(7,3)$, because the sum of anticipated lengths would be $10 + 10 = 20$. The intuition for considering FN-OPT is the same as the standard FN policy, i.e., prioritizing forming long links reduces the average waiting time. In a similar way, an \textit{SN-OPT policy} can be defined, which would minimize the sum of the ages of the anticipated links. 

\begin{figure}
    \centering
    \includegraphics[width=0.4\textwidth]{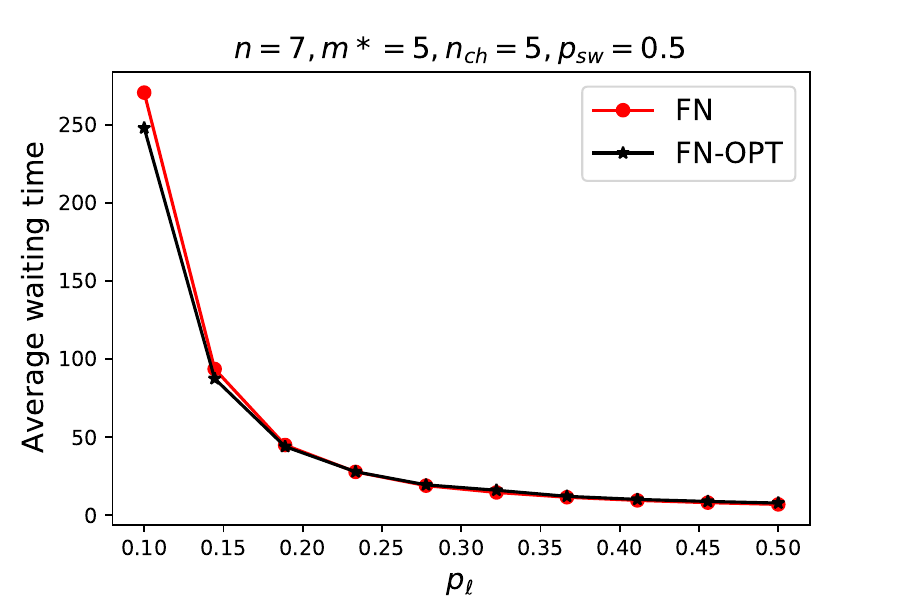}\quad
    \includegraphics[width=0.4\textwidth]{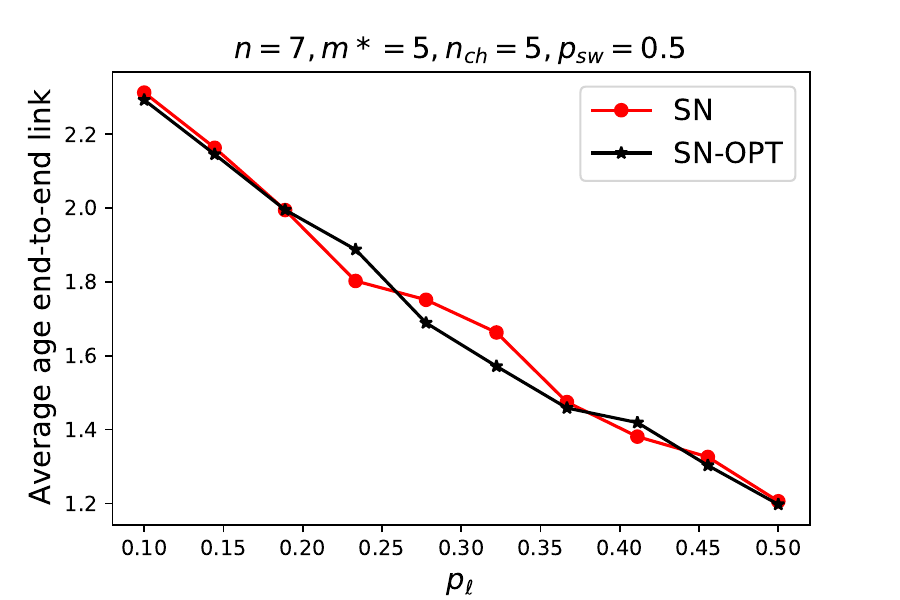}
    \caption{\textbf{Performance of our multiplexing policies compared to their optimal versions.} We compare the performance of FN to FN-OPT and SN to SN-OPT in terms of average waiting time and average age of an end-to-end link. The much simpler and standard FN and SN policies are nearly optimal in both cases.}
    \label{fig:optimal_policy}
\end{figure}

In Fig.~\ref{fig:optimal_policy}, we compare the performance of FN to FN-OPT and SN to SN-OPT. We can see that the much simpler and standard FN and SN approaches are nearly optimal, in the sense of optimality described above. It is also important to note that the optimal pairings under FN-OPT and SN-OPT will become increasingly complicated to establish with an increasing number of channels, whereas finding the pairings in SN and FN will still remain equally simple.

\end{document}